\documentclass[12pt,english]{smfart}
\usepackage[table]{xcolor}

\usepackage{mathrsfs}
\usepackage{amsthm}
\usepackage{amscd}
\usepackage{amstext}
\usepackage{amssymb}
\usepackage{graphicx}
\usepackage{setspace}
\usepackage{esint}
\usepackage{amsfonts}
\usepackage{amssymb}
\usepackage[all]{xy}
\usepackage{array}
\usepackage{lmodern}
\usepackage[OT2,T1]{fontenc}

\usepackage{linearb}

\usepackage{hyperref}
\makeatletter

\numberwithin{equation}{subsection}
\numberwithin{figure}{subsection}
\theoremstyle{plain}
\newtheorem{thm}{\protect\theoremname}[subsection]
  \theoremstyle{plain}
  \newtheorem{prop}[thm]{\protect\propositionname}
  \theoremstyle{definition}
  \newtheorem{defn}[thm]{\protect\definitionname}
  \theoremstyle{plain}
  \newtheorem{cor}[thm]{\protect\corollaryname}
  \theoremstyle{remark}
  \newtheorem{rem}[thm]{\protect\remarkname}
  \theoremstyle{plain}
 \newtheorem{lem}[thm]{\protect\lemmaname}

\def\Z{\mathfrak{Z}}
\def\cX{\mathcal{X}}
\def\cR{\mathcal{R}}
\def\cE{\mathcal{E}}
\def\cH{\mathcal{H}}
\def\QQ{\mathbb{Q}}
\def\RR{\mathbb{R}}
\def\CC{\mathbb{C}}

\def\ZZ{\mathbb{Z}}
\def\PP{\mathbb{P}}
\def\DD{\mathbb{D}}
\def\GG{\mathbb{G}}
\def\vf{\varphi}
\def\ba{\textlinb{\Bqe}}
\def\su{\circleddash}

\def\L{\mathcal{L}}
\def\X{\mathcal{X}}
\def\E{\mathcal{E}}
\def\sur{\twoheadrightarrow}

\newcommand{\ve}{{\varepsilon}}
\newcommand{\inj}{\hookrightarrow}
\newcommand{\sE}{{\mathcal E}}
\newcommand{\sS}{{\mathcal S}}
\newcommand{\sZ}{{\mathcal Z}}

\makeatother

\usepackage{babel}
  \providecommand{\corollaryname}{Corollary}
  \providecommand{\definitionname}{Definition}
  \providecommand{\propositionname}{Proposition}
  \providecommand{\remarkname}{Remark}
\providecommand{\theoremname}{Theorem}
\providecommand{\lemmaname}{Lemma}

\newcount\hh
\newcount\mm
\mm=\time
\hh=\time
\divide\hh by 60
\divide\mm by 60
\multiply\mm by 60
\mm=-\mm
\advance\mm by \time
\def\thistime{\number\hh:\ifnum\mm<10{}0\fi\number\mm}
\def\nn{\nonumber}

\def\Li#1(#2){\textrm{Li}_{#1}\left(#2\right)}
\def\cLi_#1(#2){\mathcal{L}_{#1}\left(#2\right)}
\def\bLi_#1(#2){\mathbf{L}_{#1}\left(#2\right)}

\def\ZZ{{\mathbb Z}}

\def\IR{{\mathbb R}}
\def\RR{{\mathbb R}}

\def\cE{\mathcal{E}}

\def\cR{\mathcal{R}}

\def\Ree{\Re\textrm{e}}
\def\Imm{\Im\textrm{m}}

\title[A Feynman integral via higher normal functions]{\bf 
A Feynman integral via higher normal functions
}

\author[S. Bloch]{Spencer Bloch}
\address{5765 S. Blackstone Ave., Chicago, IL 60637, USA}
\email{spencer\_bloch@yahoo.com}

\author[M. Kerr]{Matt Kerr}
\address{Department of Mathematics, Campus Box 1146\\
Washington University in St. Louis\\
St. Louis, MO, 63130, USA}
\email{matkerr@math.wustl.edu}

\author[P. Vanhove]{Pierre Vanhove}
 \address{
 Institut des Hautes \'Etudes Scientifiques\\
 Le Bois-Marie, 35 route de Chartres\\
 F-91440 Bures-sur-Yvette, France\hfill\break
Institut de physique th\'eorique,\\
 Universit\'e Paris Saclay, CEA, CNRS, F-91191 Gif-sur-Yvette}
\email{pierre.vanhove@cea.fr}

\thanks{IPHT-T/14/015, IHES/P/14/06}
\date{\today}

\begin{document}

 \begin{abstract}
%{\bf Draft version \today\ at \thistime}

We study the Feynman integral for the three-banana  graph defined as the scalar two-point
self-energy at three-loop order. The Feynman integral is evaluated for all
identical internal masses in two space-time dimensions. Two calculations are given
for the Feynman integral; one based on an interpretation of the
integral as an inhomogeneous solution of a classical Picard-Fuchs
differential equation, and the other using arithmetic algebraic
geometry, motivic cohomology, and Eisenstein series. Both methods use
the rather special fact that the Feynman integral  is a family of
regulator periods associated to a family of $K3$ surfaces. 
We show that the integral is given by a sum of elliptic
trilogarithms evaluated at  sixth roots   of unity. This elliptic trilogarithm value is related to the
regulator of a class in the motivic cohomology of the  $K3$
family. We prove a conjecture by David Broadhurst that at a special kinematical point the Feynman
integral is given by a critical value of the Hasse-Weil $L$-function
of the $K3$ surface.  This result is shown to be a particular case of Deligne's
conjectures relating values of $L$-functions inside the critical strip
to periods.

\end{abstract}
\maketitle
\newpage
\tableofcontents
\newpage

\section{Introduction}
\numberwithin{equation}{section}

The computation of scattering amplitudes in quantum field theory
requires the evaluation of Feynman integrals. This is a non-trivial task for which many techniques have been
developed by physicists over the years (cf. the reviews~\cite{Bern:1996je,Britto:2010xq,Ellis:2011cr,Elvang:2013cua}.)
Feynman integrals  are multivalued functions of the physical parameters, given by the
external momenta and internal masses.
Differentiating with respect to the  physical
parameters leads to a first
order system of differential equations as in
e.g.~\cite{Henn:2013pwa,Caron-Huot:2014lda} or to higher order
differential equations as in
e.g.~\cite{Laporta:2004rb,MullerStach:2011ru,MullerStach:2012mp,PVstringmath,Adams:2013nia,Adams:2014vja}.

\smallskip

The Feynman integral associated to a graph $\Gamma$ with $n$ edges (propagators) is an integral over the positive simplex $\Delta_n:=\{[x_1:\cdots :x_n]\in \mathbb P^{n-1}(\mathbb R)\,|\, x_i\geq0\}$ in projective $(n-1)$-space of a meromorphic differential $(n-1)$-form:
\begin{equation}\label{e:Igamma}
  I_\Gamma =\int_{\Delta_n} \, \Omega_\Gamma  \,.
\end{equation}
The form $\Omega_{\Gamma}$ depends on the physical parameters -- that is, the external momenta and internal masses attached to the graph --  and is expressed in terms of the first and second
 Symanzik polynomial~\cite{Itzykson:1980rh}.  The variables $x_i$ are the Schwinger proper times
 indexed by edges (propagators).

For the algebro-geometric  approach  of~\cite{Bloch:2005bh},  the Feynman integral
$I_\Gamma$ is a period of the mixed Hodge structure on the relative cohomology group 
$H^{n-1}(\mathbb P^{n-1} \backslash X_\Gamma, B\backslash (B \cap
X_\Gamma)),$ where $X_\Gamma$ is the graph hypersurface  defined by the
poles of $\Omega_\Gamma$ and $B$ is a blow-up of the simplex
$\Delta_n$.  Varying the physical parameters leads to a variation of the
Hodge structure.  As a result, the Feynman integral satisfies a set of first
order differential equations under the action of the the Gauss-Manin
connection~\cite{Griffiths}, leading to an inhomogeneous Picard-Fuchs equation.
The inhomogeneous term has its origin in the extension of mixed Hodge structure associated with Feynman graphs.
  The dependence on external momenta means that we have a family
of extensions, also known as a {\it normal function} from the work of
Poincar\'e~\cite{Poincare} and Griffiths~\cite{Griffiths2}.

This point of view enables us to bring to bear a number
of techniques including Picard-Fuchs differential equations, motivic
cohomology and regulators, Eisenstein series, and Hodge structures,
for the analysis of the properties of Feynman integrals.

\medskip

The main topic of this paper is the evaluation of the Feynman integral
for  the three-banana graph
\begin{equation}
I_\ba(t):= \int_{x_1,x_2,x_3\geq0} \,{1\over (1+\sum_{i=1}^3
  x_i)(1+\sum_{i=1}^3 x_i^{-1})-t} \,\prod_{i=1}^3 {dx_i\over x_i}\,.    
\end{equation}
The associated graph hypersurface $X_\ba(t):=\{(1+\sum_{i=1}^3
  x_i)(1+\sum_{i=1}^3 x_i^{-1})-t=0\}$ leads to a family of 
$K3$ surfaces with (generic) Picard number 19,
over the modular curve
$\mathbb{P}^1\setminus \{0,4,16,\infty\} \cong Y_1(6)^{+3}$. 
It is closely related to the family of elliptic curves over $Y_1(6)$, which was studied in~\cite{Bloch:2013tra} in connection with the Feynman integral arising from the sunset (two-loop banana) graph.

We prove in theorems~\ref{thm:Li3elliptic} and~\ref{thm:main} that the
Feynman integral evaluates to the product of a period $\varpi_1(\tau)$ of the
$K3$ surface and an Eichler integral of an Eisenstein series.  Explicitly, we have
\begin{equation}\label{e:1.3}
  I_\ba(t) = \varpi_1(\tau) \, \left(
  \sum_{n\geq1} {\psi(n)\over n^3} {q^n\over 1-q^n}-4(\log q)^3+16\zeta(3)\right) ,
\end{equation}
where $q=\exp(2\pi i\tau)$, $\psi(n)$ is a mod-6 character given in eq.~\eqref{e:psiDef},
and $t$ is related to $\tau$ by the Hauptmodul \eqref{e:tVerrill} for $\Gamma_1(6)^{+3}$.

 Remarkably, the Eichler integral factor can be 
 expressed as a combination of the Beilinson-Levin elliptic
 trilogarithms~\cite{BL1,L,ZagierElliptic} 
\begin{multline}\label{e:1.4}
  I_\ba(t)= \varpi_1(\tau)\, \Big(40\pi^2\log
    q+ 24\mathcal Li_3(\tau, \zeta_6)+ 21\,\mathcal Li_3(\tau ,\zeta_6^2)
\cr+ 8\mathcal Li_3(\tau,\zeta_6^3)+7  \mathcal Li_3(\tau,1)
\Big)
\end{multline}
where $\zeta_6:=\exp(i\pi/3)$ is the same sixth root of unity that
enters the expression of the sunset integral studied in~\cite{Bloch:2013tra}.  

It turns out that the three-banana integral is associated to
a {\it generalized normal function} arising from a family of
``higher'' algebraic cycles or motivic cohomology classes
\cite{KerrLewis,DoranKerr}. The passage from classical normal functions associated to families of cycles to normal functions
associated to motivic classes suggests interesting new links between
mathematics and physics (op.cit.).  Actually motivic normal functions can, in many cases, be associated with multiple-valued
holomorphic functions which arise as amplitudes as in this work or in
the context of open mirror symmetry as in~\cite{MorrisonWalcher} for instance.

\medskip

The plan of the paper is the following. In section~\ref{sec:3banana} we
derive the inhomogeneous Picard-Fuchs  equation satisfied by the
three-banana integral.  The solution of the differential equation in terms
of the elliptic trilogarithm is given
in theorem~\ref{thm:Li3elliptic}.  In section~\ref{sec Family of K3s}
we give a construction of the family of $K3$ surfaces associated
with the three-banana graph.

In section~\ref{sec HNF}  we show that the
three-banana integral $I_\ba(t)$ is an higher normal function, originating from a family of elements in $K_3(K3's)$ (a charming sort of mathematical eponym). Specifically, we show that the Milnor symbols
$\{-x_1,-x_2,-x_3\}\in K_3^M\left(\mathbb{C}(X_{\ba}(t))\right)$
extend to classes $\Xi_t \in H^3_M(X_{\ba}(t),\QQ(3)).$  We construct
a family of closed 2-currents $\tilde{R}_t$ representing the Abel-Jacobi classes $AJ(\Xi_t)\in H^2(X_{\ba}(t),\CC/\QQ(3))$,
and a family of holomorphic forms $\tilde{\omega}_t\in \Omega^2(X_{\ba}(t))$, such that
$$I_{\ba}(t)=\int_{X_{\ba}(t)}\tilde{R}_t \wedge \tilde{\omega}_t$$ (Theorem \ref{thm Feynman =00003D HNF}).    This has immediate consequences, including a conceptual proof of the inhomogeneous Picard-Fuchs equation for $I_{\ba}(t)$ (Corollary \ref{Cor PFE}).

In section~\ref{sec Eis symb} we pull the higher cycle $\Xi_t$  back from the family of $K3$ surfaces
to a modular Kuga 3-fold, where we are able to recognize it as an \emph{Eisenstein symbol} in the sense of Beilinson.
Applying a general computation (Theorem \ref{prop modular hnf}ff) of higher normal functions associated to Beilinson's cycles, gives a ``motivic'' proof (Theorem \ref{thm:main}) that the three-banana integral $I_{\ba}(t)$ takes the form claimed in \eqref{e:1.3}-\eqref{e:1.4}.  In section~\ref{sec:hodge} we give the abstract Hodge-theoretic
formulation of the Feynman integral in our case. 

Finally, in sections~\ref{sec:tzero} and theorem~\ref{thm hnf sv0} we
show that the integral at $t=0$ takes the value $I_\ba(0)=7\zeta(3)$
recovering at result of~\cite{Bailey:2008ib,BroadhurstLetter,BroadhurstProc}.
And in sections~\ref{sec:tone} and~\ref{sec:special-fiber-at}  we evaluate the three-banana
at the special value $t=1$. (The results in section \ref{subsec hnf analysis} again make crucial use of Theorem \ref{thm Feynman =00003D HNF}.)  We show the regulator to be
trivial, which means that the Feynman integral is actually a
classical rational period of the $K3$ up to a factor of $12\pi i/\sqrt{-15}$. A
conjecture of Deligne then relates the Feynman integral to the critical value of the
Hasse-Weil $L$-function of the $K3$ at $s=2$. This proves  a result first obtained numerically
by Broadhurst in~\cite{BroadhurstLetter,BroadhurstProc} up to a
rational coefficient. 

\numberwithin{equation}{subsection}

%%%%%%%%%%%%%%%%%%%%%%%%%%%%%%%%%%%%%%%%%%%%%%%%%%%%%%%%%%%%%%%%%%

%-------------------------------------------------------------------------
\section*{Acknowledgements}
We thank  A. Clingher and C. Doran for helpful discussions. We thank
David Broadhurst for many helpful comments and encouragements. 
MK thanks the IH\'ES for support and good working conditions while part of this paper was written.
We acknowledge support from the ANR grant   reference QST ANR 12 BS05 003
01, and the PICS  6076, and partial support from NSF Grant DMS-1068974.

%%%%%%%%%%%%%%%%%%%%%%%%%%%%%%%%%%%%%%%%%%%%%%%%%%%%%%%%%%%%%%%%%%%%%%
\section{The three-banana Feynman integral}\label{sec:3banana}

\subsection{The integral}
\label{sec:integral}

\begin{figure}[ht]
  \centering
  \includegraphics[width=8cm]{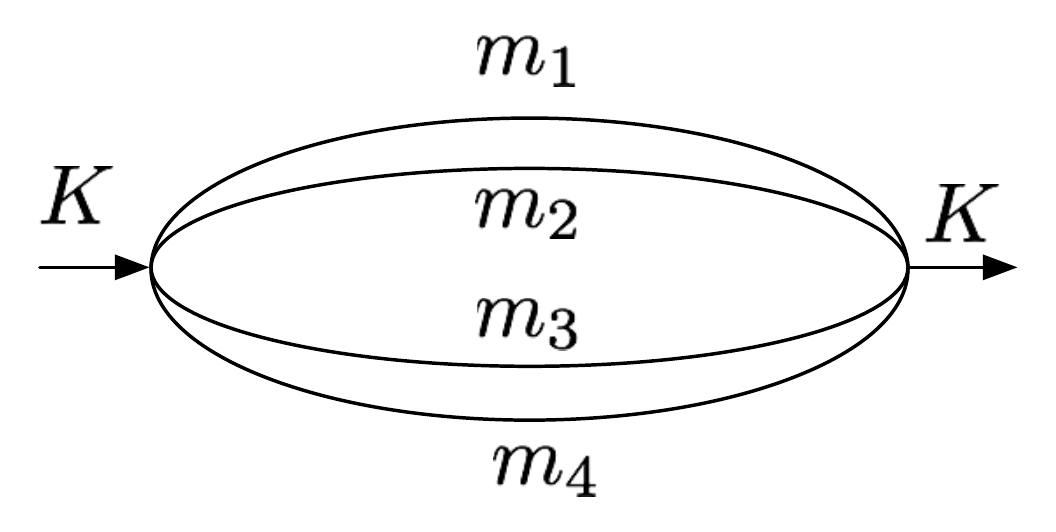}
  \caption{The three-loop three-banana Feynman graph. $K$ is the external momentum in $\IR^{2}$  and $m_i\geq0$ with $i=1,\dots,4$ are internal masses.}
  \label{fig:3banana}
\end{figure}

We look at the three-loop banana graph in  two space-time dimensions associated with
the Feynman graph in figure~\ref{fig:3banana}

\begin{equation}
  I_{\ba}(m_1,m_2,m_3,m_4;K):= \int_{\IR^8} {\delta(\sum_{i=1}^4 \ell_i+K)\prod_{i=1}^4d^2\ell_i  \over
    \prod_{i=1}^4 (\ell_i^2+m_i^2)}\,.
\end{equation}
Setting $t=K^2$, this integral can be equivalently represented as (see for instance~\cite[section~8]{PVstringmath})
\begin{equation}\label{e:I3masses}
  I_{\ba}(m_i; t)= \int_{x_i\geq0} \, {1\over    (m_4^2+\sum_{i=1}^3 m_i^2x_i)(1+\sum_{i=1}^3 x_i^{-1})-t }  \, \prod_{i=1}^3 {dx_i\over x_i}
\end{equation}

\begin{thm}\label{thm:BesselRep}
  The integral  $I_{\ba}(m_i;t)$ defined in eq~\eqref{e:I3masses} has
  the following integral representation for $t<(\sum_{i=1}^4 m_i)^2$
\begin{equation}\label{e:I3bessel}
  I_{\ba}(m_i;t)= 2^3 \int_0^\infty x \, I_0(\sqrt t x) \prod_{i=1}^4 K_0(m_ix)
  \, dx\,.  
\end{equation}
\end{thm}
The  Bessel functions $K_0, I_0$ are defined  by
\begin{equation}\label{K0Bessel}
K_0(2  \sqrt{ab}):= \frac12\int_0^\infty e^{-ax-{b\over x}}  {dx\over x};\qquad \textrm{for}~a,b>0 \,,
\end{equation}
and
\begin{equation}\label{I0Bessel}
 I_0(x):=  \sum_{k\geq0}  \left(x\over2\right)^{2k}\,
  {1\over\Gamma(k+1)^2}\,.
\end{equation}

 For the all equal mass case this Bessel representation has already
been given in~\cite{Bailey:2008ib,BroadhurstProc}.

\begin{proof}
For $t<(\sum_{i=1}^4 m_i)^3$  we can perform the series expansion 
\begin{equation}\label{e:I3series}
  I_{\ba}(m_i;t) = \sum_{k\geq 0}  t^k I_k
\end{equation}
with
\begin{equation}\label{e:defI3k}
I_k:=\int_{x_i\geq0}   {1\over  (m_4^2+\sum_{i=1}^3 m_i^2x_i)^{k+1}(1+\sum_{i=1}^3 x_i^{-1})^{k+1}}  \, \prod_{i=1}^3
  {dx_i  \over x_i}
\end{equation}
Exponentiating the denominators using $\int_0^\infty dx x^k \exp(-a
x)= \Gamma(k+1)/a^{k+1}$ for $a>0$ we have 

\begin{equation}
   I_k={1\over \Gamma(k+1)^2} \int_{x_i\geq0}\int_{u,v\geq0} e^{-u (1+\sum_{i=1}^3 x_i^{-1}) - v
  (m_4^2+\sum_{i=1}^3 m_i^2 x_i)}\,  {du dv\over (uv)^{-k}}\prod_{i=1}^3
  {dx_i  \over x_i}\,.
\end{equation}

Using the definition in~\eqref{K0Bessel} the integral over each $x_i$
leads to a $K_0(x)$ Bessel function,  therefore
\begin{equation}
  I_k =   {2^3\over \Gamma(k+1)^2} \,  \int_{u,v\geq0} e^{-u - v  m_4^2}\, \prod_{i=1}^3 K_0(2\sqrt{uv} m_i)\, {du dv\over (uv)^{-k}}\,.
\end{equation}
Changing variables $(u,v)\to ( x=2\sqrt{uv},v)$ then
\begin{eqnarray}\label{e:I3kresult}
  I_k &= &  {2^4\over \Gamma(k+1)^2} \,  \int_{v,x\geq0} e^{-{x^2\over
      4v} - v  m_4^2}\, \prod_{i=1}^3 K_0(2\sqrt{uv} m_i)\, \left(x\over2\right)^{2k+2}\, {dx
    dv\over xv}\cr
&=&  {2^5\over \Gamma(k+1)^2} \,  \int_0^{+\infty} \prod_{i=1}^4 K_0(m_i x)\, \left(x\over2\right)^{2k+2}\, {dx
    \over x}
\end{eqnarray}
Now we can perform the summation over $k$ using the  series expansion
of the Bessel function $I_0(\sqrt t \,x)$ given in~\eqref{I0Bessel} to
conclude the proof.
\end{proof}

For the all equal masses case $m_1=m_2=m_3=m_4=1$ we have
\begin{eqnarray}\label{e:I3}
  I_{\ba}(t)&:=&\int_{x_i\geq0} \, {1\over    (1+\sum_{i=1}^3 x_i)(1+\sum_{i=1}^3 x_i^{-1})-t }  \, \prod_{i=1}^3 {dx_i\over x_i}\cr
  &=&2^3\int_0^\infty x I_0(\sqrt t \,x)K_0(x)^4\,dx\,.
\end{eqnarray}

%-----------------------------------------------------------------------
\subsection{The Picard-Fuchs equation}
\label{sec:picard-fuchs-equat}

In this section we show the three-loop banana integral $I_{\ba}(t)$
satisfies an inhomogeneous Picard-Fuchs equation given
in~\cite{MullerStach:2012mp,PVstringmath}, following the  derivation given in~\cite{PVstringmath}  for the equal masses banana graphs at all loop orders.

\begin{thm}\label{e:PF}
  The three-loop banana integral 
  \begin{equation}\label{e:I formula}
  I_{\ba}(t)= \int_{x_i\geq0}   {1\over (1+\sum_{i=1}^3
    x_i)(1+\sum_{i=1}^3 x_i^{-1})-t}\, \prod_{i=1}^3 {dx_i\over x_i}    
  \end{equation}
 satisfies the
  inhomogeneous Picard-Fuchs equation $\mathcal L^3_t I_{\ba}(t) = -24$
with the Picard-Fuchs operator $\mathcal L_t^3$ given by 
\begin{equation}\label{e:PFdef}
\mathcal  L^3_t:= t^2(t-4)(t-16)   {d^3\over dt^3}+ 6 t (t^2-15t+32) {d^2\over
    dt^2}+(7t^2-68t+64){d\over dt}+ {t-4}\,.
\end{equation}
\end{thm}

This Picard-Fuchs operator already appeared in the work by Verrill
in~\cite{Verrill} and~\cite{MullerStach:2011ru}. We will comment on the relation to this work
in \S\ref{subsec Verrill fam}.

\begin{proof}

We consider the Bessel integral representation  of the previous section
\begin{equation}
  I_{\ba}(t)= \sum_{k\geq0} t^k I_k  
\end{equation}
where $I_k$ is given by~\eqref{e:I3kresult} with
$m_1=m_2=m_3=m_4=1$
\begin{equation}
I_k= {2^4\over \Gamma(k+1)^2} \, \int_0^{+\infty} \left(x\over2\right)^{2k+1} \, K_0(x)^4 \, dx\,.
\end{equation}
Then the action of the Picard-Fuchs operators on this series expansion gives
\begin{equation}
\mathcal L_t^3 I_{\ba}(t)= \sum_{k\geq0} \, \Big( t\alpha_k +\beta_k +{\gamma_k\over t}\Big)\, t^k I_k  
\end{equation}
therefore
\begin{equation}
\mathcal  L_t^3 I_{\ba}(t)=  {\gamma_0 I_0\over t}+{\gamma_1 I_1+\beta_0
      I_0}+ \sum_{k\geq1}  (\alpha_k I_k + \beta_{k+1} I_{k+1}
    +\gamma_{k+2} I_{k+2}) \, t^k 
\end{equation}
Using the result of the lemma~\ref{lem:Bessel} below, we have
$\mathcal L_t^3 I_{\ba}(t)=\gamma_1 I_1+\beta_0      I_0.$  
Evaluating the integrals gives that $  \gamma_1 I_1+\beta_0 I_0= -24  $,
which proves the theorem.

\end{proof}
\begin{lem}\label{lem:Bessel}
  The Bessel moment integrals
  \begin{equation}
    I_k={2^4\over\Gamma(k+1)^2}\,\int_0^{+\infty}\,\left(x\over2\right)^{2k+1} \,K_0(x)^4\,dx
  \end{equation}
satisfy the recursion relation 
\begin{equation}\label{e:recBessel}
     \alpha_k  I_k+ \beta_{k+1} I_{k+1}    +\gamma_{k+2} I_{k+2} =0,
     \qquad k\geq0
\end{equation}
with for $k\geq0$
\begin{eqnarray}
  \alpha_k&:=&(k+1)^3 \cr
\beta_k&:=&-2(2k+1)\,(5k^2+5k+2)\\
\nonumber\gamma_k&:=&64 k^3\,.
\end{eqnarray}
\end{lem}
\begin{proof}

The proof has been given  in~\cite[Example~6]{BorweinSalvy}
(see~\cite{ouvry} for related considerations).
Following this reference we introduce the Bessel moment
integrals $c_{4,2k+1}=2^{2k-3} \Gamma(k+1)^2\, I_{k}$.
One notices that $K_0(x)^4$ satisfies the differential equation
$L_5K_0(x)^4=0$ where 
\begin{equation}
L_5:= \left(x{d\over dx}\right)^5-20x^2\left(x{d\over dx}\right)^3-60x^2\left(x{d\over dx}\right)^2+8x^2(8x^2-9)\left(x{d\over dx}\right)+32x^2(4x^2-1)\,.  
\end{equation}
And finally one notices the identities
\begin{equation}
  \int_0^{+\infty} \, x^{k+j} \left(x{d\over dx}\right)^m \left(K_0(x)^4\right)\, dx=
  (-1-k-j)^m\,c_{4,k+j}\,.
\end{equation}
Therefore integrating term by term the expression
\begin{equation}
  \int_0^{+\infty} t^{2k+1} \, L_5  K_0(x)^4\, dx=0  
\end{equation}
leads to the recursion~\eqref{e:recBessel}.
\end{proof}

\subsection{Solution of the inhomogeneous Picard-Fuchs equation}
\label{sec:solution}

We need an intermediate result  expressing the solution of the
third order differential equation using the Wronskian method. Recall the Wronskian of a linear differential equation 
\begin{equation}\label{wronskiande}f_n(x)y(x)^{(n)}+\ldots+f_{1}(x)y'+f_0(x)y=0
\end{equation} 
is the determinant $W(x) := \det(y_j^{(i)})$ where $y_1,\dotsc,y_n$ are independent solutions. Viewing the equation \eqref{wronskiande} as a system of $n$ first order equations, the Wronskian is the solution of the first order equation given by the determinant of the system. This yields the formula
\begin{equation}W(t) = \exp(-\int^t f_{n-1}(x)/f_n(x)\, dx). 
\end{equation}
Consider the inhomogeneous differential equation 
\begin{equation}\label{inhom}
  f_3(x) y{'''}(x)+ f_2(x) y{''}(x)+ f_1(x) y{'}(x)+ f_0(x) y(x) = S(x)
\end{equation}
Let $y_i(x)$ with $i=1,2,3$ be three independent solutions of the
homogeneous equation. Let
\begin{equation}
  W(t)=\left|
    \begin{matrix}
      y_1(t)&y_2(t)&y_3(t)\cr
y'_1(t)&y'_2(t)&y'_3(t)\cr
y''_1(t)&y''_2(t)&y''_3(t)
    \end{matrix}
\right|
\end{equation}
be the Wronskian of these solutions, and introduce the modified Wronskian 
\begin{equation}
  \widetilde W(t,x)=\left|
    \begin{matrix}
      y_1(x)&y_2(x)&y_3(x)\cr
y'_1(x)&y'_2(x)&y'_3(x)\cr
y_1(t)&y_2(t)&y_3(t)
    \end{matrix}
\right|\,.
\end{equation}
We have the following identities
\begin{gather}\widetilde W(t,t)=0;\quad {\partial \widetilde W(t,x)\over\partial t}|_{x=t}=0;\quad {\partial^2\widetilde
  W(t,x)\over\partial t^2}|_{x=t}=W(t) \\
  \sum_{i=0}^3 f_i(t)\frac{\partial^i}{\partial t^i}\widetilde W(t,x) = 0
\end{gather}
A simple computation now yields the general solution for the inhomogeneous equation \eqref{inhom}
\begin{equation}\label{solution}
  y(t)= \sum_{i=1}^3 \alpha_i \, y_i(t)+ \int_0^t \, \widetilde W(t,x)\, {S(x)\, dx\over
    W(x) f_3(x)} \,.
\end{equation}

For the three-banana graph, the Picard-Fuchs operators
in~\eqref{e:PFdef} has $f_3(x)=x^2(x-4)(x-16)$ and
$f_2(x)=6x(x^2-15x+32)=\frac32\, {df_3(x)\over dx}$,  therefore the
Wronskian is given by 

\begin{equation}
  W(t)=\exp\left(-\int^t {f_2(x)\over f_3(x)}\,dx\right)= {W_0\over (t^2(t-4)(t-16))^{3\over2}}  \,.
\end{equation}
The arbitrary normalization $W_0$ of the Wronskian is be determined in~\eqref{e:WronskianN}.
We now use the fact showed in~\cite[theorem~3]{Verrill}, and reviewed in \S\ref{subsec Verrill fam}, that Picard-Fuchs operator is a symmetric square of the
sunset Picard-Fuch operator. For $t\in \mathbb
P^1\backslash\{0,4,16,\infty\}$ the solutions of the homogenous
equations are given by 
\begin{equation}
  ( y_1(t),y_2(t),y_3(t)) =\varpi_1(\tau)\, (1,  2\pi i\tau,(2\pi i\tau)^2)\,.
\end{equation}
In this expression $\varpi_1(\tau)$ is a period and $\tau$ is the
period ratio. The parameter
$t$ is the Hauptmodul  given by~\cite{Verrill} 
\begin{equation}\label{e:tVerrill}
t(\tau)= H_{\ba}([\tau])=-\left(\eta(\tau)\eta(3\tau)\over \eta(2\tau)\eta(6\tau)\right)^6\,.
\end{equation}
We recall that the Dedekind eta function $\eta(\tau)$ is defined by
\begin{equation}\eta(\tau) = \exp(\pi i\tau/12)\prod_{n=1}^\infty(1-\exp(2\pi in\tau))
\end{equation}
The special values of the Hauptmodul $t=\{0,4,16,+\infty\}$ are
obtained for the values of $\tau=\{0,{-3+i\sqrt3\over 12},{3+i\sqrt3\over6},+i\infty\}$. The nature of the fibers for
these values of the Hauptmodul are discussed in \S\ref{subsec
  Verrill fam}. The value $t=4$ is the
pseudo-threshold of the Feynman integral and the value $t=16$ is
the normal threshold of the Feynman integral.

In the neighborhood $|t|>16$ of $t=\infty$  the holomorphic period is
given by 
\begin{eqnarray}\label{e:realperiod}
  \varpi_1(\tau)&=&{1\over (2\pi i)^3}\, \int_{|x_1|=|x_2|=|x_3|=1} \,
  {1\over  (1+\sum_{i=1}^3 x_i)(1+\sum_{i=1}^3 x_i^{-1})-t
  }\,\prod_{i=1}^3 {dx_i \over x_i}\cr
&=&- \sum_{n\geq0} t^{-n-1} \,{1\over(2\pi i)^3}\int_{|x_1|=|x_2|=|x_3|=1} 
  (1+\sum_{i=1}^3 x_i)^n(1+\sum_{i=1}^3 x_i^{-1})^n\,\prod_{i=1}^3 {dx_i \over x_i}\cr
&&\hspace{1cm}=- \sum_{n\geq0} t^{-n-1} \,\sum_{p+q+r+s=n} \, \left(n!\over p!q!r!s!\right)^2\,.
\end{eqnarray} 
Using the above expression for the Hauptmodul $t$,    the period is expressed as
\begin{equation}\label{e:periodVerrill}
\varpi_1(\tau):=
{(\eta(2\tau)\eta(6\tau))^4\over (\eta(\tau)\eta(3\tau))^2}\,.
\end{equation}

\medskip
Expanding the modified Wronskian 
\begin{eqnarray}
  \widetilde W(t,x)&=& y_1( t)\, W_{23}(x)- y_2(t) \, W_{13}(x)+y_3(t) \,
  W_{12}(x)\cr
&=&\varpi_1 \, (W_{23}(x)- \tau(t) \, W_{13}(x)+\tau(t)^2\,  W_{12}(x))\,.
\end{eqnarray}
and then evaluating yields
\begin{equation}
  W_{12}(t)=2\pi i\varpi_1^2\,{d\tau\over dt},\quad
W_{13}(t)=(2\pi i)^2\varpi_1^2\, 2\tau \,{d\tau\over dt},\quad
 W_{23}(t)=(2\pi i)^3\varpi_1^2\, \tau^2 \,{d\tau\over dt}\,.
\end{equation}
Thus
\begin{equation}
  \widetilde W(t,x)= (2\pi i)^3\varpi_1(\tau)\varpi_1(x)^2\,(\tau(x)-\tau(t))^2  \,  {d\tau\over dx}\,.
\end{equation}
The condition 
\begin{equation}\label{e:WronskianN}
\partial^2_t \widetilde W(t,x)\Big|_{x=t}=W(t)
\end{equation} determines the normalization $W_0=2$ of the Wronskian.

Therefore the tree-loop banana integral is given by 
\begin{multline}\label{e:I3q}
 I_{\ba}(t)=I^\textrm{period}-12(2\pi i)^3\varpi_1(t)\,
  \int_0^{t} \, \left(\tau(x)-\tau(t)\right)^2\,
 (x^2(x-4)(x-16))^{\frac12}\,{d\tau(x)\over dx}\,dx\,.
\end{multline}
where  $I^\textrm{period} $ is an homogeneous solution belonging to $\varpi_1(\tau)(\mathbb
C+\tau\mathbb C+\tau^2\mathbb C)$.

\begin{lem}\label{lem:sigma}
Using the expressions for the Hauptmodul $t$ and the period $\varpi_1$
then the function $\sigma(\tau):=-24\varpi_1(\tau)^2\,
(t(\tau)^2(t(\tau)-4)(t(\tau)-16))^{\frac12}$  has the following representation
\begin{equation}\label{e:sigmaG4}
  \sigma(\tau)= {1\over5}\, \left(-E_4(\tau)+16 E_4(2\tau)+9 E_4(3\tau)-144E_4(6\tau)\right)
\end{equation}
where $E_4(\tau)$ is the Eisenstein series
\begin{equation}
  E_4(\tau)={1\over2\zeta(4)} \sum_{(m,n)\neq(0,0)}{1\over (m\tau+n)^4}= 1+240
  \sum_{n\geq1} n^3 {q^n\over 1-q^n}  
\end{equation}
With $q:=\exp(2\pi i \tau)$ the coefficients $\sigma_n$ of the $q$-expansion
\begin{equation}
  \sigma(\tau)= \sum_{n\geq0} \sigma_n \, q^n 
\end{equation}
are given by $\sigma_0=-24$ and 
\begin{eqnarray}\label{e:sigmanexp}
\sigma_n
= n^3\,  \sum_{m|n} {1\over m^3} \,  \psi(m)
\end{eqnarray}
where
 $\psi(n+6)=\psi(n)$ is an even mod 6 character taking the values
\begin{eqnarray}\label{e:psiDef}
 \psi(1)&=&-48,\quad \psi(2)=720,\quad\psi(3)=384,\cr
\psi(4)&=&720,\quad \psi(5)=-48,\quad \psi(6)=-5760\,.
\end{eqnarray}
\end{lem}
\begin{proof}
  The expression in~\eqref{e:sigmaG4} is obtained by performing a $q$
  expansion and verifying that the coefficients are the same to very
  high-order in the $q$-expansion using~\cite{sage}.

The expression for the Fourier coefficients in~\eqref{e:sigmanexp}
are easily obtained by using that 
\begin{equation}
E_4(\tau)=1+240\sum_{n\geq1}\sigma_3(n) q^n
\end{equation}
where $\sigma_3(n)=\sum_{m|n} m^3$ is the divisor sum, and
a reorganization of  the $q$-expansion mod 6.
\end{proof}

Recall the polylogarithm functions $Li_r(z):= \sum_{n=1}^\infty \frac{z^n}{n^r}$. 
\begin{thm}\label{thm:Li3elliptic}
  The integral $I_{\ba}(t)$ in~\eqref{e:I3}  with $t$ given
  in~\eqref{e:tVerrill}, is given by the following function of $q$
\begin{equation}\label{e:banana3BL}
 I_{\ba}(t(\tau))= \varpi_1(\tau)\, \left(16\zeta(3)+\sum_{n\geq1} {\psi(n)\over n^3} {q^n\over 1-q^n}-4(\log q)^3\right)\,.  
\end{equation}
with $\varpi_1(\tau)$ the period in~\eqref{e:periodVerrill} and
$\psi$ the even mod 6 character  with the values given
in~\eqref{e:psiDef}.
This integral can be expressed as linear combination of the  elliptic
trilogarithms introduced by Beilinson and Levin~\cite{BL1,L,ZagierElliptic}.
\begin{equation}\label{e:3logbanana}
  I_{\ba}(t(\tau))= \varpi_1(\tau) ( 40\pi^2\log q-48\mathcal H_{\ba}(\tau))  
\end{equation}
where
\begin{equation}\label{e:h3defBL}
  {\mathcal H}_{\ba}(\tau):=24\mathcal Li_3(\tau, \zeta_6)+ 21\,\mathcal Li_3(\tau ,\zeta_6^2)
+ 8\mathcal Li_3(\tau,\zeta_6^3)+7  \mathcal Li_3(\tau,1)
\end{equation}
with $\mathcal Li_3(\tau,z)$ defined by 
\begin{multline}\label{e:Eli3def}
\mathcal Li_3(\tau,z):=\Li3(z)+ \sum_{n\geq1}  (\Li3(q^n z)+\Li3(q^n
z^{-1}))\cr - \left(
  -{1\over12}(\log z)^3+{1\over24}\log q\, (\log z)^2-{1\over720}(\log q)^3  \right)\,.
\end{multline} 
\end{thm}

\begin{proof} In order to prove the theorem we just evaluate the
  integral in~\eqref{e:I3q}. We perform the  change of variables 
  $2\pi i\tau(t)=\log q$ and $2\pi i\tau(x)=\log \hat q$  to get
\begin{equation}
 I_{\ba}(t)=I^{\rm period}
+\frac12\,\varpi_1(t)\,
  \int_1^{q} \, \left(\log{\hat q\over
      q}\right)^2\,\sigma(\hat q)\,d\log \hat q\,.
\end{equation}
(Here we used that  $t=0$  for $\tau=0$, and $I^{\rm period}$ is a solution  of the homogenous
Picard-Fuchs equation in  
$\varpi_1(\tau)(\mathbb C+\tau\,\mathbb C+\tau^2\,\mathbb C)$.) The
form of the homogenous solution is determined in~\eqref{e:constants}.

Using the $q$-expansion for $\sigma(\tau)$ and the following integrals
\begin{eqnarray}
  \int_1^{q} \left(\log {\hat q\over q}\right)^2  \, \hat q^n d\log
  \hat q&=&{2
    (q^n-1)-2n\log q-n^2 (\log q)^2\over n^3}\cr
\int_1^{q} \log\left(\hat q\over q\right)^2\, d\log \hat q&=&{(\log q)^3\over 3}\,.
\end{eqnarray}
Summing all the terms we find that
\begin{multline}\label{e:solution1}
  I_{\ba}(t(\tau))=I^{\rm period}\cr
+\varpi_1(\tau)\,
\left( {\sigma_0 \over 6} (\log q)^3 + \sum_{n\geq1} {\sigma_n\over n^3}  \, \left(q^n-\frac12\,(1+\log(q^n))^2\right)\right)\,.
\end{multline}
This  leads to
\begin{multline}
 I_{\ba}(t(\tau))=I^{\rm period}
+{\sigma_0\over6}\,\varpi_1(\tau)(\log q)^3+ \varpi_1(t) \,
\sum_{n\geq1} {\sigma_n\over n^3}  \, q^n\,.
\end{multline}

We remark that the coefficients $\sigma_n$
in~\eqref{e:sigmanexp}  can be expressed in term of the sixth root of
unity $\zeta_6=\exp(i\pi/3)$ 
\begin{equation}
\sigma_n=  
-48n^3\, \left(\sum_{r=1}^6 c_r \sum_{m|n} {1\over m^3} \zeta_6^{r m}\right) \,\qquad n\geq1   
\end{equation}
with  $c_r=\{24, 21, 16, 21, 24, 14\}$.
This allows to express the $q$-expansion 
\begin{equation}\label{e:IdSumLi3}
{\sigma_0\over 6}\,(\log q)^3+\sum_{n\geq1} {\sigma_n\over n^3}\,q^n=
-48\mathcal H_{\ba}(\tau)   +40\pi^2\log q-16\zeta(3)
\end{equation}
where 
\begin{equation}
  {\mathcal H}_{\ba}(\tau):=24\mathcal Li_3(\tau, \zeta_6)+ 21\,\mathcal Li_3(\tau ,\zeta_6^2)
+ 8\mathcal Li_3(\tau,\zeta_6^3)+7  \mathcal Li_3(\tau,1)
\end{equation}
is given  in terms of the elliptic trilogarithms  $\mathcal Li_3(\tau,z)$ of Beilinson and Levin~\cite{BL1,L}  defined by 
\begin{multline}
\mathcal Li_3(\tau,z):=\Li3(z)+ \sum_{n\geq1}  (\Li3(q^n z)+\Li3(q^n
z^{-1}))\cr - \left(
  -{1\over12}(\log z)^3+{1\over24}\log q\, (\log z)^2-{1\over720}(\log q)^3  \right)\,.
\end{multline} 
Therefore the three-loop banana integral is a sum of  elliptic
trilogarithms modulo periods solutions of the homogeneous Picard-Fuchs equation
\begin{equation}\label{e:3log}
  I_{\ba}(t(\tau))
  =\varpi_1(\tau)(\alpha_1+\alpha_2\tau+\alpha_3\tau^2)- 48 \mathcal H_{\ba}(\tau)
\end{equation}
where we have expressed the homogeneous solution $I^{\rm period}$ as
$\varpi_1(\tau)(\alpha_1+\alpha_2\tau+\alpha_3\tau^2)$ with
$\alpha_1$, $\alpha_2$ and $\alpha_3$ arbitrary complex numbers.

Using the relation~\eqref{e:IdSumLi3} and that
\begin{equation}
\sum_{n\geq1}{\sigma_n\over n^3}\, q^n=\sum_{n\geq1} {\psi(n)\over n^3} \,
 {q^n\over 1-q^n}  
\end{equation}
with $\psi(n)$ given in~\eqref{e:psiDef}, one can rewrite the
expression in~\eqref{e:3log} as follows
\begin{multline}
 I_{\ba}(t(\tau))=\varpi_1(\tau)
 \Big(\alpha_1+(\alpha_2-40\pi^2)\tau+\alpha_3\tau^2\cr
+\sum_{n\geq1}
    {\psi(n)\over n^3} \,{q^n\over 1-q^n}- 4(\log q)^3 +16\zeta(3)\Big) \,.  
\end{multline}
Using lemmas~\ref{lemWeil} and~\ref{lemZero} we can evalute the
integral at $t=0$, corresponding to $\tau=0$,
\begin{equation}
  I_\ba(0)= \lim_{\tau\to0} \varpi_1(\tau)\, \left(
    \alpha_1+(\alpha_2-40\pi^2)\tau+\alpha_3\tau^2 +336\zeta(3)\right)  \,.
\end{equation}
Since $\lim_{\tau\to0}\varpi_1(\tau)\sim(48\tau^2)^{-1} $, we have
that
\begin{equation}
     I_\ba(0)=7\zeta(3)+{\alpha_3\over 48}+{1\over48} \lim_{\tau\to0} \tau^{-2}\, \left(
    \alpha_1+(\alpha_2-40\pi^2)\tau\right)  \,.
\end{equation}
Because the integral is finite at $t=0$ with the value
$I_\ba(0)=7\zeta(3)$ as shown in~\cite{Bailey:2008ib,BroadhurstLetter,BroadhurstProc},
we deduce that 

\begin{equation}\label{e:constants}
  \alpha_1=\alpha_3=0; \qquad \alpha_2=40\pi^2\,.  
\end{equation}
 This proves the theorem.
\end{proof}

\begin{rem}
Using~\cite{sage} we have numerically evaluated the integral and the elliptic
trilogarithms at the particular values given in
table~\ref{tab:numerics}, in order to check the validity of the representation in~\eqref{e:3logbanana} for the three-loop banana integral.  

The Feynman integral is regular for $t < 16$. It will be
noted that in Table~\ref{tab:numerics} we give no example with $t > 4$. We are
confident that an analytic continuation of our result
applies for $  4 < t < 16$, but do not attempt to compute any
such value here.

\begin{table}
 \centering 
\begin{tabular}[]{||c|c||}
\hline
 \vphantom{$\Big($}     $\tau$&${-3+i\sqrt3\over12}$\\[2ex]
$t(\tau)$&4\\[2ex]
$I_\ba(t)$&9.109181165853514\\[2ex]
$-48\mathcal H_\ba(\tau)$&$347.868145888636 + 637.725764198092i$\\[2ex]
$\varpi_1(\tau)$&$-0.224110197194 - 0.388170248035i$\\[1ex]
\hline
  \end{tabular}
\begin{tabular}[]{||c|c||}
\hline
  \vphantom{$\Big($}    $\tau$&${-3+i\sqrt{15}\over24}$\\[2ex]
 $t(\tau)$&1\\[2ex]
$I_\ba(t)$&8.570280443360948\\[2ex]
$-48\mathcal H_\ba(\tau)$&$404.292203809358 + 325.565905143148i$\\[2ex]
$\varpi_1(\tau)$&$0.133813847482 - 0.518258802791i$\\[1ex]
\hline
  \end{tabular}
\begin{tabular}[]{||c|c||}
\hline
 \vphantom{$\Big($}   $\tau$&$-(3 + 1.80224199747123i)^{-1}$\\[2ex]
 $t(\tau)$&$ {319\over  80}$\\[2ex]
$I_\ba(t)$&9.106670607198028\\[2ex]
$-48\mathcal H_\ba(\tau)$&$355.272552751915 +625.839953492151i$\\[2ex]
$\varpi_1(\tau)$&$-0.206610686713 - 0.388422174005i$\\[1ex]
\hline
  \end{tabular}
  \caption{Numerical evaluations of the Hauptmodul $t(\tau)$ the three-loop banana integral
    $I_\ba(t)$, the elliptic trilogarithm sum $-48\mathcal
    H_\ba(\tau)$ and the period $\varpi_1(\tau)$.}
  \label{tab:numerics}
\end{table}
\end{rem}

\begin{rem}
The integral expression in~\eqref{e:I3q}
\begin{multline}
  I_\ba(t(\tau)) =
(2\pi i)^3\varpi_1(\tau) \,\int_0^t (\tau(x)-\tau(t))^2\,
\sigma(\tau(x))\, d\tau + \varpi_1\,(\mathbb C+\tau\,\mathbb
C+\tau^2\mathbb C)
\end{multline}
shows that $ I_\ba(t(\tau))/\varpi_1(\tau)$  is an Eichler
integral of the modular form $\sigma(\tau)$.
  Another proof of this will be given in \S\ref{sec Eis
  symb} and in theorem~\ref{thm:main}.
\end{rem}

\subsection{Value of the integral at $t=0$} \label{sec:tzero} 
This section contains the two lemmas needed in proof of the
theorem~\ref{thm:Li3elliptic}, when evaluating 
the  integral at $t=0$  which corresponds to $\tau=0$.

\begin{lem}\label{lemWeil}
We have the following identity
  \begin{equation}\label{e:EisDouble}
    16\zeta(3)+ \sum_{  n\geq1}{\psi(n)\over n^3} {q^n\over 1-q^n}=     {\tau\over 2\pi i}\sum_{m\in\mathbb Z\atop n\geq1} {\psi(n)\over n^2} {1\over (m+n\tau)(m-n\tau)} \,.
  \end{equation}
\end{lem}
\begin{proof}
 Using the Kronecker-regularization  for the sum~\cite{Weil}
\begin{equation}
\sum_{m\in\mathbb Z}{}_e {1\over m+n\tau}= -i\pi {1+q^n\over 1-q^n}
\end{equation}
and that
\begin{equation}
16\zeta(3)+ \sum_{  n\geq1}{\psi(n)\over n^3} {q^n\over 1-q^n}=
 \frac12\sum_{n\geq1} {\psi(n)\over n^3}\,{1+q^n\over 1-q^n}
\end{equation}
we conclude that
\begin{equation}
16\zeta(3)+ \sum_{  n\geq1}{\psi(n)\over n^3} {q^n\over 1-q^n}=-{1\over 2\pi i} \sum_{n\geq1}\sum_{
  m\in\mathbb Z}{}_e {\psi(n)\over n^3} {1\over m+n\tau} \,,
\end{equation}
which can be rewritten as a converging sum
\begin{equation}
 {\tau\over 2\pi i}\sum_{m\in\mathbb Z\atop n\geq1} {\psi(n)\over n^2} {1\over (m+n\tau)(m-n\tau)} \,.
\end{equation}
\end{proof}
This expression is
antisymmetric under the transformation $\tau\to-\tau$.

\begin{lem}\label{lemZero}
  The series in~\eqref{e:EisDouble} has the following
asymptotic behaviour when $\tau\to0$
\begin{equation}
 \lim_{\tau\to 0} \,\tau^{-2}\,  {\tau\over 2\pi i}\sum_{m\in\mathbb Z\atop n\geq1} {\psi(n)\over n^2} {1\over m^2-(n\tau)^2} =
336\zeta(3)\,.
\end{equation}
\end{lem}
\begin{proof}

We start by rewriting the sum as
\begin{eqnarray}
  {\tau\over2\pi i}\,\sum_{n\neq0\atop
     m\geq1}\, {\psi(n)\over n^2}{1\over m^2-(n\tau)^2}
&=&{\tau^3\over2\pi i}\,\sum_{n\neq0\atop
     m\geq1}\, \psi(n) \,\left({1\over n^4 m^2\tau^2}+ {1\over m^2(m^2-(n\tau)^2)}\right)\cr
&=&{\tau^3\over2\pi i}\,\sum_{n\in\mathbb Z,n\neq0\atop
     m\geq1}\, {\psi(n)\over m^2}{1\over m^2-(n\tau)^2}\,.
\end{eqnarray}
where we used $\sum_{n\geq1}\psi(n)/n^4=0$.
Therefore
\begin{equation}
{\tau\over 2\pi i}\sum_{m\in\mathbb Z\atop n\geq1} {\psi(n)\over n^2} {1\over m^2-(n\tau)^2} ={\tau^3\over2\pi i}\,\sum_{n\in\mathbb Z\atop
     m\geq1}\, {\psi(n)\over m^2}{1\over
     m^2-(n\tau)^2}+{5760\,\tau^3\over 2\pi i}\,\zeta(4)\,.
\end{equation}
We perform a Poisson summation on $n$ to get 
\begin{eqnarray}
    \sum_{n\in\mathbb Z} {1\over m^2+((r+6n)\tau)^2}&=& \sum_{\hat
      n\in\mathbb Z}  \int_{-\infty}^{+\infty}\,{e^{-2\pi i x \hat
      n}\over m^2+((r+6x)\tau)^2}\, dx\cr
&=&{\pi\over 6m\tau }  \sum_{\hat n\in \mathbb Z} e^{-\pi {m|\hat n|\over
    3\tau}+i\pi{\hat n r\over 3}}\,.
\end{eqnarray}
Therefore
\begin{equation}
{\tau\over 2\pi}\sum_{m\in\mathbb Z\atop n\geq1} {\psi(n)\over n^2} {1\over m^2+(n\tau)^2} 
=- {\tau^2\over 12}\,\sum_{r=1}^6\sum_{\hat
    n\in\mathbb Z\atop m\geq1} {\psi(r)\over m^3}\, e^{-\pi {m|\hat n|\over
    3\tau}+i\pi{\hat n r\over 3}}
-{63\pi^3\over2}\,\tau^3
\end{equation}
which has  the limit for $\tau\to 0$ 
\begin{equation}
  \lim_{\tau\to i0^+} \,\tau^{-2}\,  {\tau\over 2\pi i}\sum_{m\in\mathbb Z\atop n\geq1} {\psi(n)\over n^2} {1\over m^2-(n\tau)^2} =
-  {\zeta(3)\over12}\,\sum_{r=1}^6\psi(r)= 336\zeta(3)\,.
\end{equation}
\end{proof}

This  result we will obtained using the higher
normal function analysis with the theorem~\ref{thm
  hnf sv}.

\subsection{Value of the integral  at $t=1$}
\label{sec:tone}

  It is numerically obtained in~\cite{BroadhurstLetter,BroadhurstProc}
that the value at $t=1$ of the banana graph is given by a $L$-function
value
\begin{equation}
  \label{e:LvaluePeriod}
  I_{\ba}(1) \stackrel{?}{=}  {12\pi\over \sqrt{15}} L(f^+,2)\,,
\end{equation}
with $L(f^+,s)=\sum_{n\geq1} a_n/n^s$ the $L$-function associated to
the weight three modular form
$f^+(q)=\eta(\tau)\eta(3\tau)\eta(5\tau)\eta(15\tau)\,
\sum_{m,n\in\mathbb Z} q^{m^2+mn+4n^2}=\sum_{n\geq0} a_n q^n$
constructed in~\cite{peterstop}. Because the functional equation
equation is $\Gamma(s) \,(\sqrt{15}/(2\pi))^s \, L(s)= \Gamma(3-s) \,(\sqrt{15}/(2\pi))^{3-s} \,L(3-s)$, the value $s=2$ is
inside the critical band. We show in \S\ref{sec:special-fiber-at} that for $t=1$ the mixed
Hodge structure (motive) associated  to the Feynman integral has rank two.

\smallskip
Anticipating on the relation between the three-banana and sunset
geometry described in~\S\ref{subsec Verrill fam}, we 
use the relation $t(-1/(6\tau))=10-9/t_\su(\tau)-t_\su(\tau)$ between the three-banana Hauptmodul $t$ and the
sunset Hauptmodul  $t_\su(\tau)=9+72\,\eta(\tau)^5\eta(2\tau)\eta(3\tau)^{-1}\eta(6\tau)^5$), one finds that the value $t=1$ is reached\footnote{There is of course another  solution obtained
for $t_\su'(\tau_\su')=\frac32(3+\sqrt5)$ and 
$
 \mathcal E_\su': \qquad y^2 = x^3 + \frac{3}{8} \left(1+3 \sqrt{5}\right)\,x^2 +\frac{3}{2} \left(3+\sqrt{5}\right)\,x
$. These two elliptic curves are isogeneous. We  refer to \S\ref{subsec Verrill fam} for a
review of the relation between the three-banana and the sunset
geometry.} for
$t_\su(\tau_\su)=\frac32(1-\sqrt 5)$  with  $\tau_\su= (3+i\sqrt{15})/6$ and  the sunset elliptic curve is
defined over $\mathbb Q[\sqrt 5]$

 \begin{equation}
\mathcal E_\su:\qquad y^2 = x^3 + \frac{3}{8}\left(1-3 \sqrt{5}\right)\,x^2 +\frac{3}{2} \left(3-\sqrt{5}\right)\,x\,.
 \end{equation}
 This curve has complex multiplication (CM) with discriminant $-15$
  as can be seen by fact that
 $(1+i\sqrt{15}) (\mathbb Z+\tau_\su \mathbb Z)= (\mathbb Z+\tau_\su
\mathbb Z)$.

Getting back to the banana period ratio by
$\tau_\ba=-1/(6\tau_\su)=(-3+i\sqrt{15})/24$, yields
\begin{equation}
  I_{\ba}(1)=\varpi_1(\tau_\ba)\, \left(-4(2\pi i \tau_\ba)^3+
    {\tau_\ba\over 2\pi i}\sum_{m\in\mathbb Z\atop n\geq1} {\psi(n)\over n^2} {1\over m^2-(n\tau_\ba)^2}\right)  \,.
\end{equation}
We remark that $\varpi_1(\tau_\ba)= -\frac34\, \tau_\su^2\,\varpi_r$
with $\tau_\su=(3+i\sqrt{15})/6$ and
\begin{equation}
  \varpi_r= { (\eta(\tau_\su)\eta(3\tau_\su))^4\over
    (\eta(2\tau_\su)\eta(6\tau_\su))^{2}}=(\theta_4(e^{4i\pi\tau_\su})\theta_4(e^{12i\pi\tau_\su}))^2
\end{equation}
which has the following sum expression\footnote{Using the cubic
  modular equation of~\cite[section~5.11]{Bailey:2008ib}, 
  this expression is equal to
  $\frac12\,(\sqrt{15}-\sqrt3)\,\left(1+2\sum_{n\geq1 }  \,e^{-n^2 \pi
      \sqrt{15}}\right)^4$ as given in~\cite{Bailey:2008ib,BroadhurstLetter,BroadhurstProc}. }
\begin{equation}
  \varpi_r=\left(1+2\sum_{n\geq1} \, e^{-n^2 \pi
      \sqrt{\frac53}}\right)^2 \left(1+2\sum_{n\geq1 }  \,e^{-n^2 \pi \sqrt{15}}\right)^2  \,.
\end{equation}
showing that $\varpi_r\in\mathbb R$.
Since the integral is real  we conclude that
\begin{equation}
\Imm \left[\tau_\su^2\, \left(-4(2\pi i \tau_\ba)^3+ {\tau_\ba\over 2\pi i}\sum_{m\in\mathbb Z\atop n\geq1} {\psi(n)\over n^2} {1\over m^2-(n\tau_\ba)^2}\right)  \right]  =0\,,
\end{equation}
that implies
\begin{equation}
  \Imm  \left({\tau_\ba\over 2\pi i}\sum_{m\in\mathbb Z\atop n\geq1} {\psi(n)\over n^2} {1\over m^2-(n\tau_\ba)^2}\right)  = \sqrt{15}\, \Ree  \left( {\tau_\ba\over 2\pi i}\sum_{m\in\mathbb Z\atop n\geq1} {\psi(n)\over n^2} {1\over m^2-(n\tau_\ba)^2}\right)  - {2\pi^3\over3}  \,.
\end{equation}
To evaluate the real part of the series we use 
\begin{eqnarray}
  &&  \Ree\left( {\tau_\ba\over 2\pi i}\, \sum_{m\in\mathbb Z \atop n\geq1}{\psi(n)\over
       n^2}\,{1\over m^2-(n\tau_\ba)^2}\right)\\
&=& {\sqrt{15}\over2\pi}  \sum_{m\geq1\atop n\geq1}{\psi(n)\over
       n^2}\,\left({1\over 24m^2-6mn+n^2}+{1\over
         24m^2+6mn+n^2}\right) \cr
\nn&=&{\sqrt{15}\over2\pi}\,11\,\zeta(4)={11\pi^3\over12\sqrt{15}}
\end{eqnarray}
It then follows
\begin{equation}
  I_{\ba}(1) 
  ={(2\pi )^3\over\sqrt{15}}\,
  {1+i\sqrt{15}\over 16}\, \varpi_1(\tau_\ba) =-{(2i\pi)^3\over
    \sqrt{-15}} {\varpi_r\over 8}\,,
\end{equation}
and the conjecture in~\eqref{e:LvaluePeriod} amounts showing 
\begin{equation}\label{e:Ltop} 
  L(f^+,2)\stackrel{?}{=} - (2\pi i )^2\, {\varpi_r\over 48}\,.
\end{equation}
This relation between the period $\varpi_r$ and the critical value of
the $L$-function is shown in section~\ref{sec:toneBis}
to be correct up to a rational coefficient.

%%%%%%%%%%%%%%%%%%%%%%%%%%%%%%%%%%%%%%%%%%%%%%%%%%%%%%%%%%%%%%%%%%%%%%
\section{\label{sec Family of K3s}The family of $K3$ surfaces}

Our analysis of the three-banana pencil is based on its presentation
both as a family of anticanonical toric hypersurfaces and as a modular
family of Picard-rank-$19$ $K3$ surfaces. Modern research in this area is influenced by the theory of toric varieties, and most particularly the toric variety associated to the Newton polytope of a Laurent polynomial.  Briefly, to a Laurent polynomial $\phi$ in $n$ variables $x_1,\dotsc,x_n$ we associate firstly the set $\frak M_\phi \subset \ZZ^n$ corresponding to exponents of monomials appearing with non-zero coefficient in $\phi$ and secondly the convex hull 
\begin{equation}\label{polytope}\Delta_\phi := \{\sum_{m\in \frak M} a_mm\ |\ a_m\ge 0,\ \sum a_m=1\}\subset \RR^n
\end{equation} 
of these points. Let $x_0$ be another variable and define the graded ring (graded by powers of $x_0$)
\begin{equation}R_\phi:= \CC[\{x_0^rx^m\ |\ r\in \ZZ^{\ge 0}, m\in r\Delta_\phi\cap \ZZ^n\}]\subset \CC[x_0,x_1^{\pm 1},\dotsc,x_n^{\pm 1}]
\end{equation}
Notice that $x_0\phi\in R_\phi$. By definition 
\begin{equation}\PP_{\Delta_\phi}= \text{Proj}\ R_\phi\supset \GG_m^n = \text{Proj}\ \CC[x_0,x_1^{\pm 1},\dotsc,x_n^{\pm 1}]
\end{equation}
where $\text{Proj}\ R$ is the set of homogeneous prime ideals in a graded ring $R$ with the ``trivial'' graded ideal consisting of all elements of graded degree $>0$ omitted. (Alternatively, one may construct $\mathbb{P}_{\Delta_{\phi}}$ by taking the normal fan to $\Delta_{\phi}$.)  Divisors at $\infty$, i.e. in the complement $\PP_{\Delta_\phi}\setminus \GG_m^n$, correspond to codimension $1$ faces (facets) of $\Delta_\phi$. 
For a summary of other important properties of this construction, see \cite{Batyrev}. 

We begin by reviewing
the simplest example of a family of anticanonical modular toric hypersurfaces, the sunset family of elliptic curves
studied in \cite{Bloch:2013tra}. 

%---------------------------------------------------------------------------
\subsection{Sunset in a nutshell}
\label{sec:sunset}

Consider the Laurent polynomial
\[
\phi_{\su}(x,y):=(1+x+y)(1+x^{-1}+y^{-1})
\]
and its associated (hexagonal) Newton polytope $\Delta_{\su}\subset\mathbb{R}^{2}$,
which defines a toric Fano surface $\PP_{\Delta_{\su}}$ ($\PP^{2}$
blown up at three points). Compactifying the hypersurface defined by
\[
t_{\su}-\phi_{\su}(x,y)=0
\]
in $\PP_{\Delta_{\su}}\times\PP^{1}\backslash\mathcal{L}_{\su}$ ($\mathcal{L}_{\su}:=\{0,1,9,\infty\}$)
defines the sunset family 
\[
\mathcal{X}_{\su}\overset{\pi_{\su}}{\twoheadrightarrow}\PP^{1}\backslash\mathcal{L}_{\su}.
\]
For its modular construction, recall that the congruence subgroup
\[
\Gamma_{1}(6)=\left\{ \left.\left(\begin{array}{cc}
a & b\\
c & d
\end{array}\right)\in SL_{2}(\ZZ)\right|a\equiv d\equiv1 \mod 6,\,
c\equiv0 \mod 6\right\} 
\]
of $SL_{2}(\ZZ)$ produces a universal family
\[
\E_{1}(6):=\left.(\ZZ^{2}\rtimes\Gamma_{1}(6))\right\backslash (\CC\times\mathfrak{H})\overset{\pi_{1}}{\sur}\left.\Gamma_{1}(6)\right\backslash \mathfrak{H}=:Y_{1}(6)
\]
of elliptic curves with six marked $6$-torsion points (forming a copy of $\ZZ/6\ZZ$).
Write $\tau$
for the parameter on $\mathfrak{H}$, and $q:=e^{2\pi i\tau}$. Then
we have an isomorphism\[\xymatrix{
\E_1(6) \ar [r]^{\mathcal{H}_{\su}}_{\cong} \ar @{->>} [d]^{\pi_1} & \X_{\su} \ar @{->>} [d]^{\pi_{\su}}
\\
Y_1(6) \ar[r]^{H_{\su}}_{\cong} & \PP^1 \setminus \L_{\su}
}\]of families, in which the Hauptmodul $H_{\su}$
\begin{equation}\label{e:tsunset}
 t_\su=H_{\su}([\tau])=9+72\frac{\eta(2\tau)}{\eta(3\tau)}\left(\frac{\eta(6\tau)}{\eta(\tau)}\right)^{5},
\end{equation}
and maps $[\tau]=[0],[i\infty],[\frac{1}{2}],[\frac{1}{3}]$ to $t_{\su}=\infty,9,1,0$,
respectively. In the semistable compactification of either family,
these points support fibers of (respective) Kodaira types $I_{6,}I_{1},I_{3},I_{2}$.
$\mathcal{H}_{\su}$ sends the marked points on $\pi_1^{-1}([\tau])$
to the six points where $\pi_{\su}^{-1}(H_{\su}([\tau]))$ meets
the toric boundary $\PP_{\Delta_{\su}}\backslash(\CC^{*})^{2}$.

%-----------------------------------------------------------------------------
\subsection{\label{subsec Verrill fam}Verrill's family}

Turning to the three-banana, the relevant pencil 
\[
\X_{\ba}\overset{\pi_{\ba}}{\sur}\PP^{1}\backslash\L_{\ba}
\]
($\L_{\ba}=\{0,4,16,\infty\}$) of $K3$ surfaces
is defined in the same fashion: namely, we compactify the hypersurface
\[
t-\phi_{\ba}(x,y,z)=0
\]
in $\PP_{\Delta_{\ba}}\times\PP^{1}\backslash\L_{\ba}$, where $\Delta_{\ba}\subset\RR^{3}$
is the Newton polytope of
\[
\phi_{\ba}=(1-x-y-z)(1-x^{-1}-y^{-1}-z^{-1}).
\]
Here we are using the coordinate change
$x_{1}=-x$, $x_{2}=-y$, $x_{3}=-z$, which swaps $\RR_{>0}^{\times3}$
with $\RR_{<0}^{\times3}$, for reasons related to the completion
of the Milnor symbol below.

Laurent polynomials with Newton polytope contained in $\Delta_{\ba}$ may be regarded as sections of an ample sheaf $\mathcal{O}(1)$ on $\mathbb{P}_{\Delta_{\ba}}$ \cite[Def. 2.4]{Batyrev}.  The polytope $\Delta_{\ba}$ has 12 vertices $\{\pm e_i\}_{i=1}^3 \cup \{\pm(e_i-e_j)\}_{1\leq i<j\leq 3},$ and a computation shows that its polar polytope $$\Delta_{\ba}^{\circ}:=\left\{ \left. v \in \RR^3  \, \right| \, v\cdot w \geq -1 \, (\forall w\in \Delta_{\ba})\right\}$$ has the 14 vertices $\left\{\pm e_i \right\}_{i=1}^3\cup\left\{ \pm(e_i+e_j)\right\}_{1\leq i<j\leq 3}\cup \left\{\pm(e_1+e_2+e_3)\right\} .$ Since $\Delta_{\ba}^{\circ}$ is evidently integral, $\Delta_{\ba}$ is reflexive \cite[Def. 12.3]{Batyrev2}, and so $\mathcal{O}(1)$ is the anticanonical sheaf [loc. cit, Thm. 12.2].  Moreover, as $\Delta_{\ba}^{\circ}\cap\ZZ^3$ consists only of vertices and $\underline{0}$, by  \cite[Thm. 2.2.9(ii)]{Batyrev}, $\mathbb{P}_{\Delta_{\ba}}$ is smooth apart from 12 point singularities corresponding to vertices of $\Delta_{\ba}$. It follows that for any Laurent polynomial $f$ which is $\Delta_\ba$-regular in the sense of \cite[Defn. 3.1.1]{Batyrev}, the (anticanonical) hypersurface in $\mathbb{P}_{\Delta_{\ba}}$ defined by $f=0$ is a smooth $K3$ \cite[Thm. 4.2.2]{Batyrev}.\footnote{We need not carry out the MPCP-desingularization in [loc. cit.], as such a hypersurface avoids the 12 singular points (of $\mathbb{P}_{\Delta_{\ba}}$) which it resolves.}

\medskip
We shall need to know the structure of ``divisors at infinity'' $\mathbb{D}_{\ba}:=\mathbb{P}_{\Delta_{\ba}}\setminus (\CC^*)^3$ and $D_{\ba}:=\pi_{\ba}^{-1}(t)\cap \mathbb{D}_{\ba}$, the latter of which is the base locus of our pencil (and independent of $t$).  This is understood by examining the facets of $\Delta_{\ba}$ and facet polynomials of $\phi_{\ba}$, as explained in  \cite[$\S$2]{DoranKerr}.  Briefly, we draw a plane $\RR_{\sigma}$ through each facet $\sigma$ and (by choosing an origin) noncanonically identify $\RR_{\sigma}\cap \ZZ^3=:\ZZ_{\sigma}$ with $\ZZ^2$.  The pair $(\sigma,\ZZ_{\sigma})$ then yields a toric Fano surface $\mathbb{D}_{\sigma}$ in the usual manner; these are the components of $\mathbb{D}_{\ba}$.  For $\Delta_{\ba}$, one may choose the identifications with $\ZZ^2$ so that the 8 triangular facets [resp. 6 quadrilateral facets] have vertices $(0,0),(1,0),(0,1)$ [resp. $(0,0),(1,0),(0,1),(1,1)$], whereupon the corresponding $\left\{\mathbb{D}_{\sigma}\right\}$ are evidently isomorphic to $\mathbb{P}^2$ [resp. $\mathbb{P}^1\times \mathbb{P}^1$] (for instance by taking normal fans).

The components $D_{\sigma}:=\pi_{\ba}^{-1}(t)\cap \mathbb{D}_{\sigma}$ of $D_{\ba}$ are obtained by retaining only the terms of the Laurent polynomial with exponent vectors in $\sigma$, and viewing this as a Laurent polynomial in two variables (in a manner made precise in $\S$2.5 of [op. cit.]).  One checks that $D_{\ba}$ is a union of 20 rational curves. The respective configurations of $\mathbb{D}_{\Delta_{\phi}}$ and $D_{\ba}$ are shown below.

\[\includegraphics[scale=0.5]{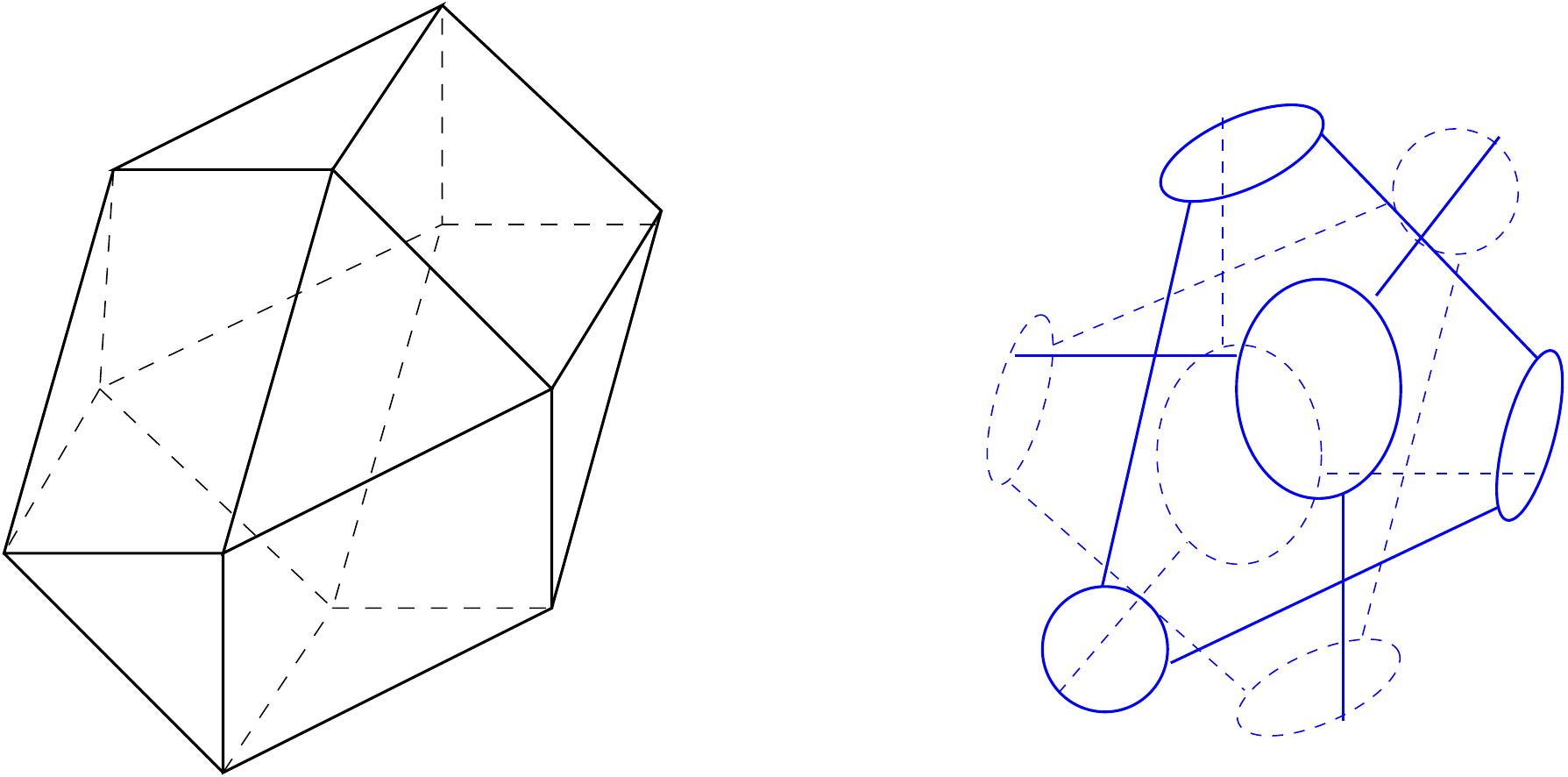}\]

Note that $t-\phi_{\ba}$ fails to be $\Delta_{\ba}$-regular at the point in each boundary $\PP^1\times \PP^1$ where the two (rational curve) components of $D_{\sigma}$ intersect.  However, in local holomorphic coordinates at each such point, $t-\phi_{\ba}$ takes the form $w=uv$; and it follows that for each $t\in \PP^1\setminus\mathcal{L}_{\ba}$, $\pi^{-1}_{\ba}(t)$ is a smooth $K3$.  Finally, as previously mentioned, $\PP_{\Delta_{\ba}}$ has 12 singular points; one way to construct it is by blowing up $\PP^3$ at the 4 ``vertices'' then along the proper transforms of the 6 ``edges'', then blowing down 12 $(-1)$-curves.  One choice of toric (MPCP-)desingularization (as in \cite{Batyrev}) in fact simply reverses this blow-down; note that this produces no additional components in $\mathbb{D}_{\ba}$ and does not affect the $K3$ hypersurfaces.  In subsequent sections, $\PP_{\Delta_{\ba}}$ will denote this smoothed toric 3-fold.

\medskip
The family $\X_{\ba}$ was studied by Verrill \cite{Verrill} (cf.
also \cite{Bertin,DoranKerr}), who proved that the generic fiber
$X_{t}=\pi_{\ba}^{-1}(t)$ has Picard rank $19$. More precisely the local system of $R^2(\pi_{\ba})_*\ZZ$ contains a $19$-dimensional subsystem spanned by divisors. We write $R^2_{var}(\pi_{\ba})_*\ZZ$ for the quotient. The fibres $R^2_{var}(\pi_{\ba})_*\ZZ =: H^2_{var}(X_t)$ have monodromy
group isomorphic to $\Gamma_{1}(6)^{+3}$. The intersection
form is $H\oplus\langle6\rangle$ with discriminant 6.  In particular, $\X_{\ba}$
is a family of $M_{6}:=E_{8}(-1)^{\oplus2}\oplus H\oplus\langle-6\rangle$-polarized
$K3$ surfaces, and is thus of Shioda-Inose type
(cf. \cite{Morrison}). There are countably many $t$ for which the Picard rank is $20$. For these fibres, the transcendental part $H^2_{tr}(X_t)$ is a quotient of $H^2_{var}$ of rank $2$. The motive $H^2_{tr}(X_t)$ for these fibres has complex multiplication, i.e. the rational endomorphism ring is an imaginary quadratic field. 

\medskip
We describe a modular construction of such a family, closely related
to that of \cite[sec. 8.2.2]{DoranKerr}. Set 
\[
\alpha_{3}:=\left(\begin{array}{cc}
\sqrt{3} & \frac{2}{\sqrt{3}}\\
-2\sqrt{3} & -\sqrt{3}
\end{array}\right),\;\beta_{3}:=\left(\begin{array}{cc}
-\sqrt{3} & \frac{1}{\sqrt{3}}\\
-4\sqrt{3} & \sqrt{3}
\end{array}\right),\;\mu_{6}:=\left(\begin{array}{cc}
0 & \frac{-1}{\sqrt{6}}\\
\sqrt{6} & 0
\end{array}\right)
\]
and note that \begin{equation}\label{eqn(*)}\left\{ \begin{array}{c}\beta_{3}\mu_{6}=\mu_{6}\alpha_{3}\;\;\;\;\;\;\;\;\;\;\;\;\;\;\;\;\;\;\;\;\;\;\;\;\\\beta_{3}^{-1}\alpha_{3}=\tiny\left(\begin{array}{cc}5 & 3\\18 & 11\end{array}\right)\normalsize\in\begin{pmatrix}-1 & 0 \\ 0 & -1\end{pmatrix}\Gamma_{1}(6).\end{array}\right.\end{equation}We
have $\alpha_{3}(\tau)=-\frac{\tau+\frac{2}{3}}{2\tau+1}$, $\mu_{6}(\tau)=\frac{-1}{6\tau}$.
These induce involutions on $Y_{1}(6)$ since \begin{equation}\label{eqn(**)}\begin{array}{ccccc} & \Gamma_{1}(6) & \vartriangleleft & \Gamma_{1}(6)^{+3} & :=\langle\Gamma_{1}(6),\alpha_{3}\rangle\\ & \vartriangle &  & \vartriangle\\\langle\Gamma_{1}(6),\mu_{6}\rangle=: & \Gamma_{1}(6)^{+6} & \vartriangleleft & \Gamma_{1}(6)^{+3+6} & :=\langle\Gamma_{1}(6),\alpha_{3},\mu_{6}\rangle\end{array}\end{equation}and
$\alpha_{3}^{2}=\left(\begin{array}{cc}
-1 & 0\\
0 & -1
\end{array}\right)=\mu_{6}^{2}$. (The action on cusps is $[i\infty]\leftrightarrow[\frac{1}{2}]$,
$[0]\leftrightarrow[\frac{1}{3}]$ for $\alpha_{3}$ and $[i\infty]\leftrightarrow[0]$,
$[\frac{1}{2}]\leftrightarrow[\frac{1}{3}]$ for $\mu_{6}$.) From
\eqref{eqn(*)} one deduces that these involutions commute; and so
$\mu_{6}$ descends to $Y_{1}(6)^{+3}:=\langle\alpha_{3}\rangle\backslash Y_{1}(6)^{*\alpha_{3}}$
and $\alpha_{3}$ to $Y_{1}(6)^{+6}:=\langle\mu_{6}\rangle\backslash Y_{1}(6)^{*\mu_{6}}$,
where ``$*$'' means to delete fixed (elliptic) points.

\medskip
Let $'\cE_{1}(6)\overset{'\pi_{1}}{\sur}Y_{1}(6)$ be the fiber-pullback
of $\pi_{1}$ by $\alpha_{3}$. (Note that $\alpha_{3}$ and $\mu_{6}$
do not lift to involutions of $\E_{1}(6)$, but do lift to $3:1$
resp. $6:1$ fiberwise isogenies.) Put $'\cE_{1}^{[2]}(6):=\cE_{1}(6)\underset{Y_{1}(6)}{\times}{}'\cE_{1}(6)$,
and let
\[
I_{3}^{[2]}:{}'\cE_{1}^{[2]}(6)\overset{\cong}{\to}{}'\cE_{1}^{[2]}(6)
\]
be the involution given by
\[
(\tau;[z_{1}]_{\tau},[z_{2}]_{\alpha_{3}(\tau)})\mapsto(\alpha_{3}(\tau);[z_{2}]_{\alpha_{3}(\tau)},[z_{1}]_{\tau}).
\]
A first approximation to the three-banana family is then
\[
\E_{1}^{[2]}(6)^{+3}:=\left.I_{3}^{[2]}\right\backslash {}'\cE_{1}^{[2]}(6)^{*\alpha_{3}}\overset{\pi_{2}}{\sur}Y_{1}(6)^{+3}.
\]
It has fibers of type $E_{[\tau]}\times E_{[\alpha_{3}(\tau)]}$,
hence intersection form $H\oplus\langle6\rangle$ on $H_{var}^{2}$,
and the same local system as $R_{var}^{2}(\pi_{\ba})_{*}\ZZ_{\X_{\ba}}$.
By Schur's lemma and the Theorem of the Fixed Part \cite{Schmid},
a $\CC$-irreducible $\ZZ$-local system can underlie at most one
polarized $\ZZ$-variation of Hodge structure, making the two variations isomorphic.

However, $\pi_{2}$ is not yet a family of $K3$ surfaces. Quotienting
fibers by $(-\text{id})^{2}$ and resolving singularities yields a
family of Kummer $K3$ surfaces, with (incorrect) intersection form
$(H\oplus\langle6\rangle)[2]$ on $H_{var,\ZZ}^{2}$. To correct this
multiplication by $2$, we require a fiberwise-birational $2:1$ cover
of the Kummer family, which is the Shioda-Inose family \cite{Morrison}
$\X_{1}(6)^{+3}$ over $Y_{1}(6)^{+3}$. Since this is a family of
$M_{6}$-polarized $K3$ surfaces with integral $H^{2}$ isomorphic
to $\pi_{\ba}$, the relevant global Torelli theorem (cf. \cite[Cor. 3.2]{Dolgachev})
yields an isomorphism \[\xymatrix{
\X_1(6)^{+3} \ar @{->>} [d]^{\pi} \ar [r]_{\cong}^{\mathcal{H}_{\ba}} & \X_{\ba} \ar @{->>} [d]^{\pi_{\ba}}
\\
Y_1(6)^{+3} \ar [r]^{H_{\ba}}_{\cong} & \PP^1 \setminus \L_{\ba}.}\]Explicitly, the Hauptmodul (mapping $[i\infty]\mapsto \infty$, $[0]\mapsto 0$,
elliptic points$\mapsto 4,16$) is given by~\eqref{e:tVerrill}
and we have the relation
\begin{equation}\label{e:smap}
t=\frac{-64t_{\su}}{(t_\su-9)(t_\su-1)}\,.\end{equation}

This relation between the Hauptmoduls of Feynman integrals
with two and three loops was obtained more than 40 years ago
by Geoffrey Joyce, who established a corresponding result
for honeycomb and diamond lattices in condensed matter
physics, exploiting results on integrals of Bessel functions
by Wilfrid Norman Bailey in the 1930s. For further details
of the striking relationships between Feynman integrals and
lattice Green functions, see~\cite{Bailey:2008ib}.

%----------------------------------------------------------------------------
\subsection{\label{subsec misc}Miscellany}

Two observations about $\mathcal{H}_{\ba}$ are in order. 
The first (used below in $\S$5.2) is that we may construct a family $\tilde{\X}\to Y_{1}(6)^{+3}$ of
smooth surfaces mapping onto $\X_{1}(6)^{+3}$ and $\E_{1}^{[2]}(6)^{+3}$
(over $Y_{1}(6)^{+3}$), with both projections generically $2:1$
on each fiber.
We may then transfer generalized
algebraic cycles from $\E_{1}^{[2]}(6)^{+3}$ to $\X_{\ba}$ by composing
this correspondence with $\mathcal{H}_{\ba}$; and the Abel-Jacobi
maps are then related by the action of this correspondence on cohomology
(which is an \emph{integral} isomorphism on $H_{tr}^{2}$ after multiplication
by $\frac{1}{2}$).  To obtain the family $\tilde{\X}$, we take (a) the fiber product $\check{\E}_a$
of $\E^{[2]}_1(6)^{+3}$ and the Kummer family over $\E_1^{[2]}(6)^{+3}/\langle(-id)^{\times 2}\rangle$
and (b) the fiber product $\check{\E}_b$ of the Kummer family and $\X_1(6)^{+3}$ over the
quotient of $\X_1(6)^{+3}$ by the Nikulin involution (cf. \cite{Morrison}).  Smoothing these families yields
$\E_a$ and $\E_b$, whose fiber product over the Kummer family followed by resolution of singularities 
yields $\tilde{\X}$.

The second observation\footnote{This is not used in the sequel, but illustrates an important difference between
this family and the Ap{\'e}ry family of $K3$ surfaces (cf. \cite{DoranKerr}), which \emph{does} admit such an involution.} is that we may use $\mathcal{H}_{\ba}$ to perform a rational
involution on relative cohomology of the family over the automorphism
$\mu:t\mapsto\frac{4^{3}}{t}$ induced by $\mu_{6}$. First of all,
$\X_{\ba}$ does not itself have a birational involution over $\mu$,
since $H_{var}^{2}(X_{t},\ZZ)\cong H_{var}^{2}(E_{\tau}\times E_{\alpha_{3}(\tau)},\ZZ)$
and $H_{var}^{2}(X_{\frac{1}{4^{3}t}},\ZZ)\cong H_{var}^{2}(E_{\mu_{6}(\tau)}\times E_{\alpha_{3}(\mu_{6}(\tau))},\ZZ)$
are rationally but not integrally isomorphic. In particular, we only
have a correspondence \[\xymatrix{
\E^{[2]}_1(6)^{+3} \ar @{~>} [r] \ar @{->>} [d] & \E^{[2]}_1(6)^{+3} \ar @{->>} [d]
\\
Y_1(6)^{+3} \ar [r]^{\mu_6}_{\cong} & Y_1(6)^{+3}
}\]which is a $2:1$ isogeny in the first factor and $1:2$ multivalued
map in the second factor, given by
\[
\left(\tau;[z_{1}]_{\tau},[z_{2}]_{\alpha_{3}(\tau)}\right)\mapsto\left(\mu_{6}(\tau);\left[\frac{(2\tau+1)z_{2}}{\tau}\right]_{\mu_{6}(\tau)},\left[\frac{z_{1}}{2(-3\tau+1)}\right]_{\alpha_{3}(\mu_{6}(\tau))}\right).
\]
However, the graph of this correspondence is a family of abelian surfaces,
mapping fiberwise $2:1$ onto both $\E_{1}^{[2]}(6)^{+3}$ and its
$\mu_{6}$-pullback, which \emph{does} have an involution over $\mu_{6}$.
This family, or its associated Shioda-Inose $K3$ family, can then
be used as a correspondence (inducing isomorphisms of \emph{rational}
$H_{tr}^{2}$) between $\X_{1}(6)^{+3}$ and its $\mu_{6}$-pullback
over $Y_{1}(6)^{+3}$.

Finally, for future reference we shall write down a family of holomorphic 2-forms
on the fibers of $\pi_{\ba}$.  For any $t\in \PP^1\setminus \mathcal{L}_{\ba}$, let
\begin{equation} \label{eqn **23} \omega_t:=Res_{X_t}\left(
    \frac{\frac{dx}{x}\wedge \frac{dy}{y}\wedge
      \frac{dz}{z}}{1-t^{-1}\phi_{\ba}} \right)\end{equation} be the
standard residue form.
 Remark that the holomorphic period in the neighborhood $|t|>16$ of $t=\infty$ may be computed by integrating $\frac{1}{2\pi i}\frac{\frac{dx}{x}\wedge \frac{dy}{y}\wedge \frac{dz}{z}}{1-t^{-1}\phi_{\ba}}$ over the product $(S^1)^{\times 3}$ of unit circles.  By the Cauchy residue theorem, this is 
 \begin{equation}\label{eqn ***23} (2\pi i)^2 \sum_{k\geq 0}a_k
   t^{-k}, \end{equation} where $a_k$,  given in~\eqref{e:realperiod}, is the constant term in
 $(\phi_\ba)^k$.

%%%%%%%%%%%%%%%%%%%%%%%%%%%%%%%%%%%%%%%%%%%%%%%%%%%%%%%%%%%%%%%%%%%%
\section{\label{sec HNF}The three-banana integral as a higher normal function}

In this section we shall explain the precise relationship between
the integral $I_{\ba}$ and the family $\X_{\ba}$ of $K3$ surfaces
defined by the denominator of the integrand. Properly understanding
this, even without the modular description (done in $\S$\ref{sec Eis symb}),
leads at once to the inhomogeneous equation ($\S$\ref{subsec reinterp of I}) and the special values
at $t=0$ and $1$ ($\S$\ref{subsec hnf analysis}).

There are a number of general comments. The integral $I_\ba$
\eqref{e:I3series} is a period, i.e. the integral of a rational
differential form $\omega$ on a variety $P$ over a chain $c$ whose
boundary $\partial c$ is supported on a proper closed subvariety
$\Sigma\subset P$. This theme goes back to Abel's theorem on Riemann
surfaces. For Abel, $P$ is a Riemann surface, $\Sigma = \{p,q\}\subset
P$ is a set of two points, $\omega$ is a holomorphic $1$-form on $P$
and $c$ is a path from $p$ to $q$. In modern terms, this process
associates to the $0$-cycle $(p)-(q)$ an extension of Hodge structures 
$$0 \to H^1(P,\QQ (1)) \to H \to \QQ(0) \to 0. 
$$
The second point is that dependence on external momenta means that we have a family of integrals depending on a parameter $t$. The corresponding family of extensions is called a {\it normal function} and first appeared in the work of Poincar\'e \cite{Poincare,Griffiths2}.

Finally, it turns out that the three-banana amplitude is associated to
a {\it generalized normal function} arising from a family of
``higher'' algebraic cycles or motivic cohomology classes
\cite{KerrLewis,DoranKerr}. The passage from classical normal
functions associated to families of cycles to normal functions
associated to motivic classes suggests interesting new links between
mathematics and physics (op.cit.). For one thing, motivic normal
functions can, in many cases, be associated with multiple-valued
holomorphic functions which arise as amplitudes. For a discussion of
normal functions in physics, cf. \cite{MorrisonWalcher} for
instance.

\medskip
Briefly, the higher Chow groups $CH^p(X,q)$ of a variety $X$ over a
field $k$ are the homology groups of a complex $\sZ^p(X,\bullet)$. By
definition $\sZ^p(X,q)$ is the free abelian group on irreducible
codimension $p$ subvarieties $V\inj X\times (\PP^1\setminus\{1\})^q$
meeting faces properly, where faces are defined by setting various
$\PP^1$-coordinates to be $0$ or $\infty$. Elements of $Z^p(X,q)$ are
called (higher Chow) \emph{precycles}. The face maps $\sZ^p(X,q) \to
\sZ^p(X,q-1)$ are defined by restrictions to faces with alternating
signs; elements of the kernel are called (higher Chow) \emph{cycles}. 

 If $f_1,\dotsc,f_p$ are rational functions on $X$, the locus $\{x,
 f_1(x),\dotsc,f_p(x)\}$ will (assuming the zeroes and poles of the
 $f_i$ are in general position) define a precycle in $\sZ^p(X,p)$. The
 easiest way for its image under the face map to vanish, so that this
 precycle is a cycle and represents a class in $CH^p(X,p)$, is for the
 $f_i$ to be units (invertible functions) on the complement of the the
 subvariety of $X$ defined by $\prod_{j=1}^p(f_j(x)-1)=0$.  A basic
 theorem of Suslin and Totaro identifies $CH^p(\text{Spec } k,p)\cong
 K^M_p(k)$, the $p$-th Milnor $K$-group of the field $k$. These groups
 are linked to algebraic $K$-theory via the $\gamma$-filtration 
$$CH^p(X,q)\otimes \QQ \cong gr_\gamma^pK_q(X).
$$

Finally, in keeping with modern usage, we will define {\it motivic cohomology} by
$$H^r_M(X, \ZZ(s)):= CH^s(X,2s-r)
$$
when $X$ is smooth.  Notice that $H^r_M(X,\ZZ(r))=CH^r(X,r)$ in this
case.  More generally, $H^r_M(X,\QQ(s))$ may be constructed from
higher Chow precycles as described in $\S$1.3 of \cite{DoranKerr},
which leads to a long-exact sequence used only briefly at the end of
$\S$\ref{subsec K3 of K3} below. 

%-----------------------------------------------------------------------------
\subsection{\label{subsec K3 of K3}$\mathbf{K_{3}}$ of a $\mathbf{K3}$!}

Let $X_{t}=\pi_{\ba}^{-1}(t)$ ($t\in\PP^{1}\backslash\L_{\ba}$)
be as in $\S$\ref{subsec Verrill fam}, $X_{t}^{*}:=X_{t}\cap(\CC^{*})^{3}=X_{t}\backslash D_{\ba}$,
$D_{\ba}=\cup_{j=1}^{20}D_{j}$ ($D_{j}\cong\PP^{1}$). The Milnor
symbol
\[
\left\{ x|_{X_{t}},y|_{X_{t}},z|_{X_{t}}\right\} \in K_{3}^{M}(\CC(X_{t}))\cong\underset{\tiny\begin{array}{c}
U\subset X_{t}\\
\text{Zar. op.}
\end{array}}{\underrightarrow{\lim}}H_M^{3}(U,\ZZ(3))
\]
extends to a (cubical) higher Chow cycle
\[
[\xi_{t}]:=\left[\Delta_{(\CC^{*})^{3}}\cap X_{t}^{*}\times\square^{3}\right]\in CH^{3}(X_{t}^{*},3)=H^3_M(X_{t}^{*},\ZZ(3)),
\]
where $\square:=\PP^{1}\backslash\{1\}$ and $[\cdots]$ denotes cycle
class. To (integrally) lift $[\xi_{t}]$ to a class
\[
[\Xi_{t}]\in H_M^{3}(X_{t},\ZZ(3))\,
\]
in the exact sequence%
\footnote{The ambiguities of this lift by the images of the $H^{1}_M(D_{j},\ZZ(2))$
may for our purposes be ignored, as they have no bearing upon the
transcendental part of its Abel-Jacobi image.%
} 
\[
\oplus_{j}H_M^{1}(D_{j},\ZZ(2))\to H_M^{3}(X_{t},\ZZ(3))\to H_M^{3}(X_{t}^{*},\ZZ(3))\overset{\tiny Tame}{\to}\oplus_{j}H_M^{2}(D_{j}^{*},\ZZ(2)),
\]
we must check vanishing of the $Tame_{D_{j}^{*}}([\xi_{t}])$. Inspection
of the edge polynomials \cite[sec. 2.5]{DoranKerr} shows that these
are all of the form $\{\pm u,1\}$, $\{1,\pm v\}$, and $\{\pm u,1-(\pm u)\}$
(in toric coordinates $\{u,v\}$ on $\DD_{j}^{*}\cong(\CC^{*})^{2}$),
which are trivial.

On the cycle level, the mechanism by which the lift takes place is
given by the moving lemma for higher Chow groups \cite{Bloch1994}.
This yields a quasi-isomorphism
\[
Z^{3}(X_{t}^{*},\bullet)\underset{\jmath^{*}}{\overset{\simeq}{\leftarrow}}Z^{3}(X_{t},\bullet)/\imath_{*}^{D}Z^{2}(D_{\ba},\bullet)
\]
inducing the above exact sequence, and there exists $\mu_{t}\in Z^{3}(X_{t}^{*},4)$
such that%
\footnote{Note: in this paper ``$\partial$'' is used both to denote the boundary
of a $C^{\infty}$ cochain and the differential in the higher Chow
complex.%
}
\[
\xi_{t}+\partial\mu_{t}=\jmath^{*}\Xi_{t}.
\]
Moreover, there are $6$ of the $D_{j}$ (say, $j=1,\ldots,6$) on
which $x$, $y$, or $z$ is identically $1$, so that we may replace
in this argument $X_{t}^{*}$ by $X_{t}^{\sim}:=X_{t}\backslash\cup_{j=7}^{20}D_{j}$,
$\xi_{t}$ by its Zariski closure $\xi_{t}^{\sim}\in Z^{3}(X_{t}^{\sim},3)$,
and $\mu_{t}$ by some $\mu_{t}^{\sim}$. The fact that the configuration
$\mathcal{J}=\cup_{j=7}^{20}D_{j}$\[\includegraphics[scale=0.5]{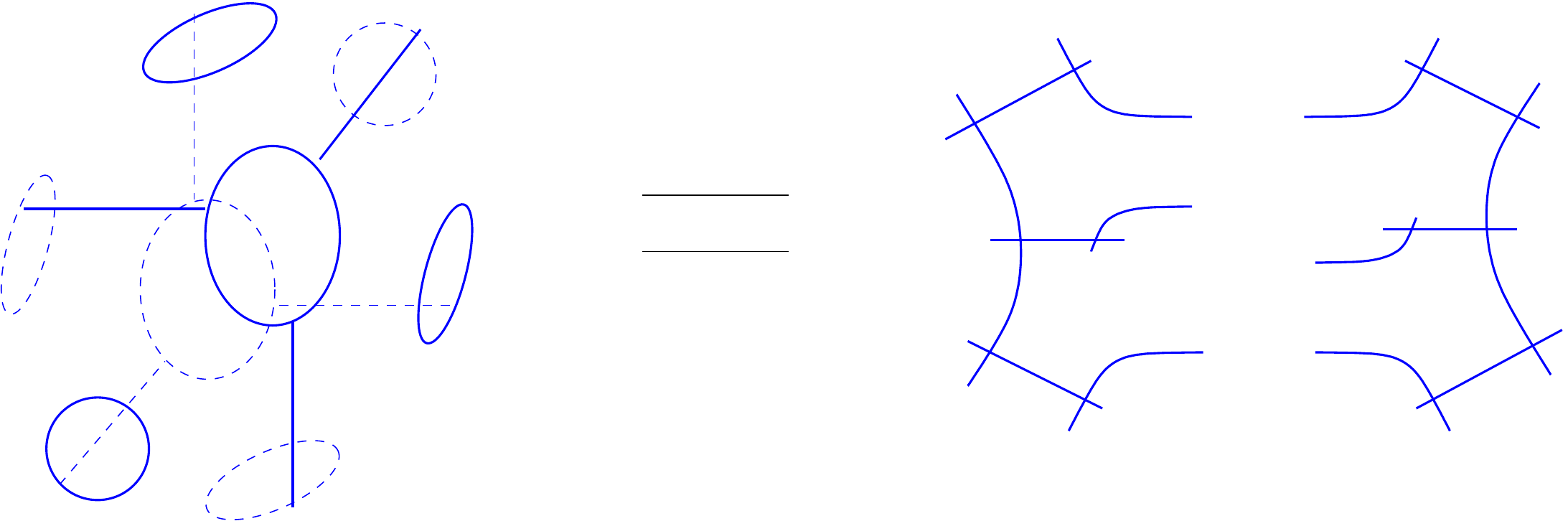}\]has
trivial $H_{1}$ will be crucial for the argument in $\S$\ref{subsec reinterp of I}
below.

Working modulo torsion, one can do somewhat better than a lift $[\Xi_{t}]$
for each $t\in\PP^{1}\backslash\L_{\ba}$ that is ambiguous by the
image of $\oplus_j H_M^{1}(D_{j},\QQ(2))$. Let $\bar{\X}_{\ba}\overset{\bar{\pi}}{\to}\PP^{1}\backslash\{\infty\}$
be the Zariski closure of $\X_{\ba}$ in $\PP_{\Delta_{\ba}}\times(\PP^{1}\backslash\{t=\infty\})$.
One shows that $\phi_{\ba}$ is reflexive and tempered and the assumptions
of \cite[Rem. 3.3(iv)]{DoranKerr} hold (with $K=\QQ$). So by \cite[Thm. 3.1]{DoranKerr},
there exists a motivic cohomology class $[\bar{\Xi}_{\ba}]\in H_{M}^{3}(\bar{\X}_{\ba},\QQ(3))$
defined over $\QQ$ and restricting to $[\{x,y,z\}]\in H_M^{3}\left((\CC^{*})^{3},\QQ(3)\right)$
under the inclusion $(\CC^{*})^{3}\hookrightarrow\bar{\X}_{\ba}$
given by $(x,y,z)\mapsto(x,y,z,\phi_{\ba}(x,y,z)^{-1})$. Its fiberwise
restrictions therefore produce rational lifts of $\xi_{t}$, and since
$H_M^{1}\left((\mathbb{A}^{1}\times D_{j})_{/\QQ},\QQ(2)\right)\cong H_M^{1}(Spec(\QQ),\QQ(2))=\{0\}$
there is also no ambiguity. This guarantees that the processes described
above can be carried out in a ``continuous'' fashion, and that the
lift extends (as a motivic cohomology class) across the singular fibers
over $t=16,4,0$.

In fact, the construction of $[\bar{\Xi}_{\ba}]$ in this
case is quite simple. The total space $\bar{\X}_{\ba}$ has six singularities
(of the local type $xy=zw$), situated over $t=0$ in the base
locus where the two $D_{j}$'s in each $\PP^{1}\times\PP^{1}$ component
cross. Blowing these points $\{p_{k}\}_{k=1}^{6}$ up, we have exceptional
divisors $E_{k}\cong\PP^{1}\times\PP^{1}$ ($k=1,\ldots,6$) in $\tilde{\bar{\X}}_{\ba}$,
and the long exact sequence 
\[
\bigoplus_{k}H_M^{2}(E_{k},\QQ(3))\to H_M^{3}(\bar{\X}_{\ba},\QQ(3))\to\mspace{200mu}
\]
\[
H_M^{3}(\tilde{\bar{\X}}_{\ba},\QQ(3))\oplus\bigoplus_{k}H_M^{3}(\{p_{k}\},\QQ(3))\overset{\alpha}{\to}\bigoplus_{k}H_M^{3}(E_{k},\QQ(3)).
\]
One easily lifts $[\{x,y,z\}]$ to $[\tilde{\Xi}_{\ba}]\in H_M^{3}(\tilde{\bar{\X}}_{\ba},\QQ(3))$
(since the Tame symbols vanish), whereupon $\alpha([\tilde{\Xi}_{\ba}],0)$
vanishes since $x$, $y$, or $z$ was $1$ at each $p_{k}$.

In the sequel, the restriction of $[\bar{\Xi}_{\ba}]$ to $H^3_M(\X_{\ba},\QQ(3))$ will be denoted by $[\Xi_{\ba}]$; we call this the \emph{three-banana cycle}.

%----------------------------------------------------------------------------
\subsection{\label{subsec AJ rev}Review of Abel-Jacobi}

We shall need a few generalities on regulator currents for the arguments
below. The presentation will be sketchy, as a more thorough exposition
may be found in \cite[sec. 1]{DoranKerr}.

Let $X$ be a smooth projective variety with complexes of currents
$\mathcal{D}^{\bullet}(X)$ and $(2\pi i)^{p}\mathbb{A}$-valued $C^{\infty}$-cochains
$C_{top}^{\bullet}(X;\mathbb{A}(p))$ ($\mathbb{A}\subset\RR$ a subring).
Given a cochain $\gamma$, we write $\delta_{\gamma}$ for the current
of integration over it, and use this to define the Deligne complex
\[
C_{\mathscr{D}}^{\bullet}(X,\mathbb{A}(p)):=\left( C_{top}^{\bullet + 1}(X;\mathbb{A}(p))\oplus F^{p}\mathcal{D}^{\bullet + 1}(X)\oplus\mathcal{D}^{\bullet}(X) \right) [-1]
\]
with differential \begin{equation}\label{eqn D-differential}D(T,\Omega,R):=(\partial T,-d[\Omega],d[R]-\Omega+\delta_{T}).\end{equation}Its
$(2p-n)^{\text{th}}$ cohomology sits in a short-exact sequence 
\[
0\to J^{p,n}(X)_{\mathbb{A}}\to H_{\mathscr{D}}^{2p-n}(X,\mathbb{A}(p))\to Hg^{p,n}(X)_{\mathbb{A}}\to0,
\]
where
\[
\left\{ \begin{array}{c}
Hg^{p,n}(X)_{\mathbb{A}}:=Hom_{\mathbb{A}\text{-MHS}}(\mathbb{A}(0),H^{2p-n}(X,\mathbb{A}(p)))\\
J^{p,n}(X)_{\mathbb{A}}:=Ext_{\mathbb{A}\text{-MHS}}^{1}(\mathbb{A}(0),H^{2p-n-1}(X,\mathbb{A}(p)))
\end{array}\right..
\]

Let $Z^{p}(X,\bullet)$ be the codimension-$p$ higher Chow cycle
complex with $n^{\text{th}}$ homology
$CH^{p}(X,n)=H_M^{2p-n}(X,\ZZ(p))$, and boundary map $\partial$; in particular, $Z^{p}(X,n)$
is a subgroup of the cycle group $Z^{p}(X\times\square^{n})$.
Denote by $Z_{\RR}^{p}(X,\bullet)\subset Z^{p}(X,\bullet)_{\QQ}$
the quasi-isomorphic subcomplex%
\footnote{These are still precycles with $\QQ$-coefficients; the ``$\RR$''
refers to intersection conditions with real-analytic chains.%
} described in \cite[sec. 8.2]{KerrLewis}. By \cite[sec. 7]{KerrLewisSMS},
the cycle class map
\[
c_{\mathscr{D}}^{p,n}:CH^{p}(X,n)_{\QQ}=H^{2p-n}_M(X,\QQ(p)) \to H_{\mathscr{D}}^{2p-n}(X,\QQ(p))
\]
defined in \cite{Bloch1986} is computed by a map of complexes
\[
Z_{\RR}^{p}(X,\bullet)\to C_{\mathscr{D}}^{2p-\bullet}(X,\QQ(p)).
\]
Taking $\bullet=n$, it is defined on irreducible components by\footnote{This differs from the formula in [KLM] by a $(2\pi i)^{-\dim(X)}$ twist arising there from Poincare duality, since we interpret currents here as
computing cohomology, not homology. (This choice is more convenient for computation.)} \begin{equation}\label{eqn KLM map}\xi\longmapsto(2\pi i)^{p-n}\left((2\pi i)^{n}T_{\xi},\Omega_{\xi},R_{\xi}\right),\end{equation}where
(writing $\pi_{X},\pi_{\square}$ for the projections from a desingularization
$\tilde{\xi}$ to $X,\,\square^{n}$) $R_{\xi}$ {[}resp. $\Omega_{\xi}$,
$T_{\xi}${]} is defined by applying $(\pi_{X})_{*}(\pi_{\square})^{*}$
to%
\footnote{Here $\log(z)$ is regarded as a $0$-current on $\PP^{1}$ with branch
cut along $\RR_{<0}$, so that $d[\log(z)]=\frac{dz}{z}-2\pi i\delta_{T_{z}}$.
Operations involving pullback are not in general defined on currents,
but a convergence argument (when $\xi$ is in the subcomplex) shows
that $R_{\xi}$ and $\Omega_{\xi}$ are in fact currents on $X$.%
}
\[
R_{n}:=\sum_{j=1}^{n}((-1)^{n}2\pi i)^{j-1}\log(z_{j})\frac{dz_{j+1}}{z_{j+1}}\wedge\cdots\wedge\frac{dz_{n}}{z_{n}}\delta_{T_{z_{1}}\cap\cdots\cap T_{z_{j-1}}}
\]
\[
\left[\text{resp. }\Omega_{n}:=\bigwedge_{j=1}^{n}\frac{dz_{j}}{z_{j}},\; T_{n}:=\bigcap_{j=1}^{n}T_{z_{j}}:=\bigcap_{j=1}^{n}z_{j}^{-1}(\RR_{<0})\right].
\]
Properties of $T_{n},\Omega_{n},R_{n}$ imply that \begin{equation}\label{eqn dR=00003DO-T+Res}d[R_{\xi}]=\Omega_{\xi}-(2\pi i)^{n}\delta_{T_{\xi}}+2\pi iR_{\partial\xi},\end{equation}so
that (by \eqref{eqn D-differential}) \eqref{eqn KLM map} gives a
map of complexes.

Suppose $\partial\xi=0$ (so that $[\xi]\in H_M^{2p-n}(X,\QQ(p))$) and
$n\geq1$. Since $[T_{\xi}]$ and $[\Omega_{\xi}]$ define the map
to $Hg^{p,n}(X)_{\QQ}$(which is zero for $n\geq1$), there exist
$K\in F^{p}\mathcal{D}^{2p-n-1}(X)$ and $\Gamma\in C_{top}^{2p-n-1}(X;\QQ(p))$
such that $\Omega_{\xi}=d[K]$ and $T_{\xi}=\partial\Gamma$, whereupon
\[
\tilde{R}_{\xi}:=R_{\xi}-K+(2\pi i)^{n}\delta_{\Gamma}
\]
defines a closed current with class $[\tilde{R}_{\xi}]\in H^{2p-n-1}(X,\CC)$
projecting to 
$$
(c_{\mathscr{D}}^{p,n}(\xi)=)\,AJ_{X}^{p,n}(\xi)\in J^{p,n}(X)_{\QQ}\cong \frac{H^{2p-n-1}(X,\CC)}{F^{p}H^{2p-n-1}(X,\CC)+H^{2p-n-1}(X,\QQ(p))}.
$$

If $X$ is a smooth algebraic $K3$ surface and $p=n=3$, then 
\begin{equation}\label{J}
AJ_{X}^{3,3}:H_M^{3}(X,\ZZ(3))\to H^{2}(X,\CC/\ZZ(3))= J^{3,3}(X)
\end{equation}
is computed by \begin{equation}\label{eqn tilde{R}}\tilde{R}_{\xi}:=R_{\xi}+(2\pi i)^3\delta_{\Gamma},\end{equation}since
$\Omega_{\xi}\in F^{3}\mathcal{D}^{3}(X)=\{0\}$. Let $U\subset X$
be a Zariski open set. Any precycle $\xi\in Z_{\RR}^{3}(U,3)$ is
a sum of components supported over divisors and components with generic
support. The simplest examples of the latter are elements of the form
\[
\{f_{1},f_{2},f_{3}\}_{U}:=\overline{\left\{ (x,f_{1}(x),f_{2}(x),f_{3}(x))\,|\, x\in U\backslash\cup|(f_{i})|\right\} },
\]
where $f_i\in \CC(X)^*$ and the bar denotes Zariski closure in $U\times\square^{3}$. One
can show that 
\[
R_{\{f_{1},f_{2},f_{3}\}}=\log(f_{1})\frac{df_{2}}{f_{2}}\wedge\frac{df_{3}}{f_{3}}+2\pi i\log(f_{2})\frac{df_{3}}{f_{3}}\delta_{T_{f_{1}}}+(2\pi i)^{2}\log(f_{3})\delta_{T_{f_{1}}\cap T_{f_{2}}}
\]
extends to a $2$-current on $X$ (even if the closure of $\{ f_1,f_2,f_3\}_U$ over
$X$ is not a precycle).

This has the following application to the general situation of $\S$\ref{subsec K3 of K3},
where $\xi=\{f_{1},f_{2},f_{3}\}_{U}=\jmath^{*}\Xi-\partial\mu$ for
$\Xi\in\ker(\partial)\subset Z_{\RR}^{3}(X,3)$ and $\mu\in Z_{\RR}^{3}(U,4)$,
under the assumption that $\cap_{j}f_{j}^{-1}(\RR_{<0})\cap U=\emptyset$.
Working modulo currents and chains supported on $D:=X\backslash U$,
formally applying \eqref{eqn KLM map} (and noting that $R_{\mu}$
extends to $X$) gives \begin{equation}\label{eqn formal KLM}\left((2\pi i)^{3}T_{\xi},0,R_{\xi}\right)+D\left((2\pi i)^{3}T_{\mu},0,\frac{1}{2\pi i}R_{\mu}\right)\equiv\left((2\pi i)^{3}T_{\Xi},0,R_{\Xi}\right),\end{equation}
while our assumption gives $T_{\xi}\equiv0$. For the chains, this
yields $T_{\Xi}=-\partial T_{\mu}+S_{D}$, where $S_{D}$ is a (closed)
1-chain supported on $D$; since $T_{\Xi}$ is exact, so is $S_{D}$
(on $X$), and we write $S_{D}=\partial\gamma$. For the currents,
\eqref{eqn formal KLM} gives 
\[
R_{\Xi}=R_{\xi}+\frac{1}{2\pi i}d[R_{\mu}]+(2\pi i)^{3}\delta_{T_{\mu}}+K_{D}
\]
for some $2$-current $K_{D}$ supported on $D$, so that (taking
$\Gamma=-T_{\mu}+\gamma$ in \eqref{eqn tilde{R}})
\[
\tilde{R}_{\Xi}=R_{\xi}+\frac{1}{2\pi i}d[R_{\mu}]+(2\pi i)^{3}\delta_{\gamma}+K_{D}
\]
gives a lift of $AJ_{X}^{3,3}(\Xi)$.

The key point is now that \emph{if $H_{1}(D)=\{0\}$, then we may
take $\gamma$ to be supported on $D$}, and up to exact currents
on $X$ and arbitrary currents supported on $D$, \begin{equation}\label{eqn how-to-compute}\tilde{R}_{\Xi}\equiv R_{\{f_{1},f_{2},f_{3}\}}.\end{equation}This
is precisely what occurs in $\S$\ref{subsec K3 of K3} with $X=X_{t}$,
$\Xi=\Xi_{t}=\Xi_{\ba}|_{X_{t}}$, $U=X_{t}^{\sim}$, $D=\mathcal{J}$,
and $\{f_{1},f_{2},f_{3}\}=\{x|_{X_{t}^{\sim}},y|_{X_{t}^{\sim}},z|_{X_{t}^{\sim}}\}$;
the assumption $T_{\{x,y,z\}}\equiv0$ holds for 
\[
t\notin\overline{\phi_{\ba}(\RR_{<0}^{\times3})}=\left[16,\infty\right].
\]
Writing 
\[
\overline{AJ}_{X_{t}}^{3,3}:=\pi_{var}\circ AJ_{X_{t}}^{3,3}:\, H_M^{3}(X_{t},\QQ(3))\to H_{var}^{2}(X_{t},\CC/\QQ(3)),
\]
\eqref{eqn how-to-compute} provides a \emph{well-defined} lift (for
$t\notin[16,\infty]\cup\{0,4\}$)
\[
\mathcal{R}_{t}:=\pi_{var}[\tilde{R}_{\Xi_{t}}]\in H_{var}^{2}(X_{t},\CC)
\]
of $\bar{\mathcal{R}}_{t}:=\overline{AJ}_{X_{t}}^{3,3}(\Xi_{t})$.
As the extension of $R_{var}^{2}(\pi_{\ba})_{*}\QQ$ across $t=0$
has only rank 1, and $\bar{\mathcal{R}}_{t}$ must extend through
$t=0$, we conclude part (i) of the
\begin{prop}
\label{prop nf}(i) $\mathcal{R}_{t}$ yields a holomorphic section
of the sheaf $\mathcal{O}\otimes R_{var}^{2}(\pi_{\ba})_{*}\CC$ over $\PP^{1}\backslash[16,\infty]\cup\{0,4\}$,
and is the unique such section lifting $\bar{\mathcal{R}}_{t}$ with
no monodromy about $t=0$ and $t=4$.

(ii) Writing $\delta_{t}:=t\frac{d}{dt}$ and $\nabla$ for the Gauss-Manin
connection, we have 
\[
\nabla_{\delta_{t}}\mathcal{R}_{t}=-[\omega_{t}],
\]
with $\omega_t:=Res_{X_t}\left(    \frac{\frac{dx}{x}\wedge \frac{dy}{y}\wedge      \frac{dz}{z}}{1-t^{-1}\phi_{\ba}} \right)$.
\end{prop}
\begin{proof}
(ii) follows at once from \cite[Cor. 4.1]{DoranKerr} (note $t_{DK}=t^{-1}$).
\end{proof}

%---------------------------------------------------------------------------
\subsection{\label{subsec reinterp of I}Reinterpreting the Feynman integral}

The term ``higher normal function'' has been used in several different
ways in the theory of algebraic cycles -- for instance, to describe
the section of $\cup_{t}J^{3,3}(X_{t})$ (i.e. the family of extension
classes \eqref{J}) associated to a family of higher cycles like $\Xi_{t}$.
Here we shall pair this section with a specific family of holomorphic
forms to get an actual function (Definition \ref{def hnf}).  We preface this with a brief discussion of the pairings used here and in later sections.

Let $X$ be a smooth projective surface, $[X]\in H_4(X,\QQ)$ its fundamental class, and 
\[
\int_{[X]} : H^4(X,\QQ) \to \QQ(0)
\]
the map (of Hodge type $(-2,-2)$) induced by pairing with $[X]$.  We can define a Poincar\'e pairing in one of two ways:
\[
\langle \; ,\;\rangle :\, H^2(X,\QQ) \times H^2(X,\QQ) \to H^4(X,\QQ) \buildrel{\int_{[X]}}\over{\to} \QQ(0) \, ;
\]
\[
\langle \; ,\;\rangle ' : \, H^2(X,\QQ)\times H^2(X,\QQ) \to H^4(X,\QQ) = \QQ(-2).
\]
While the second has type $(0,0)$,  we prefer to work with the first bracket.

We now turn to the main content of this subsection.

\begin{defn}\label{def hnf}
The (truncated) higher normal function associated to $\Xi_{\ba}$
is 
\[
V_{\ba}(t):=\langle\mathcal{R}_{t},[\tilde{\omega}_{t}]\rangle\in\mathcal{O}(U_{\ba}),
\]
where $\tilde{\omega}_{t}:=\frac{-1}{(2\pi i)^{2}t}\omega_{t}\in\Omega^{2}(X_{t})$
and $U_{\ba}\subset \PP^1\setminus \mathcal{L}_{\ba}=\PP^1\setminus \{0,4,16,\infty\}$ is the complement of the real segment $16<t<\infty$.
\end{defn}
Note that $V_{\ba}$ extends holomorphically across $t=4$
and $0$, since it pairs finite (in fact nonzero) homology resp.
cohomology classes ($[\omega_{t}']$ resp. $\mathcal{R}_{t}$) on
those singular fibers.
\begin{thm}
\label{thm Feynman =00003D HNF}$I_{\ba}(t)=V_{\ba}(t)$
on $U_{\ba}$.\end{thm}
\begin{proof}
Begin by noting that
\[
\tilde{\omega}_{t}=\frac{-1}{(2\pi i)^{2}}Res_{X_{t}}\left(\frac{\frac{dx}{x}\wedge\frac{dy}{y}\wedge\frac{dz}{z}}{t-\phi_{\ba}(x,y,z)}\right)=:Res_{X_{t}}\left(\tilde{\Omega}_{t}\right)
\]
so that (regarding $\tilde{\Omega}_{t}\in F^{3}D^{3}(\PP_{\Delta_{\ba}})$
as a 3-current)
\[
d\left[\tilde{\Omega}_{t}\right]=2\pi i\imath_{*}^{X_{t}}\tilde{\omega}_{t}.
\]
Furthermore, $R_{3}^{*}:=R_{\{x,y,z\}}$ extends to a $2$-current
on $\PP_{\Delta_{\ba}}$, and writing $\Omega_{3}^{*}:=\frac{dx}{x}\wedge\frac{dy}{y}\wedge\frac{dz}{z}$,
$T_{3}^{*}:=\overline{T_{x}\cap T_{y}\cap T_{z}}$, we have on $\PP_{\Delta_{\ba}}$\begin{equation}\label{eqn d[R_3^*]}d[R_3^*]=\Omega_3^* - (2\pi i)^3 \delta_{T_3^*} + K_{\mathbb{D}},\end{equation}where
$K_{\mathbb{D}}$ ($\in F^{1}D^{3}(\PP_{\Delta_{\ba}})$) is supported
on $\mathbb{D}_{\ba}$.

Now\footnote{The apparent sign change in the denominator (compare (\ref{e:I formula})) arises from the orientation of $T_3^*$ and the change of variables.} $I_{\ba}(t)=$
\[
=\int_{\RR_{<0}^{\times3}}\frac{\frac{dx}{x}\wedge\frac{dy}{y}\wedge\frac{dz}{z}}{t-\phi_{\ba}(x,y,z)}=-(2\pi i)^{2}\int_{T_{3}^{*}}\tilde{\Omega}_{t}
\]
\[
=-(2\pi i)^{2}\int_{\PP_{\Delta_{\ba}}}\delta_{T_{3}^{*}}\wedge\tilde{\Omega}_{t}.
\]
By \eqref{eqn d[R_3^*]}, this 
\[
=\frac{1}{2\pi i}\int_{\PP_{\Delta_{\ba}}}\left(d[R_{3}^{*}]-\Omega_{3}^{*}-K_{\mathbb{D}}\right)\wedge\tilde{\Omega}_{t}.
\]
Noting that $K_{\mathbb{D}}\wedge\tilde{\Omega}_{t}$ and $\Omega_{3}^{*}\wedge\tilde{\Omega}_{t}$
are zero by type, it becomes
\[
=\frac{1}{2\pi i}\int_{\PP_{\Delta_{\ba}}}d[R_{3}^{*}]\wedge\tilde{\Omega}_{t}
\]
\[
=\frac{1}{2\pi i}\int_{\PP_{\Delta_{\ba}}}R_{3}^{*}\wedge d[\tilde{\Omega}_{t}]
\]
\[
=\int_{\PP_{\Delta_{\ba}}}R_{3}^{*}\wedge\imath_{*}^{X_{t}}\tilde{\omega}_{t}
\]
\begin{equation}\label{eqn pre HNF}= \int_{X_{t}}R_{\left\{ x|_{X_{t}^{\sim}},y|_{X_{t}^{\sim}},z|_{X_{t}^{\sim}}\right\} }\wedge\tilde{\omega}_{t}.\end{equation}Finally,
the argument of \eqref{eqn how-to-compute}ff allows us to rewrite
this as
\[
=\int_{X_{t}}\tilde{R}_{\Xi_{t}}\wedge\tilde{\omega}_{t}=V_{\ba}(t).
\]

\end{proof}
Without the last step, \eqref{eqn pre HNF} would not pair two \emph{closed}
currents and would have no cohomological meaning. So the seemingly
bizarre criterion that $H_{1}(\mathcal{J})=\{0\}$, in the end, is
absolutely essential.

\medskip
To give an idea of the power of Theorem \ref{thm Feynman =00003D HNF},
we conclude this section with one of its basic consequences:  namely, an alternate proof of Theorem \ref{e:PF}.  The characterization of $I_{\ba}$ as a higher normal function can also be used to compute some special values, cf. $\S$\ref{subsec hnf analysis}.

For deriving the Picard-Fuchs equation, we shall modestly abuse notation and regard the family of forms as a section 
\[
\tilde{\omega}_{t}\in\Gamma\left(\PP^{1}\backslash\mathcal{L}_{\ba},\mathcal{O}\otimes R_{var}^{2}(\pi_{\ba})_{*}\CC\right) .
\]
Let $\nabla_{PF}$ be the operator on cohomology obtained from
$D_{PF}:=\mathcal L_{t}^{3}=\sum_{k=0}^{3}f_{k}(t)\frac{d^{k}}{dt^{k}}$
by replacing $\frac{d}{dt}$ by $\nabla_{t}:=\nabla_{\frac{d}{dt}}$,
so that by \cite[Prop. 8]{Verrill} $\nabla_{PF}\tilde{\omega}_{t}=0$.
Note that $f_{3}(t)=t^{2}(t-4)(t-16)$ and $f_{2}(t)=6t(t^{2}-15t+32)=\frac{3}{2}f_{3}'(t)$.
Introduce the \emph{Yukawa coupling} 
\[
\tilde{Y}(t):=\langle\tilde{\omega}_{t},\nabla_{t}^{2}\tilde{\omega}_{t}\rangle,
\]
which may be computed as follows. Observe that by type
$0=\langle\tilde{\omega}_{t},\nabla_{t}\tilde{\omega}_{t}\rangle$ implies
\[
0=\frac{d^{2}}{dt^{2}}\langle\tilde{\omega}_{t},\nabla_{t}\tilde{\omega}_{t}\rangle=\langle\tilde{\omega}_{t},\nabla_{t}^{3}\tilde{\omega}_{t}\rangle+3\langle\nabla_{t}\tilde{\omega}_{t},\nabla_{t}^{2}\tilde{\omega}_{t}\rangle,
\]
so that 
\[
\frac{d}{dt}\tilde{Y}(t)=\langle\tilde{\omega}_{t},\nabla_{t}^{3}\tilde{\omega}_{t}\rangle+\langle\nabla_{t}\tilde{\omega}_{t},\nabla_{t}^{2}\tilde{\omega}_{t}\rangle=\frac{2}{3}\langle\tilde{\omega}_{t},\nabla_{t}^{3}\tilde{\omega}_{t}\rangle
\]
implies
\[
 f_{3}(t)\tilde{Y}'(t)=\frac{2}{3}\langle\tilde{\omega}_{t},-f_{2}(t)\nabla_{t}^{2}\tilde{\omega}_{t}-f_{1}(t)\nabla_{t}\tilde{\omega}_{t}-f_{0}(t)\tilde{\omega}_{t}\rangle
\]
\[
=-f_{3}'(t)\tilde{Y}(t)
\]
that implies
\[
\tilde{Y}(t)=\frac{\kappa}{f_{3}(t)}\in\CC(t).
\]
We will see below in $\S$\ref{subsec mod pullback} that $\kappa=\frac{-24}{(2\pi i)^{2}}$. Assuming this, we conclude 
\begin{cor}
\label{Cor PFE}$D_{PF}\left(I_{\ba}(t)\right)=-24$.\end{cor}
\begin{proof}
By Prop. \ref{prop nf}(ii), 
\[
\nabla_t \mathcal{R}_t = (2\pi i)^2 \tilde{\omega}_t.
\]
Now  $I_{\ba}(t)=V_{\ba}(t)=\langle\mathcal{R}_{t},\tilde{\omega}_{t}\rangle,$
and 
\[
\frac{d}{dt}\langle\mathcal{R}_{t},\tilde{\omega}_{t}\rangle=(2\pi i)^{2}\langle\tilde{\omega}_{t},\tilde{\omega}_{t}\rangle+\langle\mathcal{R}_{t},\nabla_{t}\tilde{\omega}_{t}\rangle=\langle\mathcal{R}_{t},\nabla_{t}\tilde{\omega}_{t}\rangle
\]
\[
\frac{d^{2}}{dt^{2}}\langle\mathcal{R}_{t},\tilde{\omega}_{t}\rangle=(2\pi i)^{2}\langle\tilde{\omega}_{t},\nabla_{t}\tilde{\omega}_{t}\rangle+\langle\mathcal{R}_{t},\nabla_{t}^{2}\tilde{\omega}_{t}\rangle=\langle\mathcal{R}_{t},\nabla_{t}^{2}\omega_{t}\rangle
\]
by type (and Griffiths transversality~\cite{Griffiths}). Together with 
\[
\frac{d^{3}}{dt^{3}}\langle\mathcal{R}_{t},\tilde{\omega}_{t}\rangle=(2\pi i)^{2}\tilde{Y}(t)+\langle\mathcal{R}_{t},\nabla_{t}^{3}\tilde{\omega}_{t}\rangle,
\]
these give
\[
D_{PF}\langle\mathcal{R}_{t},\tilde{\omega}_{t}\rangle=\langle\mathcal{R}_{t},\nabla_{PF}\tilde{\omega}_{t}\rangle+(2\pi i)^{2}f_{3}(t)Y(t)
\]
\[
=(2\pi i)^{2}f_{3}(t)\tilde{Y}(t)=-24.
\]
\end{proof}
\begin{rem}
\label{rem Y}For later reference we note that $Y(t):=\langle\omega_{t},\nabla_{\delta_{t}}^{2}\omega_{t}\rangle=(2\pi i)^{4}t^{4}\tilde{Y}(t)$
$\implies$ $Y(\infty)=(2\pi i)^{4}\kappa$.
\end{rem}

%%%%%%%%%%%%%%%%%%%%%%%%%%%%%%%%%%%%%%%%%%%%%%%%%%%%%%%%%%%%%%%%%%%%%%
\section{A second computation of the three-banana integral: the Eisenstein symbol}\label{sec Eis symb}

As an application of the results in sections
\ref{sec Family of K3s} and \ref{sec HNF}, we will use $\cH_{\ba}$ to
pull back the toric three-banana cycle $\Xi_{\ba}\in H_M^{3}(\cX_{\ba},\QQ(3))$
to $\cX_{1}(6)^{+3}$. We then apply a correspondence to produce a
higher Chow cycle on a Kuga variety $\cE^{[2]}(6)$ (defined below), and recognize
this as an Eisenstein symbol in the sense of Beilinson \cite{Beilinson,DeningerScholl,DoranKerr}.
This will allow us to write the pullback $V\circ H_{\ba}$ of the
higher normal function (i.e. Feynman integral) as an elliptic trilogarithm, giving another proof of Theorem \ref{thm:Li3elliptic}.  

%----------------------------------------------------------------------------
\subsection{Higher normal functions of Eisenstein symbols}

For simplicity of exposition we shall restrict to the setting of Kuga
3-folds. We begin with an explanation of Beilinson's construction of
higher cycles ("Eisenstein symbols") on these 3-folds and their
relationship to Eisenstein series of weight 4.  Each such cycle gives
rise to a higher normal function over a modular curve (defined in
(\ref{eqn **3})), which turns out to be an Eichler integral of the
corresponding Eisenstein series.  The main result of this subsection,
Proposition \ref{prop modular hnf}, computes the $q$-expansion
(\ref{eqn *!4}) of this normal function.  In many cases it may be
rewritten in terms of trilogarithms (cf. theorems~\ref{thm:Li3elliptic}
and~\ref{thm:main}).  Everything in this subsection is general.  In $\S\S$\ref{subsec mod pullback}-\ref{sec:mainresult} we shall apply this general computation to our special case, by pulling back the three-banana cycle from $\X_{\ba}$ to the Kuga 3-fold and interpreting the result (up to Abel-Jacobi equivalence) as one of Beilinson's cycles.

\medskip
To describe these cycles, consider the elliptic modular surface $\cE(N):=(\ZZ^{2}\rtimes\Gamma(N))\backslash(\CC\times\mathfrak{H})$
over $Y(N)=\Gamma(N)\backslash\mathfrak{H}$, where $\Gamma(N)=\ker\{SL_{2}(\ZZ)\to SL_{2}(\ZZ/N\ZZ)\}$ and $N>3$. Its fibers are elliptic curves with 1-form $dz$ and standard Betti
1-cycles $\alpha=[0,1]$, $\beta=[0,\tau]$. By duality we may regard
$\alpha,\beta$ as defining $H^{1}$ classes, and write $[dz]=[\beta]-\tau[\alpha]$. 

\medskip
Let $\cE^{[2]}(N)\overset{\pi^{[2]}(N)}{\longrightarrow}Y(N)$ be
the self-fiber product of $\cE(N)$. There exists a semistable compactification
$\overline{\cE^{[2]}(N)}\to\overline{Y(N)}$ due to Shokurov \cite{Shokurov},  with
singular fibers $D^{[2]}(N)=\overline{\cE^{[2]}(N)}\backslash\cE^{[2]}(N)$.
Choose for each cusp $\sigma=\left[\frac{r}{s}\right]\in\kappa(N):=\overline{Y(N)}\backslash Y(N)$
an element $M_{\sigma}:=\left(\begin{array}{cc}
p & q\\
-s & r
\end{array}\right)\in SL_{2}(\ZZ)$. Define modular forms of weight $n$ for $\Gamma(N)$ by
\[
M_{k}(N):=\left\{ F\in\mathcal{O}(\mathfrak{H})\left|\begin{array}{ccc}
(i) & F(\tau)=\frac{F(\gamma(\tau))}{(c\tau+d)^{k}}=:F|_{\gamma}^{k} & (\forall\gamma\in\Gamma(N))\\
(ii) & r_{\sigma}(F):=\underset{\tau\to i\infty}{\lim}F|_{M_{\sigma^{-1}}}^{k}<\infty & (\forall\sigma\in\kappa(N))
\end{array}\right.\right\}.
\]
There is an isomorphism (\cite{Shokurov}, or \cite[Prop. 7.1]{DoranKerr})
\[
\begin{array}{cccc}
\Psi: & M_{4}(N) & \overset{\cong}{\to} & \Omega^{3}(\overline{\cE^{[2]}(N)})\langle\log D^{[2]}(N)\rangle\\
 & F(\tau) & \mapsto & (2\pi i)^{3}F(\tau)dz_{1}\wedge dz_{2}\wedge d\tau.
\end{array}
\]

Let $\Phi_{2}^{K}(N)$ denote the vector space of $K$-valued functions
on $(\ZZ/N\ZZ)^{2}$, with subspaces $\Phi_{2}^{K}(N)_{\circ}:=\ker\{\text{evaluation at }(\bar{0},\bar{0})\}$
and 
\begin{multline}\notag\Phi_{2}^{K}(N)^{\circ}:=\ker\{\text{augmentation}\}= \\
\{f: (\ZZ/N\ZZ)^2 \to K\ |\ \sum_{0\le m,n\le N-1}f(m,n)=0\}
\end{multline}
Assuming
$K\supset\QQ(\zeta_{N})$ ($\zeta_{N}:=e^{\frac{2\pi i}{N}}$), these
are exchanged by the finite Fourier transform
\[
\vf(m,n)\mapsto\hat{\vf}(\mu,\eta):=\sum_{(m,n)\in(\ZZ/N\ZZ)^{2}}\vf(m,n)\zeta_{N}^{\mu n-\eta m}.
\]
This allows us to define the $\QQ$-Eisenstein series $E_{4}^{\QQ}(N)$
by the image of the map 
\[
\begin{array}{ccccc}
E: & \Phi_{2}^{\QQ} & \to & M_{4}(N)\\
 & \vf & \mapsto & E_{\vf}(\tau) & :=-\displaystyle{\frac{3}{(2\pi i)^{4}}\sum_{(m,n)\in\ZZ^{2}\backslash\{(0,0)\}}\frac{\hat{\vf}(m,n)}{(m\tau+n)^{4}}}.
\end{array}
\]
The horospherical maps
\[
\begin{array}{ccccc}
\mathscr{H}_{\sigma}: & \Phi_{2}^{\QQ}(N)^{\circ} & \to & \QQ\\
 & \vf & \mapsto & \mathscr{H}_{\sigma}(\vf) & :=\frac{1}{8}\sum_{a=0}^{N-1}B_{4}\left(\frac{a}{N}\right)\cdot((\pi_{\sigma})_{*}\vf)(a)
\end{array}
\]
record the ``values'' $\lim_{\tau\to i\infty}E_{\vf}(\tau)|_{M_{\sigma}^{-1}}^{4}\,(=\mathscr{H}_{\sigma}(\vf))$
of the Eisenstein series $E_{\vf}$ at the cusps. Here $\pi_{\sigma}:(\ZZ/N\ZZ)^{2}\twoheadrightarrow\ZZ/N\ZZ$
sends $(m,n)=a(p,q)+b(-s,r)\mapsto a$, while $(\pi_{\sigma})_{*}$
sums along fibers of $\pi_{\sigma}$, and $B_{4}(x)
%=\sum_{j=0}^{4}{4 \choose j}B_{j}x^{4-j}
=x^{4}-2x^{3}+x^{2}-\frac{1}{30}$
is the fourth Bernoulli polynomial. Alternatively, one has
\[
\mathscr{H}_{\sigma}(\vf)=-\frac{6}{(2\pi i)^{4}}L(\iota_{\sigma}^{*}\hat{\vf},4)
\]
where $\iota_{\sigma}:\ZZ/N\ZZ\hookrightarrow(\ZZ/N\ZZ)^{2}$ sends
$a\mapsto a(-s,r)$ and
\[
L(\phi,n):=\sum_{k\geq 1}\frac{\phi(k)}{k^n}.
\]

To construct the cycles, let $U\subset\cE(N)$ {[}resp. $U^{[2]}\subset\cE^{[2]}(N)${]}
be the complement of the $N^{2}$ {[}resp. $N^{4}${]} $N$-torsion
sections over $Y(N)$. Fix $\vf\in\Phi_{2}^{\QQ}(N)^{\circ}$, and
(thinking of it as a $\QQ$-divisor supported on $\cE(N)\backslash U$)
let $m_{\alpha}\in\QQ$ and $f_{\alpha1},f_{\alpha2},f_{\alpha3}\in\mathcal{O}^{*}(U)$
satisfy $\sum_{\alpha}m_{\alpha}(f_{\alpha1})*(f_{\alpha2})*(f_{\alpha3})=\vf$
(Pontryagin product). (Here $(f_{\alpha i})$ is the divisor of $f_{\alpha i}$, the divisor being viewed as a function on $(\ZZ/N\ZZ)^2$, and the Pontryagin product of two functions on a finite abelian group is defined by $(f*g)(a)=\sum_{b+c=a}f(b)g(c)$. 
The group\footnote{$D_4$ denotes the dihedral group of order $8$.} $\mathcal{G}:=D_{4}\ltimes(\ZZ/N\ZZ)^{4}$
acts on $H_M^{3}(U^{[2]},\QQ(3))$, and the $\mathcal{G}$-symmetrization of $$\sum_{\alpha} m_{\alpha}\{f_{\alpha1}(-z_{1}),f_{\alpha2}(z_{1}-z_{2}),f_{\alpha3}(z_{2})\}$$
extends to a cycle in $H_M^{3}(\cE^{[2]}(N),\QQ(3))$ (cf. \cite[sec. 7.3.4]{DoranKerr}). By abuse of
notation we shall call it $\Z_{\vf}$, since its fiberwise $AJ^{3,3}$-classes \begin{equation}\label{eqn *36} \cR_{\vf}(y)\in H_{var}^{2}\left(\pi^{[2]}(N)^{-1}(y),\CC/\QQ(3)\right),\;\;\; y\in Y(N),
\end{equation} depend only on $\vf$ -- indeed, only on the $\{\mathscr{H}_{\sigma}(\vf)|\sigma\in\kappa(N)\}$
-- and not the choice of $\{f_{\alpha i}\}$ \cite[Cor 9.1]{DoranKerr}.%
\footnote{The reader is warned of the typo ``surjective'' for ``injective''
in the statement of {[}loc. cit., Lemma 9.1(ii){]}.%
}

\medskip
The connection between the cycle $\Z_{\vf}$ and the Eisenstein series $E_{\vf}$ comes about as follows.  First, by using  the moving lemma \cite{Bloch1994} and log complexes of currents, it is possible to extend the $(T,\Omega,R)$ calculus of $\S$4.2 to the quasi-projective setting (\cite[$\S$5.9]{KerrLewisSMS},\cite[$\S$3.1]{KerrLewis}).  In particular, the fundamental class of $\Z_{\vf}$ (i.e. the image of $c_{\mathscr{D}}^{3,3}(\Z_{\vf})$ in $Hg^{3,3}(\E^{[2]}(N))_{\QQ})$ is computed by the holomorphic $(3,0)$-form $\Omega_{\Z_{\vf}}$.  According to a result of Beilinson (in the form of \cite[Thm. 8.1]{DoranKerr}), we have
\begin{equation}\label{eqn *3}\Omega_{\Z_{\vf}} = \Psi(E_{\vf})=\{(2\pi i)^3 E_{\vf}(\tau) dz_1 \wedge dz_2\}\otimes d\tau . \end{equation} It follows that $\cR_{\vf}(y)$ is given (up to an important ``constant of integration'')
by the Gauss-Manin integral of \eqref{eqn *3}; that is, \eqref{eqn *3} is $\nabla \mathcal{R}_{\vf}$.

Define the associated higher normal function by\footnote{Note:  \emph{a priori} this just uses the Poincar{\'e} pairing $H^2(E_{\tau}^{\times 2},\CC)^{\otimes 2} \to \CC$ on each fiber.  However, it is better to think of $[dz_1 \wedge dz_2]$ as a class in $H_2(E_{\tau}^{\times 2},\CC)$ by Poincar{\'e} duality and (\ref{eqn **3}) as pairing $H^2\times H_2 \to \CC$, since this approach will extend across the singular fibers of $\overline{\E^{[2]}(N)}$ over cusps $\sigma$ for which $\mathscr{H}_{\sigma}(\vf)=0$.}\begin{equation}\label{eqn **3}V_{\vf}(\tau):=\langle \tilde{\cR}_{\vf}([\tau]),[dz_1 \wedge dz_2]\rangle ,\end{equation} where for now
$\tilde{\cR}_{\vf}$ is an indeterminate lift of $\cR_{\vf}$ to $\mathcal{O}\otimes R_{var}^{2}\pi^{[2]}(N)_{*}\CC$.
Arguing as in the proof of Corollary \ref{Cor PFE} above, and noting
$\nabla_{\tau}^{3}[dz_{1}\wedge dz_{2}]=0$, 
\[
\frac{d^{3}}{d\tau^{3}}V_{\vf}(\tau)=\frac{d^{2}}{d\tau^{2}}\left\langle \cR_{\vf},\nabla_{\tau}[dz_{1}\wedge dz_{2}]\right\rangle 
\]
\[
=\frac{d}{d\tau}\left\langle \cR_{\vf},\nabla_{\tau}^{2}[dz_{1}\wedge dz_{2}]\right\rangle 
\]
\[
=\left\langle (2\pi i)^{3}E_{\vf}(\tau)[dz_{1}\wedge dz_{2}],2[\alpha_{1}\times\alpha_{2}]\right\rangle 
\]
\begin{equation}\label{eqn ***4}=-2(2\pi i)^{3} E_{\vf}(\tau).\end{equation}That
is, $ $$V_{\vf}$ is an Eichler integral of $E_{\vf}$.
This leads to the following result, which is closely related to \cite[Prop. 9.2]{DoranKerr}.
\begin{prop}
\label{prop modular hnf}Assume for simplicity that $\hat{\vf}(m,n)=\hat{\vf}(-m,-n)$.
Then up to a $\QQ(3)$-period $(2\pi i)^{3}Q_{0}+(2\pi
i)^{2}Q_{1}\log q+(2\pi i)Q_{2}(\log q)^2$
($Q_{i}\in\QQ$),\begin{equation}\label{eqn
    *!4}V_{\vf}(q)\equiv\begin{array}[t]{c}\frac{2}{(2\pi
      i)^{4}}L\left(\iota_{i\infty}^{*}\hat{\vf},4\right)(\log q)^3+\frac{1}{N}L\left((\pi_{i\infty})_{*}\hat{\vf},3\right)\\+\frac{2}{N}\sum_{M\geq1}q^{\frac{M}{N}}\sum_{d|M}\frac{1}{d^{3}}\sum_{a\in\ZZ/N\ZZ}\zeta_{N}^{\frac{aM}{d}}\hat{\vf}(d,a).\end{array}\end{equation}(Note
that $\frac{1}{(2\pi i)^{4}}L(\iota_{i\infty}^{*}\hat{\vf},4)\in\QQ$.)\end{prop}
\begin{proof}
By a classical result (cf. \cite{Gunning}), we have\begin{equation}\label{eqn s4}E_{\vf}(\tau)=\mathscr{H}_{[i\infty]}(\vf)-\frac{1}{N^{4}}\sum_{M\geq1}q^{\frac{M}{N}}\sum_{r|M}r^{3}\sum_{a\in\ZZ/N\ZZ}\zeta_{N}^{ar}\hat{\vf}\left(\frac{M}{r},a\right).\end{equation}In
accordance with \eqref{eqn ***4}, we must take three indefinite integrals
of $-2(2\pi i)^{3}E_{\vf}(\tau)$ with respect to $d\tau=\frac{1}{2\pi i}d\log q$,
i.e. of $-2E_{\vf}(\tau)$ with respect to $d\log q$. Applying
this to the second term of \eqref{eqn s4} gives
\begin{equation}\label{e:d3E}
\frac{2}{N}\sum_{M\geq1}\frac{q^{\frac{M}{N}}}{M^{3}}\sum_{r|M}r^{3}\sum_{a\in\ZZ/N\ZZ}\zeta_{N}^{ar}\hat{\vf}\left(\frac{M}{r},a\right),
\end{equation}
and replacing $r$ by $d=\frac{M}{r}$ recovers the sum in \eqref{eqn
  *!4}.
Doing the same to $\mathscr{H}_{[i\infty]}(\vf)$ would give
$-\frac{1}{3}\mathscr{H}_{[i\infty]}(\vf)(\log q)^3$
plus an arbitrary quadratic polynomial in $\log q$. The more precise
stated result follows at once from \cite[(9.29)]{DoranKerr},%
\footnote{Note that while this formula is derived in {[}op. cit.{]} for $\vf$
of the form $\frac{1}{N}\pi_{i\infty}^{*}\vf'$, any $\vf$ is of
this form modulo $\ker(\mathscr{H}_{[i\infty]})\subset\Phi_{2}^{\QQ}(N)^{\circ}$.%
} which is based on the delicate fiberwise $AJ^{3,3}$ computation
for $\Z_{\vf}$ carried out in $\S$9.2 of {[}op. cit.{]}. 
\end{proof}
The connection of this formula to trilogarithms arises as follows.
Define\begin{equation}\label{eqn *s6}\widehat{Li}_{3}(x):=\sum_{k\geq1}Li_{3}(x^{k})=\sum_{k\geq 1}\sum_{\delta\geq 1}\frac{x^{k\delta}}{\delta^3}=\sum_{m\geq1}x^{m}\sum_{\delta|m}\frac{1}{\delta^{3}},\end{equation}and
suppose that we can write
\[
\hat{\vf}=\sum_{\tiny\begin{array}[t]{c}
\alpha|N\\
\beta|N
\end{array}}\mu_{\alpha\beta}\psi_{\alpha,\beta}
\]
where 
\begin{equation}\label{e:psidef}
\psi_{\alpha,\beta}(m,n):=\left\{ \begin{array}{c}
1,\text{ if }\alpha|m\text{ and }\beta|n\\
0,\text{ otherwise}\;\;\;\;\;\;\;\;\;\;
\end{array}\right..
\end{equation}
In the $\sum_{M\geq1}$ term 
\[
\frac{2}{N}\sum_{\alpha,\beta}\mu_{\alpha\beta}\sum_{M\geq1}q^{\frac{M}{N}}\sum_{d|M}\frac{1}{d^{3}}\sum_{a\in\ZZ/N\ZZ}\zeta_{N}^{\frac{aM}{d}}\psi_{\alpha,\beta}\left(d,a\right)
\]
of \eqref{eqn *!4}, the sum $\sum_{a\in\ZZ/N\ZZ}\zeta_{N}^{\frac{aM}{d}}\psi_{\alpha,\beta}(d,a)$
is zero unless $\alpha|d$ and $\frac{N}{\beta}|\frac{M}{d}$ (implies $\frac{\alpha N}{\beta}|M$),
in which case it is $\frac{N}{\beta}$. So after putting $M=\frac{\alpha N}{\beta}m$
and $d=\alpha\delta$, the last displayed expression becomes
\[
=\frac{2}{N}\frac{N}{\beta}\sum_{\alpha,\beta}\mu_{\alpha\beta}\sum_{m\geq1}\left(q^{\frac{\alpha}{\beta}}\right)^{m}\sum_{\delta|m}\frac{1}{\delta^{3}\alpha^{3}},
\]
which (upon putting $k=\frac{m}{\delta}$)\begin{equation}\label{eqn ss6}=2\sum_{\alpha,\beta}\frac{\mu_{\alpha\beta}}{\beta\alpha^{3}}\widehat{Li}_{3}\left(q^{\frac{\alpha}{\beta}}\right).\end{equation}
%--------------------------------------------------------------------------
\subsection{\label{subsec mod pullback}Modular pullback of the three-banana cycle}
In this subsection, we identify the pullback of $\Xi_{\ba}$ to $\E^{[2]}(6)$ as an Eisenstein symbol.  We begin with a general statement.

Let $\mathbb{X}\overset{\rho}{\to}\PP^{1}$ be a 1-parameter family
of anticanonical hypersurfaces in a toric Fano 3-fold $\PP_{\Delta}$,
with smooth total space obtained by a blow-up $\mathbb{X}\overset{\beta}{\twoheadrightarrow}\PP_{\Delta}$,
and $\beta(\mathrm{X}_{0}):=\beta(\rho^{-1}(0))=\DD_{\Delta}:=\PP_{\Delta}\backslash(\CC^{*})^{3}$.
Suppose we have a higher cycle $\bar{\Xi}\in H_M^{3}(\mathbb{X}\backslash\mathrm{X}_{0},\QQ(3))$
with $\partial T_{\bar{\Xi}}$ \eqref{eqn KLM map} the integral generator of $H_{2}(\mathrm{X}_{0},\ZZ)$,\footnote{Alternatively, $Res(\bar{\Xi})\in H^4_{M,X_0}(\mathbb{X},\QQ(3))$ has cycle class in $H^4_{X_0}(\mathbb{X},\QQ(3))$, $\frac{1}{(2\pi i)^3}$ of which integrally generates $H_2(X_0,\ZZ)$.}
and a rational map (or even a correspondence) \[\xymatrix{
\overline{\cE^{[2]}(N)} \ar @{-->} [r]^{\Theta} \ar [d]_{\overline{\pi^{[2]}(N)}} & \mathbb{X} \ar [d]
\\
\overline{Y(N)} \ar @{->>} [r]^{\sf{H}} & \PP^1.
}\]
Let $\Theta: \E^{[2]}(N) \dashrightarrow \X$ be the restriction to the complement of the singular fibers, and $\Xi\in H^3_M(\X,\QQ(3))$ the restriction of $\bar{\Xi}$.  Defining coefficients $r_{\sigma}(\Xi)\in\QQ$ by
\[
\bar{\Theta}^{*}(\mathrm{X}_{0})=\sum_{\sigma\in\kappa(N)}r_{\sigma}(\Xi)\cdot\overline{\pi^{[2]}(N)}^{-1}(\sigma),
\]
we have
\begin{prop}
Modulo cycles with trivial fiberwise $AJ^{3,3}$, we have
\[
\Theta^{*}\Xi=\Z_{\vf}\in H_M^{3}(\cE^{[2]}(N),\QQ(3))
\]
for any $\vf\in\Phi_{2}^{\QQ}(N)^{\circ}$ with $\mathscr{H}_{\sigma}(\vf)=r_{\sigma}(\Xi)$
$(\forall\sigma\in\kappa(N))$.
\end{prop}
\begin{proof} 
This is immediate from the fact that (\ref{eqn *36}) depends only on the ``residues'' $\mathscr{H}_{\sigma}(\vf)$. \end{proof}
To apply this general statement to the three-banana cycle $\Xi_{\ba}$ constructed in $\S$\ref{subsec K3 of K3}, we begin by analyzing the transformation of the family
of holomorphic forms $\omega:=\{\omega_{t}\}\in\Gamma(\PP^{1}\backslash\mathcal{L}_{\ba},(\pi_{\ba})_{*}\Omega_{\pi_{\ba}}^{2})$
(cf. \eqref{eqn **23}) under the correspondence\[\xymatrix{
\cX_{\ba} \ar [d]_{\pi_{\ba}} & \cX_1(6)^{+3} \ar [l]^{\cong}_{\cH_{\ba}} \ar [rd] & \tilde{\cX} \ar [d] \ar @{->>} [r]^{p_1}_{2:1} \ar @{->>} [l]_{p_2}^{2:1} & '\cE_1^{[2]}(6)^{+3} \ar [ld] & \cE_1^{[2]}(6) \ar @{->>} [l]_{\bar{J}^{[2]}_3} \ar [d] \ar @/_2pc/ @{-->} [llll]_{\theta}
\\
\PP^1 \setminus \mathcal{L}_{\ba} & & Y_1(6)^{+3} \ar [ll]^{\cong}_{H_{\ba}} & & Y_1(6). \ar @{->>} [ll]^{2:1}
}\]between $\cE_{1}^{[2]}(6)$ and $\cX_{\ba}$. Here $\tilde{\cX}$
is described in $\S$\ref{subsec misc}, and 
\[
J_{3}^{[2]}:\cE_{1}^{[2]}(6):=\cE_{1}(6)\underset{Y_{1}(6)}{\times}\cE_{1}(6)\overset{3:1}{\twoheadrightarrow}{}'\cE_{1}^{[2]}(6)
\]
is the map over $Y_{1}(6)$ defined by 
\[
\left(\tau;[z_{1}]_{\tau},[z_{2}]_{\tau}\right)\longmapsto\left(\tau;[z_{1}]_{\tau},\left[\frac{-z_{2}}{2\tau+1}\right]_{\alpha_{3}(\tau)}\right),
\]
and $\bar{J}_{3}^{[2]}$ its composition with the quotient by $I_{3}^{[2]}$
(cf. $\S$\ref{subsec Verrill fam}).

By (\ref{e:realperiod}), the period of $\cH_{\ba}^{*}\omega$ over the minimal invariant cycle
in $H_{2,\ZZ}^{tr}$ about $q=0$ ($t=\infty$) limits to $(2\pi i)^{2}$. Applying $p_{2}^{*}$, $(p_{1})_{*}$, $(\bar{J}_{3}^{[2]})^{*}$
multiplies this by $2$, $2$, resp. $3$. Writing $\theta^{*}:=(\bar{J}_{3}^{[2]})^{*}(p_{1})_{*}p_{2}^{*}$,
it follows that\begin{equation}\label{eqn *8}\theta^* \omega \equiv 12(2\pi i)^2 dz_1 \wedge dz_2\;\;\text{mod }\mathcal{O}(q)\end{equation}hence
(noting $\delta_{q}=\frac{1}{2\pi i}\frac{\partial}{\partial\tau}$)
\[
Y(\infty)=\left\langle \omega,\nabla_{\delta_{t}}^{2}\omega\right\rangle |_{t=\infty} 
\]
\[
=\frac{1}{12}\left\langle \theta^{*}\omega,\nabla_{\delta_{q}}^{2}\theta^{*}\omega\right\rangle |_{q=0}
\]
\[
=\frac{{12}^{2}(2\pi i)^{4}}{12(2\pi i)^{2}}\left\langle [dz_{1}\wedge dz_{2}],\nabla_{\tau}^{2}[dz_{1}\wedge dz_{2}]\right\rangle 
\]
\[
=-24(2\pi i)^{2},
\]
where $Y(t)$ was defined in Remark \ref{rem Y}.  In fact, by that remark we now have $\kappa=\frac{-24}{(2\pi i)^{2}}$ as claimed in the proof of Corollary \ref{Cor PFE}.

\medskip
Turning to the computation of the $\{r_{\sigma}(\Xi_{\ba})\}$, we
take $\Theta$ to be the composition of $\theta$ with the base change
over $Y(6)\twoheadrightarrow Y_{1}(6)$. We examine the pullback by $\Theta$ of the $(3,0)$ form $\Omega_{\Xi_{\ba}}$ which computes the fundamental class of the cycle.  By (\ref{eqn dR=00003DO-T+Res}) and Proposition \ref{prop nf}(ii), 
$\Omega_{\Xi_{\ba}}=-\omega\wedge\frac{dt}{t}\in\Omega^{3}(\cX_{\ba})$,
and \eqref{eqn *8} now gives 
\[
\Omega_{\Theta^{*}\Xi_{\ba}}=-\Theta^{*}\Omega_{\Xi_{\ba}}=\Theta^{*}\omega\wedge \text{dlog} H_{\ba}(\tau)
\]
\[
\equiv 12(2\pi i)^{3}dz_{1}\wedge dz_{2}\wedge d\tau\;\;\text{mod }\mathcal{O}(q),
\]
which implies at once that $r_{[i\infty]}(\Xi_{\ba})=12$. (Note the
consistency with \eqref{eqn *3} and \eqref{eqn s4}.) Now the (partial)
pullback of $\Xi_{\ba}$ to $'\cE_{1}^{[2]}(6)$ is invariant under
$I_{3}^{[2]}$; a calculation as in \cite[sec. 8.2.2]{DoranKerr}
shows that consequently $r_{[-\frac{1}{2}]}(\Xi_{\ba})=r_{[\alpha_{3}(i\infty)]}(\Xi_{\ba})=-\frac{r_{[i\infty]}(\Xi_{\ba})}{3^{2}}=-\frac{4}{3}.$
In fact, writing $\Omega_{\Theta^{*}\Xi_{\ba}}=(2\pi i)^{3}E_{\ba}(\tau)dz_{1}\wedge dz_{2}\wedge d\tau$,
we have $E_{\ba}(\tau)\in M_{4}(\Gamma_{1}(6)^{+3})$; and $r_{\sigma}(\Xi_{\ba}):\kappa(6)\to\QQ$
is the pullback of the function on $\kappa_{1}(6)=\left\{ [i\infty],[0],[\frac{1}{2}],[\frac{1}{3}]\right\} $
taking the respective values $12,0,-\frac{4}{3},0$. (Under $\kappa(6)\twoheadrightarrow\kappa_{1}(6)$,
the preimage of $[i\infty]$ resp. $[\frac{1}{2}]$ is $\{[i\infty]\}$
resp. $\left\{ [\frac{1}{2}],[\frac{3}{2}],[-\frac{1}{2}]\right\} $.)
Using the formula for $\mathscr{H}_{\sigma}$, one then finds that
the function $\vf_{\ba}$ on $(\ZZ/N\ZZ)^{2}$ with Fourier transform
\begin{equation}\label{eqn
    s10}\hat{\vf}_{\ba}(m,n):=\left\{ \begin{array}{cc}-2^{6}3^{5}/5,
      & (m,n)\equiv(0,\pm1) \mod 6\\2^{6}3^{3}/5, &
      (m,n)\equiv (\pm2,\pm1\text{ or }3)\mod 6\\0, &
      \text{otherwise}\end{array}\right.
\end{equation}
satisfies
$\mathscr{H}_{\sigma}(\vf_{\ba})=r_{\sigma}(\Xi_{\ba})$.

Finally we determine the pullbacks of $\omega$ and
$\tilde{\omega}$. 
In \cite{Verrill}, it is shown that
  $\varpi_1(\tau)=(\eta(2\tau)\eta(6\tau))^4\,(\eta(\tau)\eta(3\tau))^{-2}$
is the $H_{\ba}$-pullback of a solution to $D_{PF}$; so $\Theta^{*}(\tilde{\omega})=C\cdot \varpi_1(\tau)dz_{1}\wedge dz_{2}$
for some constant $C$. But then $\Theta^{*}(\omega)=-(2\pi i)^{2}C\varpi_1(\tau)H_{\ba}(\tau) dz_{1}\wedge dz_{2}$
and by \eqref{eqn *8} $C=12$.
\begin{rem}
One further immediate consequence is that
$E_{\ba}(\tau)=12{\varpi_1(\tau)\over 2\pi i}
{dH_{\ba}^{-1}(\tau)\over d\tau}=12+24q-168q^{2}+\cdots$;
but the equality $E_{\ba}(\tau)=E_{\vf_{\ba}}(\tau)$ is more useful for
us as it allows us to apply Proposition \ref{prop modular hnf} and
get the ``constant of integration'' right. 
\end{rem}

%-----------------------------------------------------------------------------
\subsection{The main result}\label{sec:mainresult}

Recall that $V_{\ba}(t)=\left\langle \cR_{t},[\tilde{\omega}_{t}]\right\rangle $.
Putting everything together, we arrive at the
\begin{thm}\label{thm:main}
Up to a $\QQ(3)$-period $(2\pi i)^{3}Q_{0}+(2\pi i)^{2}Q_{1}\tau+(2\pi i)Q_{2}\tau^{2}$
$(Q_{i}\in\QQ)$, we have \textup{$\frac{V_{\ba}(H_{\ba}(\tau))}{\varpi_1(\tau)}=$}
\[
-4(\log q)^3+16\zeta(3)-16\left\{ 2\widehat{Li}_{3}(q^{6})-\widehat{Li}_{3}(q^{3})-6\widehat{Li}_{3}(q^{2})+3\widehat{Li}_{3}(q)\right\},
\]
where $\widehat{Li}_{3}(x):=\sum_{k\geq1}Li_{3}(x^{k})$.
\end{thm}
\begin{proof}
First notice that
\[
V_{\ba}=\langle\cR,\tilde{\omega}\rangle=\frac{1}{12}\langle\Theta^{*}\cR,\Theta^{*}\tilde{\omega}\rangle
\]
\[
=\frac{1}{12}\langle\cR_{\vf_{\ba}},12\varpi(\tau)[dz_{1}\wedge dz_{2}]\rangle
\]
\[
=\varpi_1(\tau)\langle\cR_{\vf_{\ba}},[dz_{1}\wedge dz_{2}]\rangle
\]
so that $V_{\ba}=\varpi_1\, V_{\vf_{\ba}}$. The leading term in
\eqref{eqn *!4} is $-\frac{2}{3!}\mathscr{H}_{[i\infty]}(\vf_{\ba})(\log q)^3=-4(\log q)^3$.
For the constant term we compute
\[
\left((\pi_{[i\infty]})_{*}\hat{\vf}_{\ba}\right)(n)=\left\{ \begin{array}{cc}
-2^{7}3^{5}/5, & n\equiv 0\mod 6\\
2^{6}3^{4}/5, & n\equiv \pm2\mod 6\\
0, & \text{otherwise}
\end{array}\right.
\]
\[
\implies\;\frac{1}{6}L\left((\pi_{[i\infty]})_{*}\hat{\vf}_{\ba},3\right)=\frac{1}{6}\sum_{n\geq1}\frac{(\pi_{[i\infty]})_{*}\hat{\vf}_{\ba}(n)}{n^{3}}
\]
\[
=\frac{1}{6}\cdot\frac{-2\cdot6^{5}}{5}\left\{ \frac{7}{3}\cdot\frac{1}{6^{3}}\zeta(3)-\frac{1}{3}\cdot\frac{1}{2^{3}}\zeta(3)\right\} 
\]
\[
=\frac{-2\cdot6^{4}}{5}\cdot\frac{-20}{3\cdot6^{3}}\zeta(3)=16\zeta(3).
\]
Finally, we write using the character $\psi_{a,b}$ defined
  in~\eqref{e:psidef}
\[
\hat{\vf}_{\ba}=\frac{-3^{3}2^{6}}{5}\left\{ 10\psi_{6,1}-10\psi_{6,2}-9\psi_{6,3}+9\psi_{6,6}-\psi_{2,1}+\psi_{2,2}\right\} .
\]
Substituting this into \eqref{eqn ss6} gives the remaining terms
in the result.\end{proof}
%%%%%%%%%%%%%%%%%%%%%%%%%%%%%%%%%%%%%%%%%%%%%%%%%%%%%%%%%%%%%%%%%%
\section{Foundational Results via Hodge Theory}
\label{sec:hodge}

The methodology of sections \ref{sec HNF} and \ref{sec Eis symb}
involving higher Chow cycles and currents is delicate. Care is needed
to avoid bad position and ill-defined multiplication of currents. The
purpose of this section is to give a general Hodge-theoretic context
for proving basic results about periods in related situations. In the
context of this paper, arguments using currents are required to lift
the Milnor symbol regulator, defined a priori only on $X_t^*$, over
all of $X_t$. Arguments in this section only give results upto periods
over $X_t^*$. Because $X_t\backslash X_t^*$ in our case is a union of
rational curves, it turns out that these extra periods associated to
$2$-chains on $X_t$ relative to $X_t\backslash X_t^*$ are themselves of motivic interest. This point will be discussed briefly at the end of the section. 
\subsection{Some lemmas}\label{subsect lems}

In this subsection, we give an elementary but useful application of Verdier duality  (Lemma \ref{dualitylemma}) -- also known, thanks to R. MacPherson, as ``red-green duality'' (cf. Remark \ref{RGremark}).We work throughout with sheaves for the complex topology.

\begin{lem}Let $P$ be a smooth, quasi-projective variety over $\CC$, and let $X, Y \subset P$ be closed subvarieties. Consider the diagram
\begin{equation}\label{pushmepullyou diag}\begin{CD} P\setminus (X\cup Y) @> j'>> P\setminus X \\
@VV k' V @VV k V \\
P\setminus Y @> j>> P @< i<< Y
\end{CD}
\end{equation}
Assume that for every point $z\in X\cap Y$ there exists a ball $B$ about $z$ in $P$ and a decomposition $B=B_X\times B_Y$ (where $B_X, B_Y$ are smaller dimensional balls). Assume further there exist analytic subvarieties $X' \subset B_X$ and $Y'\subset B_Y$ such that $X\cap B = X'\times B_Y$ and $Y\cap B = B_X\times Y'$. Then viewed as maps on the respective derived categories of sheaves for the complex topology (in keeping with modern usage we write e.g. $j_*$ in place of $Rj_*$) we have
\begin{equation}\label{pushmepullyou eqn} j_!k'_*\QQ_{P\setminus  (X\cup Y)} = k_*j_!'\QQ_{P\setminus (X\cup Y)}.
\end{equation}
\end{lem}
\begin{proof} We have 
\begin{equation} j^*k_*j_!'\QQ_{P\setminus (X\cup Y)} = k_*'\QQ_{P\setminus (X\cup Y)}.
\end{equation}
Since $j_!$ is left adjoint to $j^*$ we deduce the existence of a map (extending the identity map on $P\setminus (X\cup Y)$) from left to right in \eqref{pushmepullyou eqn}.  To check that this map is a quasi-isomorphism is a local problem. The assertion is evident except  at points of $X\cap Y \subset P$. By assumption, near such a point our diagram \eqref{pushmepullyou diag} looks like
\begin{equation}\begin{CD}(B_X\setminus X')\times (B_Y\setminus Y') @>>> (B_X\setminus X')\times B_Y \\
@VVV @VVV \\
B_X\times (B_Y\setminus Y') @>>> B_X\times B_Y @<<< B_X\times Y'. 
\end{CD}
\end{equation}
The assertion is now clear by a variant of the Kunneth formula. Namely, both sides are identified with
\begin{equation}(k_{B_X*}\QQ_{B_X\setminus X'}) \otimes (j_{B_Y!}\QQ_{B_Y\setminus Y'}).
\end{equation}
\end{proof}
\begin{rem}The hypotheses of the lemma are satisfied if $X\cup Y \subset P$ is a normal crossings divisor locally at points of $X\cap Y$. 
\end{rem}
\begin{lem}\label{lem2}Let notation be as above and write $Z=X\cap Y$. We have
\begin{equation}H^*(P\setminus X, Y\setminus Z;\QQ) \cong H^*(P, j_!k_*'\QQ) 
\end{equation}
\end{lem}
\begin{proof}We have 
\begin{equation}j_!k_*'\QQ_{P\setminus (X\cup Y)} = j_!j^*k_*\QQ_{P\setminus X}.
\end{equation}
The functorial distinguished triangle of sheaves on $P$
$$j_!j^*\sS \to \sS \to i_*i^*\sS \xrightarrow{+1} \ldots$$
yields a distinguished triangle
\begin{equation}\label{16}j_!k_*'\QQ_{P\setminus (X\cup Y)} \to k_*\QQ_{P\setminus X} \to i_*i^*k_*\QQ_{P\setminus X}.
\end{equation}
Consider the diagram
\begin{equation}\label{17}\begin{CD} P\setminus X @<\ell << Y\setminus Z \\
@VVkV @VVk''V \\
P @<i<< Y
\end{CD}
\end{equation}
The lemma will follow if we show $i^* k_*\QQ_{P\setminus X} \xrightarrow{\cong} k_*''\ell^*\QQ_{P\setminus X}$ in \eqref{17}. Since $i^*$ is left-adjoint to $i_*$, the existence of such a map is equivalent to the existence of a map
\begin{equation}k_* \to i_*k_*''\ell^* = k_*\ell_*\ell^*.
\end{equation}
It is enough to define a map from the identity functor to $\ell_*\ell^*$. But again by adjunction, this is the same as a map $\ell^* \to \ell^*$. Here we can take the identity. 

Arguing as before, the problem is now local and we can work in a small ball $B=B_X\times B_Y$. The local picture with the notation of the previous lemma is
\begin{equation}\begin{CD}(B_X\setminus X')\times B_Y @<<< (B_X\setminus X')\times Y' \\
@VVV @VVV \\
B_X\times B_Y @<<< B_X\times Y'.
\end{CD}
\end{equation}
Again the assertion is clear by Kunneth. 
\end{proof}

\begin{lem}\label{dualitylemma} Let notation and assumptions be as above, and write $n=\dim P$. Assume $P$ is smooth and projective. Then we have a perfect pairing
\begin{equation}\label{Lpair} H^*(P\setminus Y, X\setminus Z;\QQ(n)) \times H^{2n-*}(P\setminus X,Y\setminus Z;\QQ) \to \QQ.
\end{equation}
Said another way, we have
\begin{equation}\label{homology}H^*(P\setminus Y, X\setminus Z;\QQ(n)) \cong H_{2n-*}(P\setminus X,Y\setminus Z;\QQ).
\end{equation}
\end{lem}
\begin{proof}From the previous lemma applied twice we are reduced to showing
\begin{equation}H^*(P, j_!k_*'\QQ)(-n) \cong H^{2n-*}(P, k_!j_*'\QQ)^\vee. 
\end{equation} 
The Verdier duality functor $\DD$ is a contravariant functor on the derived category of sheaves on $P$ such that the sheaves $\sS$ and $\DD \sS$ are Poincar\'e dual, i.e. there is a perfect pairing $H^i(P, \sS) \times H^{-i}(P, \DD\sS) \to \QQ$. We have $\DD\QQ = \QQ[-2n](n)$, and $\DD$ intertwines lower shriek and lower star. Thus
\begin{multline}H^{2n-*}(P, k_!j_*'\QQ)^\vee = H^{-*}(P, k_!j_*'\DD\QQ)(-n) = \\
H^{-*}(P, \DD(k_*j_!'\QQ))(-n) = H^*(P,k_*j_!'\QQ)(-n) = H^*(P,j_!k_*'\QQ)(-n).
\end{multline}
\end{proof}

\begin{rem}\label{RGremark}
In the analytic context, one way of representing the factors of (\ref{Lpair}) is in terms of topological cycles (using (\ref{homology}) and its analogue for the other factor).  For the left-hand factor, these must avoid $X$ (red) but are allowed to bound on $Y$ (green); whereas for the right-hand factor, red and green are swapped.
\end{rem}

\subsection{Applications: CY periods}\label{CYper} Take $n\ge 2$ and assume (various generalizations are possible) that $\pi: P \to \PP^n$ is a toric variety obtained by a sequence of blowups. Let $X\subset P$ be the strict transform of a hypersurface of degree $n+1, X_0 \subset \PP^n$. Let $Y_0 \subset \PP^n$ be the coordinate simplex $Y_0:\prod_0^{n}T_i=0$ where the $T_i$ are homogeneous coordinates, and let $Y=\pi^{-1}Y_0$. We assume that $X$ is smooth, and $Y\cup X$ is a normal crossings divisor. Let $Z=X\cap Y$. Note that $P\setminus Y\cong \PP^n\setminus Y_0 \cong \GG_m^n$. The exact sequence of relative cohomology yields
\begin{multline}H^{n-1}(\GG_m^n,\QQ(n)) \to H^{n-1}(X\setminus Z,\QQ(n)) \to \\
H^n(P\setminus Y,X\setminus Z;\QQ(n)) \to H^n(\GG_m^n,\QQ(n)) \to 0.
\end{multline}
This can be rewritten (the superscript $\ \widetilde{}\ $ indicating we take the quotient modulo the image of $H^{n-1}(\GG_m^n,\QQ(n))$) 
\begin{multline}\label{sequence} 0 \to H^{n-1}(X\setminus Z,\QQ(n))\ \widetilde{}\   \buildrel{\alpha}\over{\to} \\
 H^n(P\setminus Y,X\setminus Z;\QQ(n)) \to \QQ(0) \to 0.
\end{multline}

Assume further that the topological chain given by $T_i\ge 0, 0\le i\le n$, lifts to a chain $\sigma$ on $P$ with $\partial\sigma \subset Y$ and $\sigma\cap X=\emptyset$.\footnote{One can check for our family of $K3$-surfaces that blowing up the vertices and then the faces of dimension $1$ in $\PP^3$ suffices to achieve $\sigma\cap X=\emptyset$.}  Then $\sigma$ represents a class in $H_n(P\backslash X,Y\backslash Z;\QQ)$ which maps to $1\in \QQ(0) = H_n(P\backslash Y,\QQ)$. Via equation \eqref{homology} above, we can interpret $\sigma \in H^n(P\backslash Y,X\backslash Z;\QQ(n))$ as a splitting of \eqref{sequence} as an exact sequence of $\mathbb{Q}$-vector spaces. The extension class of \eqref{sequence} in the ext group of mixed Hodge structures 
\begin{equation}Ext^1_{MHS}(\QQ(0), H^{n-1}(X\backslash Z,\QQ(n))\ \widetilde{}\ ) \cong H^{n-1}(X\backslash Z,\CC(n)/\QQ(n))\ \widetilde{}
\end{equation}
is computed as follows. By \cite{DeligneHII} corollaire 3.2.15(ii) it follows that $$F^0H^{n-1}(X\backslash Z,\CC(n))\ \widetilde{}\ =(0).
$$  
As a consequence, one has $F^0H^n(P\backslash Y,X\backslash Z;\CC(n)) \cong \CC(0)$, so there is a unique $s_F \in F^0H^n(P\backslash Y,X\backslash Z;\CC(n))$ lifting $1$.  So the class of the extension \eqref{sequence} is given by 
$$\varepsilon \in H^{n-1}(X\backslash Z,\CC(n)/\QQ(n))\ \widetilde{}\, ,
$$ 
where $\varepsilon$ is the unique class with $\alpha(\varepsilon)=\sigma - s_F$.

By assumption, $X_0$ is an anti-canonical hypersurface in $\PP^n$. Let
$\Omega_0\neq 0$ be a global $n$-form on $\PP^n$ with a pole of order
$1$ along $X_0$ and no other singularities. Assume further the
pullback $\Omega:=\pi^*\Omega_0$ has a pole along the strict transform
$X$ of $X_0$ and no other singularities, so $\Omega$ represents a
class in $F^nH^n(P\backslash X,\CC)$. We have $H^n(Y\backslash
Z,\CC)=(0)$ by cohomological dimension, and $F^nH^{n-1}(Y\backslash
Z,\CC)=(0)$ by \cite{DeligneHII} corollaire 3.2.15(ii), so the exact
sequence of relative cohomology yields an isomorphism
$F^nH^n(P\backslash X,\CC)\cong F^nH^n(P\backslash X,Y\backslash
Z;\CC)$. Thus, $\Omega$ lifts canonically to $\Omega \in
F^nH^n(P\backslash X,Y\backslash Z;\CC)$. We have a perfect pairing of
mixed Hodge structures by lemma
\ref{dualitylemma}\footnote{We refer to the beginning of
  section~\ref{subsec reinterp of I} for the definition of $\langle\ ,\ \rangle'$.}
\begin{equation}\langle\ ,\ \rangle': H^n(P\backslash X,Y\backslash Z;\QQ) \otimes H^n(P\backslash Y,X\backslash Z;\QQ(n)) \to \QQ(0)
\end{equation} 
In particular, the element $\langle\Omega,s_F\rangle' \in F^n\CC(0) = (0)$. We have proven
\begin{prop}\label{FeynmanPairing} With notation as above, the pairing
  of $\Omega$ with the extension class \eqref{sequence} is given up to
  (relative) periods $$ \left\{\left.\int_{\Gamma}\Omega \,\right| \,\Gamma\in \text{image}\left\{H_{n-1}(X,Z;\QQ)\buildrel{Tube}\over{\to} H_n(P\backslash X,Y\backslash Z;\mathbb{Q})\right\} \right\}$$  by the integral of $\Omega$ over the chain $\sigma$:
\begin{equation}\label{regulatorpairing} \langle\Omega, \sigma-s_F\rangle' = \int_\sigma \Omega .
\end{equation} 
Alternatively, with $\omega:=Res_X(\Omega)$, we have $$\langle\omega,\varepsilon\rangle' \equiv \frac{1}{2\pi i}\int_{\sigma}\Omega$$ modulo relative periods $\int_{\gamma}\omega$, $\gamma \in H_{n-1}(X,Z;\mathbb{Q})$.
\end{prop}
To relate the above to the Abel-Jacobi viewpoint for Milnor symbols explained in section \ref{subsec AJ rev}, one can use Deligne cohomology $H^p_\mathcal D(V,\ZZ(q))$ for any quasi-projective variety $V$ over $\CC$, \cite{EsnaultViehweg}. There is a functorial cycle class map $CH^a(V,b) \xrightarrow{[\ ]} H_\mathcal D^{2a-b}(V, \ZZ(a))$.
One has the universal Milnor symbol in degree $n$ which represents a class $sym_n \in CH^n(\GG_m^n, n)$. In our situation, one has $X\backslash Z\inj P\backslash Y = \GG_m^n$. Consider the diagram
\begin{equation}\label{delignecoh} \begin{CD}CH^n(\GG_m^n, n) @>>> CH^n(X\backslash Z, n) @>>> CH^n(\GG_m^n,X\backslash Z;n-1) \\
@VVV @VVV @VVV \\
H^n_\mathcal D(\GG_m^n,\ZZ(n)) @>>> H^n_\mathcal D(X\backslash Z,\ZZ(n)) @>>> H^{n+1}_\mathcal D(\GG_m^n, X\backslash Z;\ZZ(n))
\end{CD}
\end{equation}
Deligne cohomology fits into an exact sequence
\begin{equation} 0 \to \text{Ext}^1_{HS}(\QQ(0), H^{n-1}_{Betti}(V, \ZZ(r))) \to H^n_{\mathcal D}(V, \ZZ(r)) \to H^n_{Betti}(V, \ZZ(r))
\end{equation}
By cohomological dimension, we have 
$$H^{n+1}_{Betti}(\GG_m^n,\ZZ) = (0)=H^{n}_{Betti}(X\backslash Z,\ZZ),
$$ 
so the bottom line in \eqref{delignecoh} can be written
\begin{multline}H^n_\mathcal D(\GG_m^n,\ZZ(n)) \xrightarrow{a} \text{Ext}^1_{HS}(\QQ(0), H^{n-1}_{Betti}(X\backslash Z, \ZZ(n))) \\
 \to \text{Ext}^1_{HS}(\QQ(0), H^{n}_{Betti}(\GG_m^n,X\backslash Z; \ZZ(n))) \label{delignecoh2}
\end{multline}
Consider the diagram with top row the extension given by $a[sym_n]$ in \eqref{delignecoh2}. \minCDarrowwidth.1cm \begin{small}
\begin{equation}\label{milnorsymbolext}\begin{CD} 0 @>>> H^{n-1}(X\backslash Z, \QQ(n)) @>>> M @>>> \QQ(0) @>>> 0 \\
@. @VVV @VV b V @| \\
 0 @>>> H^{n-1}(X\backslash Z, \QQ(n))\ \widetilde{}\ @>> > H^{n}(\GG_m^n,X\backslash Z; \QQ(n)) @>>> \QQ(0) @>>> 0
\end{CD}
\end{equation}
\end{small}
It follows from \eqref{delignecoh2} that there exists an arrow $b$ as
indicated. This means that up to rational scale, the Milnor symbol
extension coincides with the extension \eqref{sequence}.  Note that
this does not recover Theorem \ref{thm Feynman =00003D HNF}. Indeed,
quite generally, the ambiguity is given by periods 
  the  $\int_c\omega$ where $c$ represents a class in
  $H_{n-1}(X,Z;\QQ)$. In our situation, where we have a family $X_t$
  of $K3$-surfaces, the resulting multi-valued function of $t$ does
  not satisfy the inhomogeneous Picard-Fuchs equation because the
  local system with fibres $H^2(X_t\backslash Z_t)$ is larger than the local system $H^2(X_t)$. For us, the ``extra'' periods have the form $\int_{c_t}\omega_t$ where $c_t$ is a $2$-disc on $X_t$ with boundary on $Z_t$. Since $Z_t$ is a union of rational curves, such periods are associated to motivic cohomology classes in $H^3_M(X_t, \QQ(2))$. For more detail on these interesting periods, see \cite{Kerr} and the references cited there.

%%%%%%%%%%%%%%%%%%%%%%%%%%%%%%%%%%%%%%%%%%%%%%%%%%%%%%%%%%%%%%%%%%
\section{Special values of the integral}
\label{subsec hnf analysis}

As promised in $\S$\ref{subsec reinterp of I}, we present some
consequences for special values of the identification of the Feynman
integral as a higher normal function (Theorem \ref{thm Feynman =00003D
  HNF}), by evaluating the three-banana integral at the special values $t=1$ and $t=0$.

\subsection{Special value at $t=1$}
\label{sec:toneBis}

It has been  conjectured in~\cite{Bailey:2008ib,BroadhurstLetter,BroadhurstProc}
that the value at $t=1$ of the three-banana integral is given by an $L$-function
value
\begin{equation}
  \label{e:Lvalue}
  I_{\ba}(1)= {12\pi i\over \sqrt{-15}} L(f^+,2)
\end{equation}
where $L(f^+,s) =\sum_{n\geq1 } a_n/n^s$ is the $L$-function
associated to the weight-three conductor 15 modular form
\begin{equation}\label{Lfunct} f^+(\tau)=\eta(\tau)\eta(3\tau)\eta(5\tau)\eta(15\tau)\sum_{m,n\in\mathbb Z} q^{m^2+mn+4n^2}=\sum_{n\geq1} a_nq^n
\end{equation}
constructed in~\cite{peterstop} . 

We will show that \ref{e:Lvalue} holds up to a rational coefficient using a triviality result, theorem \ref{thm hnf sv} below, for the trace of a certain $\ZZ/5\ZZ$-action on the Milnor symbol. The proof invokes Deligne's conjecture \cite{DeligneLvalue} for critical values of $L$-functions. In this case, the $L$-function in question \eqref{Lfunct} is a Hecke $L$-series associated to an algebraic Hecke character, and Deligne's conjecture was proven by Blasius \cite{Blasius}. The specific application we will use of their work is the following
\begin{prop}\label{deligneconj} Let $\omega_1 \in \Gamma(X_1, \Omega^2)$ be the algebraic differential form over $\QQ$, \eqref{eqn **23}. \newline\noindent(i) Let $0\neq c\in H_2(X_1, \QQ)_{tr}$ be a $2$-cycle. Then $L(f^+,2) \in \QQ(\sqrt{-15})\cdot\int_c\omega_1$.  \newline\noindent(ii) Let $0\neq x\in H^2(X_1, \QQ(2))_{tr}$ be a Betti cohomology class. Then $ L(f^+,2) \in\QQ(\sqrt{-15})\cdot \langle x,\omega_1\rangle'$. Here $\langle x,\omega_1\rangle'$ is the Poincar\'e duality pairing. 
\end{prop}
\begin{proof} Note that (i) and (ii) are equivalent because $H_2(X_1,
  \QQ) \cong H^2(X_1, \QQ(2))$, an isomorphism of Hodge structures
  which is compatible with the pairings with $H^2$.  (To see that the
  $L$-function is critical at $s=2$ the reader can consult~\cite[\S2]{hardersch}.)
The usual formulation of Deligne's conjecture would say that if $x$ in (ii) is invariant under the real conjugation, then $L(f^+,2) \in \QQ\cdot \langle x, \omega_1\rangle'$. However, in this case we have complex multiplication by $\QQ(\sqrt{-15})$, i.e. $H^2(X_1,\QQ)_{tr}$ is a rank one $\QQ(\sqrt{-15})$-vector space, so changing $x$ multiplies the pairing by an element in the CM field. 
\end{proof}
%-----------------------------------------------------------------------------
\subsubsection{Special fiber at $t=1$}
\label{sec:special-fiber-at}

Recall from $\S$\ref{subsec Verrill fam} that countably many fibers $X_t$ in the $K3$ family have Picard number 20, and hence are of CM type.  That $X_1$ is one of these CM fibers is shown in \cite{peterstop} (so that $H^2_{tr}(X_1)$ is a CM Hodge structure).  What makes $X_1$ special amongst the CM fibers is an \emph{additional} symmetry property which arises as follows.

Consider $\PP^{4}$ with homogeneous coordinates $T_{0},\ldots,T_{4}$,
hyperplane $H=\left\{ \sum_{i=0}^{4}T_{i}=0\right\} $, and hypersurface
$Y=\left\{ \sum_{i=0}^{4}\prod_{j\neq i}T_{j}=0\right\} $. Then $X_{1}$
is a resolution of singularities of $H\cap Y$, which can be seen
by writing $U_{i}:=T_{i}|_{H}$ ($i=0,\ldots,4$) and $x_{i}:=\frac{U_{i}}{U_{0}}$
($i=1,2,3$). Since $Y$ and $H$ are stable under the permutation
action of the symmetric group $\mathfrak{S}_{5}$ on the $\{T_{i}\}$,
it is clear that $\mathfrak{S}_{5}$ acts on $H\cap Y$ hence birationally
on $X_{1}$.  Let $\omega_1\in \Omega^2(X_1)$ be as in (\ref{eqn **23}). Since we may express $\omega_{1}$ as 
\[
Res_{X}Res_{H}\left(\frac{\sum_{i=0}^{4}(-1)^{i}T_{i}dT_{0}\wedge\cdots\wedge\widehat{dT_{i}}\wedge\cdots\wedge dT_{4}}{\left(\sum_{i}\prod_{j\neq i}T_{j}\right)\left(\sum T_{i}\right)}\right)\in\Omega^{2}(H\cap Y),
\]
the action of  $\mathfrak{S}_5$ on $\CC\omega_{1}$ hence $H_{tr}^{2}(X_{1})\,(\subsetneq H_{var}^{2}(X_{1}))$
is through the alternating representation.

%We know from the work~\cite{peterstop}  that $X_{1}$ has Picard rank $20$, so that $H_{tr}^{2}(X_{1})$ is a rank $2$
%CM Hodge structure.

\subsubsection{The higher normal function analysis}
\label{sec:high-norm-funct}

\begin{thm}\label{thm hnf sv}
 $I_{\ba}(1)$ is a  $(2\pi i)^3$ times a period of  
$$
\omega_1:=Res_{X_1}\left(
    \frac{\frac{dx}{x}\wedge \frac{dy}{y}\wedge
      \frac{dz}{z}}{1-(1-x-y-z)(1-x^{-1}-y^{-1}-z^{-1})} \right)\,.
$$
\end{thm}
\begin{proof}
Let $\sigma:X_{1}\to X_{1}$ be the automorphism induced by the
cyclic permutation $T_{0}\mapsto T_{1}\mapsto\cdots\mapsto T_{4}\mapsto T_{0}$
of the $\{T_{i}\}$. Write $\hat{\Xi}_{1}:=\sum_{j=0}^{4}(\sigma^{j})^{*}\Xi_{1}\in H^3_M(X_1,\QQ(3)).$
Since $\sigma_{*}\tilde{\omega}_1=\tilde{\omega}_1$, we have
\[
5V_{\ba}(1)=5\langle\mathcal{R}_{1},\tilde{\omega}_1\rangle
\]
\[
=\sum_{j=0}^{4}\langle\mathcal{R}_{1},(\sigma^{j})_{*}\tilde{\omega}_1\rangle=\langle\sum_{j=0}^{4}(\sigma^{j})^{*}\mathcal{R}_{1},\tilde{\omega}_1\rangle,
\]
where the cohomology class $\sum_{j=0}^{4}(\sigma^{j})^{*}\mathcal{R}_{1}\in H_{var}^{2}(X_{1},\CC)$
gives a lift of $\overline{AJ}_{X_{1}}^{3,3}(\hat{\Xi}_{1})\in H_{var}^{2}(X_{1},\CC/\QQ(3)).$
To show that $V_{\ba}(1)$ is a $\QQ(3)$-period, it will suffice
to establish that the image of the latter in $H_{tr}^{2}(X_{1},\CC/\QQ(3))$
is zero.

Let $U\subset X_{1}$ be any Zariski open set, $Y=X\backslash U$.
In the commutative diagram\footnote{Note:  $H^3_{M,Y}(X,\QQ(3))\cong CH^2(Y,3)_{\QQ}$.} \[\xymatrix{
H^3_{M,Y}(X,\QQ(3)) \ar [d]^{AJ_Y} \ar [r] & H^3_M(X_1,\QQ(3)) \ar [d]^{AJ_{X_1}} \ar [r] & H^3_M(U,\QQ(3)) \ar [d]^{AJ_U}
\\
H^2_Y(X_1,\CC/\QQ(3)) \ar [r] & H^2(X_1,\CC/\QQ(3)) \ar [r]^{\nu} & H^2(U,\CC/\QQ(3)),
}\]the image of $\nu$ factors the projection from $H^{2}(X_{1})$
to $H_{tr}^{2}(X_{1}).$ This reduces the problem to checking that
the image $\hat{\Xi}_{1}|_{\eta_{X_{1}}}$ of $\hat{\Xi}_{1}$ in
\[
\underset{U}{\underrightarrow{\lim}}\, H^3_M(U,\QQ(3)) \cong K_{3}^{M}(\CC(X_{1}))\otimes\QQ
\]
 is zero.

This is now a simple computation in Milnor $K$-theory (written additively).
Working modulo (2-)torsion, we have
\[
\xi:=\{x,y,z\}=\left\{ x,\frac{y}{x},\frac{z}{x}\right\} ,
\]
\[
\sigma^{*}\xi=\left\{ \frac{y}{x},\frac{z}{x},\frac{1+x+y+z}{x}\right\} =\left\{ \frac{1+x+y+z}{x},\frac{y}{x},\frac{z}{x}\right\} ,
\]
\[
(\sigma^{2})^{*}\xi=\left\{ \frac{z}{y},\frac{1+x+y+z}{y},\frac{1}{y}\right\} =-\left\{ 1+x+y+z,y,z\right\} ,
\]
\[
(\sigma^{3})^{*}\xi=\left\{ \frac{1+x+y+z}{z},\frac{1}{z},\frac{x}{z}\right\} =-\left\{ 1+x+y+z,\frac{1}{x},z\right\} ,
\]
\[
(\sigma^{4})^{*}\xi=\left\{ \frac{1}{1+x+y+z},\frac{x}{1+x+y+z},\frac{y}{1+x+y+z}\right\} 
\]
\[
=-\left\{ 1+x+y+z,\frac{y}{x},\frac{1}{x}\right\} .
\]
Now observe that 
\[
\xi+\sigma^{*}\xi=\left\{ 1+x+y+z,\frac{y}{x},\frac{z}{x}\right\} 
\]
and
\[
(\sigma^{2})^{*}\xi+(\sigma^{3})^{*}\xi+(\sigma^{4})^{*}\xi=-\left\{ 1+x+y+z,\frac{y}{x},\frac{z}{x}\right\} ,
\]
so that $\hat{\Xi}_{1}|_{\eta_{X_{1}}}=\sum_{j=0}^{4}(\sigma^{j})^{*}\xi=0.$\end{proof}

%------------------------------------------------------------------
\subsubsection{Value at $t=1$}
\label{sec:formula}
The proof for Broadhurst's formula \eqref{e:Lvalue} up to a rational coefficient is now straightforward. By theorem \ref{thm hnf sv}, the regulator class in $H^2(X_1, \CC/\QQ(3))_{tr}$ is trivial, which implies that the lifting $\cR$ of this class to $H^2(X_1,\CC)_{tr}$ lies in 
$H^2(X_1,\QQ(3))_{tr}= 2\pi i\cdot H^2(X_1,\QQ(2))_{tr}$. Thus,
\begin{equation}I_\ba(1) = \langle \cR,\omega_1\rangle' \in 2\pi i\langle H^2(X_1, \QQ(2)), \omega_1\rangle' = \QQ(\sqrt{-15})\cdot 2\pi i L(f^+,2).
\end{equation}
The identity on the right follows from proposition \ref{deligneconj}.

%-----------------------------------------------------------------------------
\subsection{Special value at $t=0$}
 It has been showed
in~\cite{Bailey:2008ib,BroadhurstProc} that $I_{\ba}(0)=
7\zeta(3)$. We provide this section a derivation of this result form
the point of view of higher normal functions.

\begin{thm}\label{thm hnf sv0}
 $I_{\ba}(0)=7\zeta(3)$.
\end{thm}
\begin{proof}
The fiber $X_{0}$ (after semistable reduction) has the two
components $Y_{1}$ resp. $Y_{2}$ arising from $1-x-y-z=0$ resp.
$1-x^{-1}-y^{-1}-z^{-1}=0$, and six arising from the semistable reduction process
which we may ignore since $R_{\{x,y,z\}}$ is zero there. The motivic
cohomology formalism tells us to compute the pairing 
\[
V_{\ba}(0)=\langle[R_{\{x,y,z\}}],[\tilde{\omega}_{0}]\rangle=\sum_{i=1}^{2}\int_{Y_{i}}R_{\{x,y,z\}}\wedge\tilde{\omega}_{0}
\]
of a cohomology and homology class. 

Observing that $Y_{1}\cap Y_{2}$ is essentially the ``triangle''
$\{(x,y,1-x-y)\,|\,(1-x)(1-y)(x+y)=0\}$, let $\gamma=\gamma_{1}+\gamma_{2}+\gamma_{3}$
be a generator of $H_{1}(Y_{1}\cap Y_{2},\mathbb{Z})$. Also let $\beta=\beta_{1}+\beta_{2}$
be a 2-cycle on $X_{0}$ with $\partial\beta_{1}=\gamma=-\partial\beta_{2}$,
and where $(x,y,z)\mapsto(x^{-1},y^{-1},z^{-1})$ sends $\beta_{1}\mapsto\beta_{2}$.
We have in $H_{2}(X_{0},\QQ)$ (really in $H_{var}^{2}$, i.e.
working modulo classes in the limit of the fixed part) that $[\tilde{\omega}_0]=\frac{1}{2}\beta$.
The $\frac{1}{2}$ is obtained by computing
\[
Res_{x=1}Res_{y=1}Res_{z=1-x-y}\frac{\frac{dx}{x}\wedge\frac{dy}{y}\wedge\frac{dz}{z}}{\phi_{\ba}(x,y,z)}
\]
\[
=Res_{x=1}Res_{y=1}\frac{\frac{dx}{x}\wedge\frac{dy}{y}}{\left(1-x^{-1}-y^{-1}-\frac{1}{1-x-y}\right) (x+y-1)}
\]
\[
=Res_{x=1}Res_{y=1}\frac{dx\wedge dy}{(1-x)(1-y)(x+y)}=\frac{1}{2},
\]
which is a period of $\tilde{\omega}_0$ over a vanishing cycle $\alpha\in H^{2}(X_0)$
with $\langle\alpha,\beta\rangle=1$.

It remains to compute
\[
\frac{1}{2}\sum_{i=1}^{2}\int_{\beta_{i}}R_{\{x,y,z\}}=\int_{\beta_{1}}R_{\{x,y,z\}}
\]
\[
=\int_{\beta_{1}}\log(x)\frac{dx}{x}\wedge\frac{dy}{y}=\int_{\beta_{1}}\frac{\log(x)}{y(1-x-y)}dx\wedge dy
\]
\[
=\int_{\beta_{1}}d\left\{ \frac{\log\left(\frac{1-x-y}{y}\right)\log(x)}{1-x}dx\right\} =\int_{\gamma_{1}+\gamma_{2}+\gamma_{3}}\frac{\log\left(\frac{1-x-y}{y}\right)\log(x)}{1-x}dx
\]
\[
=2\int_{-1}^{1}\frac{\log(-x)\log(x)}{1-x}dx\,.
\]
This integral is readily evaluated as follows:
\[
2\int_{-1}^{1}\frac{\log(-x)\log(x)}{1-x}dx\equiv 4\int_{-1}^{1}\log(1-x)\log(x)\frac{dx}{x}
\mod \QQ(3)
\]
\[
\equiv -4\sum_{k\geq1}\frac{1}{k}\int_{-1}^{1}\log(x)x^{k-1}dx\mod \QQ(3)
\]
\[
\equiv  8\sum_{\tiny\begin{array}{c}
k\geq1\\
\text{odd}
\end{array}}\frac{1}{k^{3}}\equiv 7\zeta(3) \mod \QQ(3)\,.
\]
Now $I_{\ba}(0)$ is obviously real, so we can ignore the $\mathbb{Q}(3)$
ambiguity.
\end{proof}
\begin{rem}
Alternatively we can give a very different proof of Theorem \ref{thm hnf sv0} using the Eisenstein analysis of $\S$\ref{sec Eis symb}. Referring to the proof of Theorem \ref{thm:main}, we have $$I_{\ba} = V_{\ba} = \varpi_1(\tau)\cdot V_{\vf_{\ba}}(\tau).$$  Applying Props. 9.2 and 9.4 of \cite{DoranKerr} (the former suitably modified for the cusp $[0]$), we have that $$V_{\vf_{\ba}}(\tau)\sim -\frac{\tau^2}{6} L((\pi_0)_*\hat{\vf}_{\ba},3) = 7\cdot3\cdot2^4\zeta(3)\tau^2$$ as $\tau\to 0$.  For the other factor, the property $\eta(-1/\tau)=\sqrt{\tau}\eta(\tau)$ of Dedekind eta allows us to to pull back $\varpi_1(\tau)=\frac{(\eta(6\tau)\eta(2\tau))^4}{(\eta(3\tau)\eta(\tau))^2}$ under $\mu_6:\tau\mapsto -1/6\tau =:\tilde{\tau}.$  Namely, we have $$\varpi_1(\tau)=\varpi_1(-1/6\tilde{\tau})=-\frac{3}{4}\tilde{\tau}^2 H_{\ba}(\tilde{\tau})\varpi_1(\tilde{\tau}) \sim \frac{3}{4}\tilde{\tau}^2 = \frac{1}{2^4 3 \tau^2}$$ as $\tau\to 0$.  Taking the product (and noting the correspondence $\tau=0 \leftrightarrow t=0$) gives $I_{\ba}(0)=7\zeta(3)$.
\end{rem}

%%%%%%%%%%%%%%%%%%%%%%%%%%%%%%%%%%%%%%%%%%%%%%%%%%%%%%%%%%%%%%%%%%%%%
\appendix\section{Higher symmetric powers of the sunset motive}
\label{sec:higher}

\numberwithin{equation}{section}

In this section we consider the higher symmetric powers for the sunset
regulator. This leads immediately to generalization of the Eichler
integral found for the two-loop sunset (cf.~\cite{Bloch:2013tra} and
\S\ref{sec:sunset})  and three-banana (cf.~\S\ref{subsec Verrill fam}) Feynman integrals. It remains to be seen whether this
has any relevance for the higher loop banana integrals studied in~\cite[\S9]{PVstringmath}. 

\medskip
Consider the series
\begin{equation}\label{eq1}
\sum_{a\neq 0}{}_e\,\,\frac{\psi(a,b)}{a^{n-1}(a\tau+b)}\quad \text{Eisenstein summation, } n=3,4
\end{equation}
 (Here $\psi: (\Z/N\Z)^2 \to \CC$ is some map.)

Let $A$ be a finite dimensional $\QQ$-vector space, and let $A^\vee:=
\text{Hom}(A,\QQ)$ be the dual. There is a natural embedding $A^\vee
\inj \text{Der}(\text{Sym}(A))$ identifying $A^\vee$ with the
translation invariant derivations of $\text{Sym}(A)$, the symmetric algebra. (For example, if $a_i$ is a basis of $A$, the dual basis elements $a_i^\vee$ are identified with $\frac{\partial}{\partial a_i}$.) This leads to perfect pairings
\begin{equation}\label{eq2}
\langle , \rangle: \text{Sym}^n(A^\vee)\otimes \text{Sym}^n(A) \to \QQ;\quad \langle D^I,a^J\rangle = D^I(a^J)|_0
\end{equation} 
Notice, however, that because of factorials, this pairing is not perfect integrally. (The integral dual of $\text{Sym}$ is the divided power algebra.) 

Let $B:=\Z\ve_1\oplus \Z\ve_2$. Identify $B \cong B^\vee$ via the pairing $\langle\ve_1,\ve_2\rangle = -\langle\ve_2,\ve_1\rangle = 1$. 
With the above identification we find
\begin{equation}\label{eq3}
\langle\ve_1^{i_1}\ve_2^{i_2},\ve_1^{j_1}\ve_2^{j_2}\rangle = \begin{cases}(-1)^{i_2}i_1!i_2! & i_k=j_{1-k} \\
0 & \text{else} \end{cases}
\end{equation}

We now compute
\begin{eqnarray}\label{eq4}
&&\langle (\tau\ve_1+\ve_2)^{n-2},\int_{\tau}^{i\infty}\frac{(x\ve_1+\ve_2)^{n-2}d\tau}{(ax+b)^n}\rangle \cr
&=&(n-2)!\sum_{k=0}^{n-2}  \left( n\atop k\right)\, (-\tau)^{n-2-k} \int_\tau^{i\infty} {dx\,
  x^k\over (ax+b)^{n}}\cr
&=& (n-2)! \int_\tau^{i\infty} {(x-\tau)^{n-2}\over
  (ax+b)^{n}}\,dx\cr
&=& {(n-2)! \over (n-1) a^{n-1} (ax+b)}
\end{eqnarray}
Notice the left-hand-side is exactly the pairing we would expect to
compute for $\text{Sym}^{n-2}H^1(\sE_t)$, where $\sE_t$ is the sunset
elliptic curve, while the right-hand-side when Eisenstein summed over $a,b$ yields the corresponding function \eqref{eq1}.

%%%%%%%%%%%%%%%%%%%%%%%%%%%%%%%%%%%%%%%%%%%%%%%%%%%%%%%%%%%%%%%%%%%%%


\begin{thebibliography}{Bl1}



\bibitem[ABW]{Adams:2013nia}
L.~Adams, C.~Bogner, and S.~Weinzierl,
\emph{The Two-Loop Sunrise Graph with Arbitrary Masses,}
J. Math. Phys. {\bf 54}, 052303 (2013)
[arXiv:1302.7004 [hep-ph]].
%%CITATION = ARXIV:1302.7004;%%

\bibitem[ABW2]{Adams:2014vja}
---------, 
\emph{The Two-Loop Sunrise Graph with Arbitrary Masses in Terms of Elliptic Dilogarithms,}
arXiv:1405.5640 [hep-ph].
%%CITATION = ARXIV:1405.5640;%%



\bibitem[BBDG]{Bailey:2008ib}
D.~H.~Bailey, J.~M.~Borwein, D.~Broadhurst and M.~L.~Glasser, \emph{Elliptic Integral Evaluations of Bessel Moments,}
arXiv:0801.0891 [hep-th].
%%CITATION = ARXIV:0801.0891;%%

\bibitem[Bat1]{Batyrev}V. Batyrev, \emph{Dual polyhedra and mirror symmetry for Calabi-Yau hypersurfaces in toric varieties}, J. Algebraic Geom. 3(1994), no. 3, 493-535.

\bibitem[Bat2]{Batyrev2}---------, \emph{Variations of the mixed Hodge structure of affine hypersurfaces in algebraic tori},
Duke Math. J. 69 (1993), no. 2, 349-409.

\bibitem[Beil]{Beilinson}A. Beilinson, \emph{Higher regulators of modular
curves}, in ``Applications of algebraic $K$-theory to Algebraic
Geometry and Number Theory (Boulder, CO, 1983)'', Contemp. Math.
55, AMS, Providence, RI, 1986, 1-34.



\bibitem[BDK]{Bern:1996je}
Z.~Bern, L.~J.~Dixon and D.~A.~Kosower,
\emph{Progress in One Loop QCD Computations,}
Ann.\ Rev.\ Nucl.\ Part.\ Sci.\ {\bf 46} (1996) 109
[hep-ph/9602280].
%%CITATION = HEP-PH/9602280;%%


\bibitem[BL]{BL1}A.~Beilinson and A.~Levin, \emph{The Elliptic
    Polylogarithm,} in Motives (ed. Jannsen, U., Kleiman, S,. Serre,
    J.-P.), Proc. Symp. Pure Math. vol 55, Amer. Math. Soc., (1994),
    Part 2, 123-190.


\bibitem[Ber]{Bertin}M.-J. Bertin, \emph{Mahler's measure and $L$-series
of  K3 hypersurfaces,} in ``Mirror Symmetry V'' (Lewis, Yau,
Yui, Eds.), AMS/IP Stud. Adv. Math. 38, 2006, 3-18.

\bibitem[B]{Blasius} D. Blasius, \emph{On the Critical Values of Hecke
  L-Series,} Annals of Mathematics
Second Series, Vol. 124, No. 1 (Jul., 1986), pp. 23-63



\bibitem[Blo1]{Bloch1986} S. Bloch, \emph{Algebraic cycles and the
Beilinson conjectures}, Contemp. Math. 58 (1), 1986, 65-79.

\bibitem[Blo2]{Bloch1994}---------, \emph{The moving lemma for higher
Chow groups}, J. Algebraic Geom. 3 (1993), no. 3, 537-568.

\bibitem[BEK]{Bloch:2005bh}
S.~Bloch, H.~Esnault and D.~Kreimer,
\emph{On Motives Associated to Graph Polynomials,}
Commun.\ Math.\ Phys.\ {\bf 267} (2006) 181
[math/0510011 [math.AG]].
%%CITATION = MATH/0510011;%%

\bibitem[BV]{Bloch:2013tra}S.~Bloch and P.~Vanhove,
\emph{The Elliptic Dilogarithm for the Sunset Graph,}
arXiv:1309.5865 [hep-th].
%%CITATION = ARXIV:1309.5865;%%

\bibitem[BS]{BorweinSalvy} J.M.~Borwein and B.~Salvy, \emph{A Proof of a
  Recursion for Bessel Moments,} 
Experimental Mathematics
Volume 17, Issue 2, 2008, [arXiv:0706.1409 [cs.SC]]



\bibitem[Bri]{Britto:2010xq}
R.~Britto,
\emph{Loop Amplitudes in Gauge Theories: Modern Analytic Approaches,}
J.\ Phys.\ A {\bf 44} (2011) 454006
[arXiv:1012.4493 [hep-th]].
%%CITATION = ARXIV:1012.4493;%%



\bibitem[Broad1]{BroadhurstLetter}  D.~Broadhurst,  \emph{Schwinger's banana numbers
  and $L$-series,}  letter (July 2011).

\bibitem[Broad2]{BroadhurstProc} ---------, \emph{Multiple Zeta Values and Modular
  Forms in Quantum Field Theory,} in ``Computer Algebra
in Quantum Field
Theory''  p33-72, ed Carsten Schneider Johannes Bl\"umlein (Springer) 2013


\bibitem[CHH]{Caron-Huot:2014lda}
S.~Caron-Huot and J.~M.~Henn,
\emph{Iterative Structure of Finite Loop Integrals,}
JHEP {\bf 1406} (2014) 114
[arXiv:1404.2922 [hep-th]].
%%CITATION = ARXIV:1404.2922;%%



\bibitem[D]{DeligneHII}P. Deligne, \emph{Th\'eorie de Hodge II}, Publications Math\'ematiques de l'IH\'ES, 40 (1971), p. 5-57. 

\bibitem[D2]{DeligneLvalue}---------, \emph{Valeurs de Fonctions $L$ et P\'eriodes d'Int\'egrales}, Proc. Symp. Pure Math. Vol. 33(1979), part 2, pp. 313-346. 


\bibitem[DS]{DeningerScholl}C. Deninger and A. Scholl, \emph{The
Beilinson conjectures}, in ``$L$-functions and arithmetic (Durham,
1989)'', London Math. Soc. Lect. Note Ser. 153, Cambridge Univ. Press,
Cambridge, 1991, 173-209.

\bibitem[Do]{Dolgachev}I. Dolgachev, \emph{Mirror symmetry for lattice
polarized K3 surfaces}, J. Math. Sci., 81(3):2599\textendash{}2630,
1996.

\bibitem[DK]{DoranKerr}C. Doran and M. Kerr, \emph{Algebraic K-theory
of toric hypersurfaces}, CNTP 5 (2011), no. 2, 397-600.



\bibitem[EKMZ]{Ellis:2011cr}
R.~K.~Ellis, Z.~Kunszt, K.~Melnikov and G.~Zanderighi,
\emph{One-Loop Calculations in Quantum Field Theory: from Feynman Diagrams to Unitarity Cuts,}
Phys.\ Rept.\ {\bf 518} (2012) 141
[arXiv:1105.4319 [hep-ph]].
%%CITATION = ARXIV:1105.4319;%%

\bibitem[EH]{Elvang:2013cua}
H.~Elvang and Y.~-t.~Huang, \emph{Scattering Amplitudes,}
arXiv:1308.1697 [hep-th].
%%CITATION = ARXIV:1308.1697;%%


\bibitem[EV]{EsnaultViehweg}H.~Esnault and E.~Viehweg, \emph{
    Deligne-Beilinson cohomology} in  ``Beilinson's conjectures on special values of L-functions,'' 43-91, Perspect. Math., 4, Academic Press, Boston, MA, 1988. 


\bibitem[G]{Griffiths}P. A. Griffiths, \emph{On the periods of certain rational integrals. I, II}, Ann. of Math. (2) 90 (1969), 460-495; ibid. (2) 90 (1969) 496-541.

\bibitem[G2]{Griffiths2} ---------, \emph{A theorem concerning the differential equations satisfied by normal functions associated to algebraic cycles,} Amer. J. Math. 101 (1979) 94-131.

\bibitem[Gu]{Gunning}R. Gunning, \emph{Lectures on modular forms,}
Annals of Math. Stud. 48, Princeton Univ. Press, 1962.

\bibitem[HS]{hardersch} G. Harder and N. Schappacher, \emph{Special Values
  of Hecke $L$-Functions and Abelian Integrals}, Lecture Notes in
  Math. 1111, Springer, Berlin 1985, pp. 17-49.


\bibitem[H]{Henn:2013pwa}
J.~M.~Henn,
\emph{Multiloop Integrals in Dimensional Regularization Made Simple,}
Phys.\ Rev.\ Lett.\ {\bf 110} (2013) 25, 251601
[arXiv:1304.1806 [hep-th]].
%%CITATION = ARXIV:1304.1806;%%


\bibitem[IZ]{Itzykson:1980rh}
  C.~Itzykson and J.~B.~Zuber,
  \emph{Quantum Field Theory,}
  New York, Usa: Mcgraw-hill (1980) 705 P.(International Series In Pure and Applied Physics)

\bibitem[K]{Kerr} M. Kerr, \emph{$K_1^{ind}$ of elliptically fibered $K3$ surfaces:  a tale of two cycles}, in ''Arithmetic and Geometry of $K3$ Surfaces and Calabi-Yau Threefolds'' (Laza, Sch\"utt, and Yui Eds.), Fields Inst. Comm. 67, Springer, New York, 2013, pp. 387-409.

\bibitem[KL]{KerrLewis}M. Kerr and J. Lewis, \emph{The Abel-Jacobi
map for higher Chow groups, II,} Invent. Math. 170 (2007), 355-420.

\bibitem[KLM]{KerrLewisSMS}M. Kerr, J. Lewis, and S. M\"uller-Stach,
\emph{The Abel-Jacobi map for higher Chow groups,} Compos. Math. 142
(2006), no. 2, 374-396.


\bibitem[LR]{Laporta:2004rb}
S.~Laporta and E.~Remiddi,
\emph{Analytic Treatment of the Two Loop Equal Mass Sunrise Graph,}
Nucl.\ Phys.\ B {\bf 704} (2005) 349
[hep-ph/0406160].
%%CITATION = HEP-PH/0406160;%%


\bibitem[L]{L} A.~Levin, \emph{Elliptic polylogarithms: an analytic
    theory,} Compositio Math. 106 (1997), no. 3, 267-282.

\bibitem[Mo]{Morrison}D. Morrison, \emph{On $K3$ surfaces with large
Picard number}, Invent. Math. 75 (1984), 105-121.

\bibitem[MW]{MorrisonWalcher}D. Morrison, J. Walcher, \emph{$D$-branes and normal functions,} Adv. Theor. Math. Phys. 13 (2009), pp. 553-598.

\bibitem[MSWZ]{MullerStach:2011ru}
S.~M\"uller-Stach, S.~Weinzierl and R.~Zayadeh,
\emph{A Second-Order Differential Equation for the Two-Loop Sunrise
  Graph with Arbitrary Masses,} Commun. Num. Theor. Phys. {\bf 6}, 203
(2012), [arXiv:1112.4360 [hep-ph]].
%%CITATION = ARXIV:1112.4360;%%

\bibitem[MSWZ2]{MullerStach:2012mp}
---------,  \emph{Picard-Fuchs Equations for Feynman Integrals,}
Commun.\ Math.\ Phys.\ {\bf 326} (2014) 237
[arXiv:1212.4389 [hep-ph]].
%%CITATION = ARXIV:1212.4389;%%


\bibitem[O]{ouvry} St.~Ouvry, \emph{Random Aharonov-Bohm vortices and some exactly solvable families of integrals,} Journal of Statistical Mechanics: Theory and Experiment, 1 (2005), P09004, [arXiv:cond-mat/0502366].


\bibitem[PTV]{peterstop} C.~Peters,  J.~Top and M. van der Vlugt , \emph{The Hasse zeta
  function of a K3 surface related to the number of words of weight 5
  in the Melas codes,}  J. reine angew. Math.  {\bf 432}  (1992) 151-176.

 \bibitem[P]{Poincare}H. Poincar\'e, \emph{Sur les courbes trac\'ees sur les surfaces alg\'ebriques,} Ann. Sci. de l'Ecole Norm. Sup. 27 (1910), 55-108.

\bibitem[Sage]{sage}
W. A. Stein et al. Sage Mathematics Software (Version 6.2),
   The Sage Development Team, 2014, {\tt http://www.sagemath.org}.


\bibitem[Sc]{Schmid}W. Schmid, \emph{Variation of Hodge structure:
the singularities of the period mapping}, Invent. Math. 22 (1973),
211-319.

\bibitem[Sh]{Shokurov}S. Shokurov, \emph{Holomorphic forms of highest degree on Kuga's modular varieties.} (Russian)
Mat. Sb. (N.S.) 101 (143) (1976), no. 1, 131-157, 160.

\bibitem[Va]{PVstringmath} P.~Vanhove,
\emph{The Physics and the Mixed Hodge Structure of Feynman Integrals,}
Proc.\ Symp.\ Pure Math.\ {\bf 88} (2014) 161
[arXiv:1401.6438 [hep-th]].
%% CITATION = ARXIV:1401.6438;%%



\bibitem[Ve]{Verrill}H. Verrill, \emph{Root lattices and pencils
of varieties}, J. Math. Kyoto Univ. 36 (2) (1996), 423-446.


\bibitem[W]{Weil}
A. Weil, \emph{Elliptic Functions according to Eisenstein and Kronecker,} Ergebnisse der Mathematik und ihrer Grenzgebiete 88, Springer-Verlag, Berlin, Heidelberg, New York, (1976). 

\bibitem[Z]{ZagierElliptic} D. Zagier, \emph{The Bloch-Wigner-Ramakrishnan
  polylogarithm function,} Math. Ann. {\bf 286} 613-624 (1990).

\end{thebibliography}
\end{document}